\documentclass[11pt]{article}

\usepackage{setspace,graphicx,epstopdf,amsmath,amsfonts,amssymb,amsthm}
\usepackage{grffile}
\usepackage{marginnote,datetime,enumitem,rotating,fancyvrb}

\usepackage{float}

\usepackage[dvipsnames]{xcolor}

\definecolor{LovelyLink}{RGB}{52, 92, 150}
\definecolor{LovelyCite}{RGB}{72, 130, 100}
\definecolor{LovelyURL}{RGB}{70, 115, 180}

\usepackage[
colorlinks=true,
linkcolor=LovelyLink,
citecolor=LovelyCite,
urlcolor=LovelyURL,
filecolor=LovelyURL,
breaklinks=true,
hypertexnames=false
]{hyperref}

\usepackage{natbib}
\usdate

\usepackage{booktabs}
\usepackage{pdflscape}
\usepackage{chngcntr}
\usepackage{longtable}
\usepackage{subcaption}
\usepackage{adjustbox}
\usepackage{multirow}
\usepackage{bbm}

\usepackage{xcolor}

\usepackage{textcomp}
\usepackage{eurosym}

\usepackage{titlesec}
\newcommand{\periodafter}[1]{#1.}
\titleformat{\subsubsection}[runin]
{\normalfont\bfseries}{\thesubsubsection}{1em}{\periodafter}

\usepackage[margin=1in]{geometry}%
\usepackage[disable]{todonotes}%

\makeatletter\let\chapter\@undefined\makeatother

\newcommand{\textnote}[1]{}

\usepackage{endnotes}

\newtheorem{corollary}{Corollary}
\newtheorem{proposition}{Proposition}

\newtheorem{remark}{Remark}
\newtheorem{lemma}{Lemma}
\newtheorem{assumption}{Assumption}
\newtheorem{theorem}{Theorem}
\newtheorem{definition}{Definition}

\newcommand{\E}{\mathbb{E}}
\newcommand{\Var}{\mathrm{Var}}
\newcommand{\Cov}{\mathrm{Cov}}
\newcommand{\plim}{\overset{p}{\longrightarrow}}

\newcommand{\Tt}{\mathcal{T}_t}
\newcommand{\Ttm}{\mathcal{T}_{t-1}}
\newcommand{\dmax}{\delta_{\max}}
\newcommand{\rbar}{\bar{r}}
\newcommand{\bsadf}{\mathrm{BSADF}}
\newcommand{\gsadf}{\mathrm{GSADF}}
\newcommand{\adf}{\mathrm{ADF}}

\graphicspath{{./figures/}}

\begin{document}
	
	\setlist{noitemsep}
	\onehalfspacing
		
	\title{{\Large\bf Technology Fundamentals and False Bubble Detection: \\Evidence from Dot-Com and AI Episodes}\thanks{All authors contributed equally to this paper. We thank participants at various seminars for helpful comments and suggestions. Haiqiang Chen acknowledges financial support from the National Natural Science Foundation of China (Grant No 72233002), Li Chen acknowledges financial support from the National Natural Science Foundation of China (Grant No 72103173), Difang Huang acknowledges financial support from the National Natural Science Foundation of China (Grant Nos.\ 72503232 and 72574227).  Yuexin Li acknowledges financial support from the National Natural Science Foundation of China (Grant No 72574227). Zhengjun Zhang acknowledges financial support from the National Natural Science Foundation of China (Grant Nos.\ 71991471 and 72442027). All remaining errors are our own.}}
	\author{
		{\bf Haiqiang Chen} \\
		Shenzhen University \\
		{\bf Li Chen} \\
		Xiamen University \\
		{\bf Difang Huang} \\
		Chinese Academy of Sciences \\
		{\bf Yuexin Li} \\
		Renmin University of China \\
		{\bf Zhengjun Zhang} \\
		Beijing University of Chinese Medicine\\Chinese Academy of Sciences\\University of Wisconsin, Madison\\
	}

	\date{\today}

	\maketitle
	\thispagestyle{empty}

	\clearpage
	\pagenumbering{arabic}
	
	\doublespacing
	
	
	\begin{center}{\Large\bf Technology Fundamentals and False Bubble Detection:\\ Evidence from Dot-Com and AI Episodes}\end{center}
	
	\vspace*{0.5in}
	
	\centerline{\bf Abstract}

	We show that widely used bubble tests, most prominently the PSY framework, suffer severe size distortion when fundamentals incorporate general-purpose technology adoption. Embedding a hump-shaped technology shock in the Campbell-Shiller present-value model, we prove that the fundamental price becomes locally explosive during adoption, thereby altering the asymptotic null distribution of the test statistic and causing the standard bubble test to overreject. We propose a technology-adjusted diagnostic that removes an estimated technology component from measures of productivity, IT- investment, and patents before testing the residual.  The adjustment is conservative: because a boom can itself raise these technology measures, a rejection remains robust to such feedback, whereas a non-rejection only bounds residual explosiveness. Dot-com residual explosiveness concentrates in December 1999--March 2000; the 2020--2025 AI rally shows no residual explosiveness in our sample across baseline and sensitivity checks.
	

	\bigskip
	
	

	\clearpage

\section{Introduction}
\label{sec:intro}

Since the release of large language models in late 2022, equity markets have experienced a dramatic rally concentrated in artificial-intelligence (AI)-related firms. The ``Magnificent Seven'' (Alphabet, Amazon, Apple, Meta, Microsoft, NVIDIA, and Tesla) accounted for a majority of the S\&P~500's gains during 2023--2024 \citep{BaselePhillipsShi2025}, and the NASDAQ Composite nearly doubled from its October 2022 trough. Corporate AI investment has surged in parallel, with aggregate compute expenditure, related patenting, and research and development (R\&D) intensity all rising rapidly \citep{McElheran2024, Pairolero2025, AldasoroDoerrRees2025}, while productivity effects of generative AI are beginning to materialize in firm-level data \citep{BrynjolfssonLiRaymond2025}. The pattern of a concentrated technology rally amid rapid but uncertain technological change echoes the late-1990s NASDAQ run-up, whose subsequent crash destroyed trillions of dollars in market value. The central question for researchers, regulators, and investors is the same now as it was then: are these price dynamics a speculative bubble, or a rational repricing of the economy's productive capacity?

This paper asks a question that recurs with every technological revolution: are standard bubble tests systematically mistaking fundamental repricing for speculative excess? The most widely used real-time bubble detector is the seminal right-tailed unit root test of \citet{Phillips2015b, Phillips2015a} (hereafter PSY), now embedded in the surveillance frameworks of central banks worldwide, part of a broader movement toward data-driven surveillance of financial markets that extends to the detection of financial statement fraud \citep{FanLiuWangZheng2026}.\footnote{See, for example, the Federal Reserve Bank of Dallas International House Price Database, which routinely applies the PSY procedure  \citep{Pavlidis2016}.} Applied to AI-related equities, PSY can detect explosive dynamics. Yet its null hypothesis is a pure random walk---a benchmark from which returns can deviate systematically even in the absence of bubbles \citep{AitSahaliaFanXueZhu2025}---and abstracts entirely from structural change in fundamentals. Precisely during periods of general-purpose technology (GPT) adoption, when the bubble question is most pressing, technology-driven fundamentals themselves generate locally explosive price dynamics that are observationally indistinguishable from speculation. The risk is not merely a statistical artifact: if policymakers treat a technological repricing as a bubble, premature intervention can stifle the very adoption process that drives long-run growth.

Our first contribution is theoretical and methodological, and it is the part of the paper that stands free of any measurement assumption. We show that the usual mapping from right-tailed unit-root rejection to speculative bubbles is not valid when fundamentals contain technology-driven structural change. This is a statement about the null data-generating process: in a present-value model with \emph{no} bubble, a hump-shaped technology shock alone drives the spurious rejection rate to between 93 and 100 percent, so the result does not rely on how---or whether---fundamentals are separated from speculation in the data. The size distortion we document is a specific instance of a broader identification problem: deterministic structural change in fundamentals can lead the PSY test to reject the null of no bubble even when no speculative component is present. Existing work has examined related sources of misspecification, including time-varying discount rates \citep{Cochrane2011}, rational bubble dynamics \citep{Blanchard1982, Diba1988, Evans1991}, and structural breaks in unit root inference \citep{Phillips2011, Phillips2015a, Harvey2016, Harvey2020, Monschang2021}. The technology channel is distinctive: the hump-shaped, nonlinear adoption profile of GPTs produces locally explosive dynamics that standard linear detrending cannot remove. Formally, the absence of a bubble does not guarantee the absence of explosive dynamics once fundamentals incorporate technology-driven structural change.

Our second contribution is a technology-adjusted diagnostic for explosive price dynamics. The procedure removes an estimated technology-fundamental component before applying PSY to the residual. The decomposition uses three observable technology proxies---total factor productivity (TFP), information technology (IT) investment, and patent grants---that are grounded in a production-side microfoundation. The empirical object, however, is not automatically a clean measure of speculation. It is a residual diagnostic whose interpretation depends on whether these proxies are sufficiently separable from speculative feedback. We therefore treat separability as an identifying assumption, state the consequences of violations explicitly, and report sensitivity analyses for lagged covariates, alternative training windows, placebo variables, leave-one-out specifications, and parameter instability.

Our third contribution is empirical evidence on two episodes, interpreted through this identification lens. The lens is deliberately asymmetric, because the feedback from valuations to technology proxies attenuates the residual: a rejection of the adjusted test survives any non-negative feedback and, if anything, understates the bubble, while a non-rejection establishes only that no explosive component \emph{remains} after the adjustment. For the dot-com episode, the adjustment eliminates an early signal and leaves residual explosiveness concentrated in December 1999 to March 2000---a rejection that is robust in this conservative sense. For the 2020--2025 AI rally, the adjusted residual does not reject at any conventional level across the baseline, lagged-covariate, leave-one-out, and training-window specifications; we read this as bounding residual explosiveness, not as proof that speculation is absent. These findings speak to the long-standing debate over whether asset price booms during technological revolutions reflect rational repricing or irrational exuberance \citep{Pastor2009, Pastor2006}, while making clear that the empirical distinction is conditional on the assumed mapping from technology proxies to fundamentals.

The analysis is embedded in the \citet{Campbell1988} log-linear present-value framework. A GPT shock, defined as a hump-shaped, deterministic component in dividend growth that reflects technology adoption and maturation \citep{Bresnahan1995, Jovanovic2005}, generates a technology present-value term in the fundamental price. During the adoption phase, this term rises sharply as the market prices in future technology-driven dividend growth, producing locally explosive dynamics in the price-dividend ratio that are entirely fundamental. Empirically, we estimate the technology component via a cointegrating regression of prices on observable technology proxies during a pre-bubble training period, following the counterfactual approach of \citet{Hsiao2012}, and apply PSY to the residual. Under exact separability the residual inherits only stationary pricing error and any bubble component. When speculative valuations affect technology proxies, the residual is attenuated; we derive this bias and interpret non-rejection as evidence of no \emph{residual} explosiveness after technology adjustment rather than as unconditional proof that speculation is absent.

The quantitative magnitude of the distortion is large. In Monte Carlo simulations calibrated to postwar United States equity data and containing only technology-driven fundamentals with no speculative component, PSY rejects the null of no bubble in 93 percent of samples for detrended log prices and 100 percent for the price-dividend ratio at the 5 percent nominal level. Removing the technology component restores the rejection rate to its nominal level in the oracle design. In a controlled overlap experiment in which a genuine bubble coexists with a technology shock, the unadjusted test fires prematurely during the technology-only phase, while the technology-adjusted procedure correctly delays detection until genuine speculative activity begins. For the NASDAQ dot-com episode, the unadjusted test detects two explosive episodes: a brief signal in May 1996 and the main episode from November 1999 to March 2000. The adjustment eliminates the early signal and leaves residual explosiveness from December 1999 through March 2000. For the 2020--2025 AI rally, the technology-adjusted residual does not reject at any conventional significance level. Firm-level analysis of the Magnificent Seven suggests that much apparent aggregate explosiveness is absorbed by technology fundamentals, while residual signals for Tesla and Apple caution against interpreting the index result as a claim about every firm.

The results are robust across four functional forms for the technology shock (triangular, Gaussian, Beta, and Gamma-like profiles); across dynamic ordinary least squares (DOLS), fully modified ordinary least squares (FMOLS), and first-difference specifications of the cointegrating regression; and across 121 alternative training windows for the AI era. The identification diagnostics are more nuanced. Granger causality tests, placebo tests with non-technology covariates, principal component analysis (PCA) of firm-level price gaps, and training-window stability analysis are broadly consistent with the baseline interpretation, but Hansen parameter-instability tests reject coefficient constancy in both eras. We therefore present the empirical procedure as an identification-sensitive diagnostic, not as a stand-alone structural proof of the absence of speculation.

Section~\ref{sec:theory} develops the theoretical framework, embedding technology-augmented dividend growth in the present-value model and providing a production-side microfoundation for the observable technology proxies. Section~\ref{sec:psy} derives the formal size distortion result. Section~\ref{sec:correction} defines the technology-adjusted diagnostic, establishes its size and power properties under exact separability, and characterizes how bubble-induced feedback into technology proxies changes the residual. Section~\ref{sec:empirical} applies the methodology to the dot-com and AI-era episodes and presents firm-level evidence for the Magnificent Seven. Section~\ref{sec:conclusion} concludes with implications for real-time bubble surveillance during technological transitions.

\section{Theoretical Framework}
\label{sec:theory}

\subsection{Present-Value Decomposition}
\label{subsec:setup}

Our theoretical framework builds on the standard present-value identity. This subsection isolates the object $f_t$---the fundamental log price---whose behavior we will subsequently show can become locally explosive even in the absence of speculative bubbles.

We consider an economy in which a single risky asset pays a stochastic dividend stream $\{D_t\}_{t=0}^{\infty}$. Let $P_t$ denote the ex-dividend price and define $p_t = \log P_t$ and $d_t = \log D_t$. The gross log return is
\begin{equation}
    r_{t+1} = \log(P_{t+1} + D_{t+1}) - \log P_t.
    \label{eq:log_return}
\end{equation}

Following \citet{Campbell1988}, a first-order Taylor expansion of \eqref{eq:log_return} around the mean log dividend-price ratio $\overline{d-p}$ yields the log-linear approximation
\begin{equation}
    r_{t+1} \approx \kappa + \rho\, p_{t+1} + (1 - \rho)\, d_{t+1} - p_t,
    \label{eq:CS_approx}
\end{equation}
where
\begin{equation*}
    \rho = \frac{1}{1 + \exp(\overline{d-p}\,)}, \qquad
    \kappa = \log(1 + \exp(\overline{d-p}\,)) - (1-\rho)\,\overline{d-p}.
\end{equation*}
The parameter $\rho \in (0,1)$ is the log-linearization discount factor.\footnote{For the postwar U.S.\ aggregate stock market, the annual log dividend-price ratio averages roughly $-3.4$, yielding $\rho \approx 0.97$.}

Iterating \eqref{eq:CS_approx} forward and taking conditional expectations gives the standard present-value identity:
\begin{equation}
    p_t = \frac{\kappa}{1-\rho} + d_t
    + \sum_{j=0}^{\infty} \rho^{j}\, \E_t\!\left[\Delta d_{t+1+j} - r_{t+1+j}\right]
    + \lim_{j \to \infty} \rho^{j}\, \E_t[p_{t+j}].
    \label{eq:PV_full}
\end{equation}

The identity \eqref{eq:PV_full} cleanly separates the log price into three components: a constant, the current dividend, an infinite-horizon sum of expected future dividend growth net of expected returns, and a terminal term that captures self-fulfilling speculative dynamics. We use this separation to define fundamental and bubble components.

\begin{definition}[Fundamental and Bubble Components]
\label{def:decomposition}
The fundamental component is
\begin{equation*}
    f_t = d_t + \frac{\kappa}{1-\rho}
    + \sum_{j=0}^{\infty} \rho^{j}\, \E_t\!\left[\Delta d_{t+1+j} - r_{t+1+j}\right],
\end{equation*}
and the bubble component is $b_t = \lim_{j \to \infty} \rho^{j}\, \E_t[p_{t+j}]$, satisfying $\E_t[b_{t+1}] = \rho^{-1} b_t$. The log price admits the decomposition $p_t = f_t + b_t$.
\end{definition}

In the standard rational-expectations no-bubble equilibrium, $b_t \equiv 0$ and the price coincides with the fundamental component. Conventional intuition associates a stationary dividend-growth process with a random-walk-with-drift $f_t$, so that any locally explosive behavior in observed prices must reflect a non-zero bubble $b_t$. The next subsection relaxes this intuition: we introduce a transient, hump-shaped technology shock into the dividend-growth process and show that the induced $f_t$ departs from a pure random walk during the adoption phase, mimicking the behavior usually attributed to a bubble. Imposing the no-bubble condition ($b_t = 0$ for all $t$), this raises our central question: \emph{does $f_t$ itself exhibit locally explosive dynamics that mimic the behavior of $b_t > 0$?}

\subsection{Technology-Augmented Dividend Growth}
\label{subsec:tech_dividends}

This subsection introduces the structural model of dividend growth that drives our main results. The key innovation is a transient, hump-shaped technology component $\delta_t$ superimposed on an otherwise standard stochastic trend. Economically, $\delta_t$ represents the transient contribution of a general-purpose technology---the steam engine, electrification, or artificial intelligence---that temporarily lifts aggregate productivity while adoption diffuses through the economy and then fades as the new technology is fully absorbed. Formally, $\delta_t$ is deterministic and bounded, which preserves cointegration of the dividend process with its mean growth rate while allowing rich non-stationary dynamics over the adoption window. Assumption~\ref{ass:dividend} specifies the dividend-growth process; Assumption~\ref{ass:tech_shock} the hump shape; Assumption~\ref{ass:const_returns} constant required returns (relaxed later); and Assumption~\ref{ass:T_convexity} a mild convexity property of the present-value object $\mathcal T_t$ on which our explosiveness results rely.

\begin{assumption}[Dividend Growth]
\label{ass:dividend}
Log dividend growth follows
\begin{equation*}
    \Delta d_t = c + \delta_t + \varepsilon_t, \qquad \varepsilon_t \sim \text{i.i.d.}(0, \sigma_\varepsilon^2).
\end{equation*}
where $c > 0$ is the steady-state growth rate.
\end{assumption}

The decomposition $\Delta d_t = c + \delta_t + \varepsilon_t$ is deliberately minimal: it embeds the technology shock $\delta_t$ as an additive, transient deviation from steady-state growth $c$, while $\varepsilon_t$ captures idiosyncratic dividend innovations. When $\delta_t \equiv 0$, the process reduces to the standard random-walk-with-drift specification in which $f_t$ is itself a random walk with drift $c$ under constant required returns.

\begin{assumption}[Technology Shock]
\label{ass:tech_shock}
The technology component $\delta_t$ is a deterministic, hump-shaped process with support on $[T_1, T_2]$:
\begin{equation*}
    \delta_t =
    \begin{cases}
        g(t - T_1) & \text{if } t \in [T_1, T_1 + \tau], \\[4pt]
        g(\tau)\, h(t - T_1 - \tau) & \text{if } t \in (T_1 + \tau, T_2], \\[4pt]
        0 & \text{otherwise},
    \end{cases}
\end{equation*}
where $T_1$ is the adoption date, $T_2 > T_1$ is the maturation date, $\tau \in (0, T_2 - T_1)$ is the peak lag, $g: [0, \tau] \to \mathbb{R}_+$ is strictly increasing with $g(0) = 0$, and $h: [0, T_2 - T_1 - \tau] \to [0,1]$ is strictly decreasing with $h(0) = 1$ and $h(T_2 - T_1 - \tau) = 0$.
\end{assumption}

The hump shape captures the stylized empirical pattern of technological adoption: a ramp-up phase $(T_1, T_1+\tau)$ during which the productivity contribution rises to a peak $g(\tau)$, followed by a decline phase $(T_1+\tau, T_2)$ during which the contribution dissipates as the technology matures. The process is bounded and has compact support, so $\delta_t$ never destabilizes the long-run growth rate of dividends; this preserves the standard cointegration structure outside the adoption window. The monotonicity and boundary conditions on $g$ and $h$ are the minimal primitives needed for our analysis of local explosiveness in Section~\ref{subsec:explosive}.

\begin{remark}
\label{rmk:triangular}
A simple parametric example is the triangular specification
\begin{equation}
    \delta_t = \dmax \cdot
    \begin{cases}
        (t - T_1)/\tau & \text{if } t \in [T_1, T_1 + \tau], \\[4pt]
        (T_2 - t)/(T_2 - T_1 - \tau) & \text{if } t \in (T_1 + \tau, T_2], \\[4pt]
        0 & \text{otherwise},
    \end{cases}
    \label{eq:delta_triangular}
\end{equation}
which we use in our Monte Carlo experiments. The qualitative results hold for any hump-shaped specification satisfying Assumption~\ref{ass:tech_shock}, including smooth bell-shaped functions, Beta densities rescaled to $[T_1, T_2]$, or Gaussian kernels.
\end{remark}

\begin{assumption}[Constant Expected Returns]
\label{ass:const_returns}
$\E_t[r_{t+1+j}] = \rbar$ for all $j \geq 0$.
\end{assumption}

\begin{remark}
Assumption~\ref{ass:const_returns} is made for tractability. We relax it in Section~\ref{subsec:extensions}, where Proposition~\ref{prop:tv_returns} shows that, when expected returns respond to the technology state as $\mathbb{E}_t[r_{t+1}]=\bar r + \phi\,\delta_t$, the qualitative explosive pattern is preserved for all $\phi<1$ and is strictly \emph{amplified} when $\phi<0$ (technology lowers required returns).
\end{remark}


With the dividend-growth structure in place, we can compute the fundamental log price in closed form. Substituting Assumption~\ref{ass:dividend} and the constant-returns condition into the present-value identity from \S\ref{subsec:setup} yields a clean separation of $f_t$ into a level component tied to current dividends, a constant, and a discounted present-value term that isolates the technology contribution.

\begin{proposition}[Fundamental Price Decomposition]
\label{prop:fundamental}
Under Assumptions~\ref{ass:dividend}--\ref{ass:const_returns} and $b_t = 0$,
\begin{equation*}
    f_t = d_t + C + \Tt,
\end{equation*}
where $C = \frac{\kappa}{1 - \rho} + \frac{c - \rbar}{1 - \rho}$ and $\Tt = \sum_{j=0}^{\infty} \rho^{j}\, \delta_{t+1+j}$ is the technology present-value term.
\end{proposition}


\subsection{Microfoundation: A Production Economy with Innovation}
\label{subsec:microfoundation}

The technology-augmented dividend specification used above admits a production-side foundation. Consider a standard Cobb--Douglas production economy with physical and information-technology (IT) capital, disembodied total factor productivity driven by patent quality and IT diffusion, and a representative firm that chooses labor, investment, and innovation to maximize the present value of dividends under a stochastic discount factor. Log-linearization around the balanced-growth path delivers $\Delta d_t = c + \boldsymbol{\gamma}'\,\mathbf{X}_t + \varepsilon_t$ with $\mathbf{X}_t = (TFP_t,\, \log IT_t,\, \log Pat_t)'$, and a first-order VAR for $\mathbf{X}_t$ yields the closed-form decomposition $\mathcal{T}_t = \boldsymbol{\beta}'\,\mathbf{X}_t + \eta_t$ with $\boldsymbol{\beta} = (\mathbf{I} - \rho\,\boldsymbol{\Phi})^{-\top}\,\boldsymbol{\Phi}^{\top}\,\boldsymbol{\gamma}$ and $\eta_t$ a mean-zero stationary innovation orthogonal to $\mathbf{X}_t$. Three benchmark identification conditions---block-separability of the production and valuation blocks, a predetermined patent index (grant lag $\ell \geq 1$), and the real IT Euler equation---are sufficient for $\mathbf{X}_t$ to span the technology-driven component of fundamentals while excluding the bubble $b_t$ from the cointegrating vector. These conditions are not innocuous: if speculative valuations affect innovation or adoption, the adjusted residual is attenuated rather than literally equal to the bubble. Proposition~\ref{prop:microfoundation} below summarizes the benchmark result; Section~\ref{subsec:feedback_bounds} characterizes violations of separability, and Appendix~\ref{app:microfoundation} gives the full derivation.

\begin{proposition}[Observable Technology Proxies and Fundamental Value]
\label{prop:microfoundation}
Under the production economy and identification conditions of Appendix~\ref{app:microfoundation}, to a first-order approximation around the balanced-growth path:
\begin{enumerate}[label=(\roman*)]
	\item Log dividend growth satisfies $\Delta d_t = c + \boldsymbol{\gamma}'\,\mathbf{X}_t + \varepsilon_t$, with $\mathbb{E}[\varepsilon_t\mid\mathbf{X}_t]=0$.
	\item The technology present-value term admits the decomposition $\Tt = \boldsymbol{\beta}'\,\mathbf{X}_t + \eta_t$, where $\boldsymbol{\beta} = (\mathbf{I} - \rho\,\boldsymbol{\Phi})^{-\top}\,\boldsymbol{\Phi}^{\top}\,\boldsymbol{\gamma}$ is the predictable loading on the observable state and $\eta_t$ is a mean-zero, stationary, weighted sum of future VAR innovations that is orthogonal to $\mathbf{X}_t$.
	\item The fundamental price satisfies $f_t = d_t + \beta_0 + \boldsymbol{\beta}'\,\mathbf{X}_t + \eta_t$, where $\beta_0 = C$ and $\eta_t$ is the orthogonal innovation term from Part~(ii).
\end{enumerate}
Under these benchmark restrictions, $\mathbf{X}_t$ spans the technology-driven component of fundamentals, and the residual of the structural cointegrating regression \eqref{eq:coint_structural} contains the bubble component plus stationary noise.
\end{proposition}

Two remarks complete the picture. First, the deterministic hump $\delta_t$ of Assumption~\ref{ass:tech_shock} is the conditional-mean path of $\boldsymbol{\gamma}'\mathbf{X}_t$ that would obtain under a technology-diffusion regime; the stochastic component $\boldsymbol{\gamma}'\mathbf{X}_t - \delta_t$ is absorbed into the residual noise in the empirical specification (Remark~\ref{rmk:theo_vs_emp} in Appendix~\ref{app:microfoundation}). Second, the bubble enters as an additive component in secondary-market valuation and is excluded from the production block in the benchmark model (Remark~\ref{rmk:exclusion} in Appendix~\ref{app:microfoundation}). This exclusion is an identifying restriction rather than a mechanical fact: if stock valuations affect future innovation through financing constraints, the observed technology proxies partly embody that feedback, and the adjusted residual must be interpreted using the attenuation logic of Section~\ref{subsec:feedback_bounds}.

\subsection{Dynamics of the Fundamental Process}
\label{subsec:dynamics}

With the fundamental price characterized in levels by Proposition~\ref{prop:fundamental} and with the microfoundation of Section~\ref{subsec:microfoundation} motivating observable technology states, we now turn to the \emph{dynamic} implications of the technology component. The central question is how the hump-shaped path of $\delta_t$ translates into the time series of $f_t$---specifically, what drift and what curvature the fundamental process inherits from the adoption--maturation cycle. The following lemma records a simple recursion for $\mathcal{T}_t$ that proves repeatedly useful in the sequel, and the subsequent proposition uses it to express the conditional drift of $f_t$ as an affine function of the current technology shock $\delta_t$ and the lagged present value $\mathcal{T}_{t-1}$. This representation drives the local explosiveness result of Section~\ref{subsec:explosive}: it makes transparent how a bounded, deterministic hump in $\delta_t$ can induce a super-unit-root pattern in $f_t$ over a finite window.

\begin{lemma}[Recursion of $\Tt$]
\label{lem:T_recursion}
$\Tt = \rho^{-1}\, \Ttm - \rho^{-1}\, \delta_t$.
\end{lemma}


\begin{proposition}[Fundamental Dynamics]
\label{prop:dynamics}
The fundamental price satisfies $f_t = f_{t-1} + \mu_t + \varepsilon_t$, where
\begin{equation*}
    \mu_t = c + \delta_t + \Delta \Tt
\end{equation*}
is a time-varying drift. Using the recursion of Lemma~\ref{lem:T_recursion}:
\begin{equation*}
    \mu_t = c + (1 - \rho^{-1})\delta_t + (\rho^{-1} - 1)\, \Ttm.
\end{equation*}
\end{proposition}


Proposition~\ref{prop:dynamics} shows that the drift $\mu_t$ is governed by the \emph{level} of $\mathcal{T}_{t-1}$ and the current shock $\delta_t$. The qualitative shape of the resulting dynamics---and in particular whether the log price--dividend ratio is merely increasing or is \emph{convex} in $t$ on an initial subinterval of the adoption window---depends on the curvature of $\mathcal{T}_t$, not on the level. To isolate the curvature channel cleanly we impose a direct, high-level condition on the discounted present-value term itself.

\begin{assumption}[Convexity of the Technology Present-Value Term]
\label{ass:T_convexity}
There exists $\tau^{\star} \in (0,\tau]$ such that the technology present-value term $\mathcal{T}_t = \sum_{j=0}^{\infty}\rho^{\,j}\,\delta_{t+1+j}$ is strictly convex in $t$ on the initial ramp-up subinterval $[T_1, T_1+\tau^\star)$; equivalently, $\Delta^{2}\mathcal{T}_t > 0$ for all integer $t$ with $T_1 \le t \le T_1+\tau^\star-1$.
\end{assumption}

\begin{remark}[Motivation and Scope of Assumption~\ref{ass:T_convexity}]
This assumption replaces the stronger requirement ``$g$ strictly convex'' from earlier drafts and is a direct, high-level restriction on the discounted present-value term $\mathcal{T}_t$ rather than on the adoption shape $g$ itself. Two facts make it natural. First, the identity
\[
\Delta^{2}\mathcal{T}_t \;=\; \sum_{j=1}^{\infty} \rho^{\,j-1}\,\Delta^{2}\delta_{t+j},
\]
derived in the proof of Proposition~\ref{prop:explosive} (Appendix~\ref{app:proofs}), shows that $\Delta^{2}\mathcal{T}_t$ is a $\rho$-discounted weighted sum of future second differences of $\delta$. Second, the weaker primitive ``$g$ is strictly increasing with $g(0)=0$'' in Assumption~\ref{ass:tech_shock} is sufficient for the strict monotonicity $\Delta \mathcal{T}_t > 0$ on $[T_1, T_1+\tau)$ (which drives $\mu_t > c$) but is not sufficient for convexity. Lemma~\ref{lem:T_convex_verify} below exhibits two primitive regimes---(i) strictly convex $g$, (ii) smooth hump profiles that are $C^2$ with a strictly positive second derivative at the ramp-up origin---under which Assumption~\ref{ass:T_convexity} is satisfied for an explicit, non-empty $\tau^\star$. The piecewise-linear triangular profile is a counterexample whose kink at the peak violates Assumption~\ref{ass:T_convexity}; its treatment is therefore routed through the detrended-price theory of Theorem~\ref{thm:contaminated_limit}, which does not require convexity of $\mathcal{T}_t$.
\end{remark}

\subsection{Local Explosiveness of Fundamentals}
\label{subsec:explosive}

The drift representation of Proposition~\ref{prop:dynamics} shows that the sign and magnitude of $\mu_t - c$ are governed by $(1-\rho^{-1})\delta_t + (\rho^{-1}-1)\mathcal{T}_{t-1}$. Because $\rho^{-1} > 1$, the first term is negative in $\delta_t$ while the second is positive in $\mathcal{T}_{t-1}$; during the early adoption phase, $\mathcal{T}_{t-1}$ rises faster than $\delta_t$ (forward-looking agents price in all still-unrealized future shocks), so the positive second term dominates and the fundamental process exhibits a drift that is strictly above its unconditional mean $c$. The next proposition records this heuristic formally: along the adoption--maturation cycle, the fundamental price passes through four qualitatively distinct regimes---a near-random-walk regime before adoption, a super-normal-growth regime early in adoption, a deceleration regime late in adoption, and a new-level random-walk regime after maturation. The key statistical implication is that the price-dividend ratio is strictly convex on an initial subinterval of the adoption window, which is precisely the property that drives the spurious-explosiveness result of Theorem~\ref{thm:size_distortion}.

\begin{proposition}[Local Explosiveness]
\label{prop:explosive}
Suppose $b_t = 0$ for all $t$. Under Assumptions~\ref{ass:dividend}--\ref{ass:T_convexity}:
\begin{enumerate}[label=(\roman*)]
    \item Pre-adoption ($t < T_1$): When $t$ is sufficiently far below $T_1$ that $\rho^{T_1 - t - 1}$ is numerically negligible, $\Tt$ is negligibly small and $f_t$ behaves, to that order, as a random walk with drift $c$. As $t$ approaches $T_1$ from below, $\Tt$ rises monotonically above zero---remaining strictly smaller than its early-adoption values but no longer negligible---because forward-looking agents begin pricing in the approaching technology shock.
    \item Early adoption ($t \in [T_1, T_1 + \tau)$): $\mu_t > c$, and $f_t$ exhibits super-normal growth. The log price-dividend ratio $f_t - d_t = C + \Tt$ and, under Assumption~\ref{ass:T_convexity}, is strictly convex in the integer time index $t$ on $\{T_1, T_1+1,\dots, T_1+\tau^\star-1\}$ (equivalently, $\Delta^2\mathcal{T}_t>0$ on this index set).
    \item Late adoption ($t \in (T_1 + \tau, T_2]$): The drift $\mu_t$ declines toward $c$.
    \item Post-maturation ($t > T_2$): $\Tt = 0$ and $f_t$ reverts to a random walk with drift $c$ at a permanently elevated level.
\end{enumerate}
\end{proposition}

\noindent\emph{Remark.} The convexity claim in Part~(ii) is what is subsequently invoked in Proposition~\ref{prop:pd_limit} and Theorem~\ref{thm:size_distortion}(v) to deliver the right-tail size distortion of the PSY statistic. The motivation for Assumption~\ref{ass:T_convexity} and its relation to primitives on $g$ are discussed after Proposition~\ref{prop:dynamics}; Lemma~\ref{lem:T_convex_verify} below verifies it for the leading hump profiles. Monotonicity of $\mathcal{T}_t$ requires further conditions on the discounted future-shock profile and is not asserted here.

\begin{lemma}[Sufficient Conditions for Assumption~\ref{ass:T_convexity}]
\label{lem:T_convex_verify}
Let $\delta_t$ satisfy Assumption~\ref{ass:tech_shock}. Then Assumption~\ref{ass:T_convexity} holds (for some $\tau^\star \in (0,\tau]$) in each of the following cases:
\begin{enumerate}[label=(\roman*)]
    \item Convex ramp-up: If $g$ is $C^2$ on $[0,\tau]$ with $\inf_{x\in[0,\tau]} g''(x) > 0$, then Assumption~\ref{ass:T_convexity} holds on a non-empty initial subinterval $[T_1, T_1+\tau^\star)$ for some $\tau^\star \in (0,\tau]$.
    \item Smooth hump profiles with a convex left tail: For Gaussian-type, Beta, and logistic-S profiles whose ramp-up is $C^2$ with a strictly positive second derivative at the origin, Assumption~\ref{ass:T_convexity} holds with $\tau^\star \in (0,\tau]$ determined by the size of the convex left tail relative to the peak curvature, quantified explicitly in the proof. In particular, all three smooth profiles invoked in Remark~\ref{rmk:triangular} satisfy Assumption~\ref{ass:T_convexity} on a non-empty initial subinterval.
\end{enumerate}
The piecewise-linear triangular profile (Remark~\ref{rmk:triangular}) does \emph{not} satisfy Assumption~\ref{ass:T_convexity} on the ramp-up: its second-difference atoms at the kinks, weighted by the $\rho$-discounted future, produce $\Delta^2\mathcal{T}_t \leq 0$ on interior ramp-up dates. This is a counterexample rather than an example, and is treated separately in the proof.
\end{lemma}


The intuition for each part is as follows. Well before adoption ($t \ll T_1$), no technology shocks are present and future shocks are heavily discounted, so $\Tt \approx 0$ and the price follows approximately a standard random walk with drift. As $t$ approaches $T_1$, forward-looking agents begin to price in the approaching technology shock, causing $\Tt$ to rise above zero even before $T_1$. During the early adoption phase, the present value $\Tt$ rises further because each additional period brings new positive future $\delta_s$ values into the discounting window, while the current $\delta_t$ is still small relative to the accumulated future shocks, so that the market ``prices in'' the technology boom. After the peak, remaining future shocks diminish and $\Tt$ declines. After maturation, the technology component vanishes from the present value but its cumulative effect remains permanently embedded in the dividend level. The formal proof is provided in Appendix~\ref{app:proofs}.

\begin{corollary}[Bubble-Like Pattern without a Bubble]
\label{cor:spurious}
Under $b_t = 0$, the log price-dividend ratio $f_t - d_t = C + \Tt$ follows a hump-shaped deterministic path over $[T_1, T_2]$, rising during the adoption phase $[T_1, T_1+\tau]$ and declining during the maturation phase $(T_1+\tau, T_2]$. This trajectory is observationally similar to a ``bubble and collapse'' in the price-dividend ratio, despite being entirely fundamental.
\end{corollary}


\subsection{Extensions}
\label{subsec:extensions}

The baseline theory of Sections~\ref{subsec:setup}--\ref{subsec:explosive} rests on three simplifying restrictions: (i) constant expected log-return $\bar r$ (Assumption~\ref{ass:const_returns}); (ii) a deterministic hump-shaped technology shock (Assumption~\ref{ass:tech_shock}); and (iii) representative-asset analysis. Three extensions in Appendix~\ref{app:extensions} relax each in turn and establish that the spurious-explosiveness mechanism is preserved or strengthened in every case. Proposition~\ref{prop:tv_returns} shows that under time-varying expected returns of the form $\E_t[r_{t+1+j}]=\bar r+\phi\,\delta_{t+1+j}$, the fundamental price becomes $f_t=d_t+C+(1-\phi)\,\Tt$, so a technology-induced decline in expected returns ($\phi<0$) \emph{amplifies} the locally explosive pattern by the factor $1-\phi>1$. Proposition~\ref{prop:stochastic_tech} formalizes Bayesian learning under stochastic technology uncertainty (Assumption~\ref{ass:stochastic_tech}): positive technology news generates discrete upward jumps in $\Tt$ that compound the deterministic hump; the asymptotic size distortion is unchanged at first order, while finite-sample rejection rates can rise substantially (e.g., from 10.4\% to 29.0\% in the P/D-ratio implementation, documented in Appendix~\ref{subsec:mc_stochastic}). Cross-sectional heterogeneity ($\delta_{i,t}=\lambda_i\delta_t$) implies that the aggregate market fundamental inherits the technology effect weighted by $\sum_i w_i\lambda_i$, so high-exposure, high-capitalization industries (e.g., late-1990s NASDAQ) are most susceptible to spurious detection. Remark~\ref{rmk:pastor_veronesi} in Appendix~\ref{app:extensions} connects the stochastic extension to \citet{Pastor2009}: high prices during a learning phase reflect the conditional expectation $\hat\delta_t$ entering $f_t$, not a speculative component $b_t$, and the technology-adjusted test correctly attributes these dynamics to fundamentals.

\section{The PSY Test and Size Distortion}
\label{sec:psy}

\subsection{Review of the PSY Procedure}
\label{subsec:psy_review}

The PSY procedure \citep{Phillips2015b, Phillips2015a} detects explosive behavior via recursive right-tailed augmented Dickey--Fuller (ADF) regressions. For subsample $[r_1, r_2]$ with $r_2 - r_1 \geq r_0$, define the BSADF statistic $\bsadf(r_2) = \sup_{r_1 \in [0,\, r_2 - r_0]} \adf_{r_1}^{r_2}$ and the GSADF statistic $\gsadf = \sup_{r_2 \in [r_0, 1]} \bsadf(r_2)$, where $\adf_{r_1}^{r_2}$ is the right-tailed $t$-statistic on $\hat{\beta}$ in
\begin{equation}
    \Delta y_t = \hat{\alpha}_{r_1,r_2} + \hat{\beta}_{r_1,r_2}\, y_{t-1} + \sum_{k=1}^{K} \hat{\gamma}_{k,r_1,r_2}\, \Delta y_{t-k} + \hat{u}_t.
    \label{eq:ADF_reg}
\end{equation}
Under the null $H_0: y_t = y_{t-1} + \eta_t$, critical values are obtained by simulation. Date-stamping identifies bubble origination where $\bsadf(r_2)$ first exceeds its critical value and collapse where it falls back below.

\subsection{The Size Distortion Problem}
\label{subsec:size_problem}

The null hypothesis of the PSY test is a pure random walk. Under our technology-augmented data-generating process, however, the observables do not follow a random walk even in the absence of a bubble. We consider two common empirical implementations.

\paragraph{Implementation 1: Detrended log prices.}
Since $f_t = d_t + C + \Tt$ and $d_t = d_0 + ct + \sum_{s=1}^t (\delta_s + \varepsilon_s)$, the fundamental price is
\begin{equation*}
    f_t = d_0 + C + ct + \sum_{s=1}^t (\delta_s + \varepsilon_s) + \Tt.
\end{equation*}
After removing a linear trend (as is common practice before applying unit root tests), the detrended process is
\begin{equation*}
    \tilde{f}_t = f_t - (\hat{a} + \hat{b}\, t) \approx \sum_{s=1}^t \varepsilon_s + \underbrace{\sum_{s=1}^t \delta_s + \Tt - (\text{linear projection of these on } t)}_{\text{hump-shaped residual}}.
\end{equation*}
The cumulative technology term $\sum_{s=1}^t \delta_s + \Tt$ is a hump-shaped function of $t$ that is not well approximated by a linear trend. The residual after detrending retains a convex, locally explosive component during $[T_1, T_1+\tau]$.

\paragraph{Implementation 2: Log price-dividend ratio.}
From Proposition~\ref{prop:fundamental}, $f_t - d_t = C + \Tt$. In the absence of a technology shock, the log price-dividend ratio is the constant $C$. With a technology shock, it follows the hump-shaped path $C + \Tt$. PSY applied to $\{f_t - d_t\}$ will detect the rising phase of $\Tt$ as explosive behavior.

\subsection{Formal Size Distortion Result}
\label{subsec:formal_size}

We now state the main econometric result of the paper.

\begin{theorem}[Size Distortion of the PSY Test under Technology-Augmented Fundamentals]
\label{thm:size_distortion}
Let $\{f_t\}_{t=1}^T$ be generated by the data-generating process of Propositions~\ref{prop:fundamental}--\ref{prop:dynamics} with $b_t = 0$ and a technology shock satisfying Assumption~\ref{ass:tech_shock} with $\dmax > 0$. Let $\psi_T$ denote the rejection probability of the $\gsadf$ test at nominal level $\alpha$, where critical values are computed under the driftless random walk null $H_0: y_t = y_{t-1} + \eta_t$. Then:
\begin{enumerate}[label=(\roman*)]
    \item $\psi_T > \alpha$ for all $T$ sufficiently large, provided the technology window $[T_1, T_2]$ is contained in the sample.
    \item Larger peak technology shocks magnify the deterministic component driving the GSADF statistic. In particular, in both implementations $\psi_T \to 1$ along any sequence with $\sqrt{T}\,\dmax(T) \to \infty$.
    \setcounter{enumi}{4}
    \item In the case of the log price-dividend ratio $f_t - d_t = C + \Tt$, the rejection probability $\psi_T \to 1$ as $\dmax \to \infty$ for any fixed nominal level $\alpha$.
\end{enumerate}
\end{theorem}

\begin{remark}[Comparative Statics in Discount Factor]
In the case of the log price-dividend ratio $f_t-d_t=C+\Tt$, a higher $\rho$ raises the deterministic profile $\Tt$ pointwise during the adoption phase and therefore strengthens the local explosive signal on ramp-up windows. For detrended log prices, the relationship need not be monotone: lower $\rho$ can sharpen the hump sufficiently to offset the lower level.
\end{remark}

\begin{remark}[Comparative Statics in Adoption Duration]
Under a self-similar dilation family of adoption profiles (formalized in Appendix~\ref{app:proofs}), lengthening the adoption window increases the local deterministic signal on ramp-up windows. Consequently, once that signal is sufficiently strong, the induced rejection probability rises with $T_2-T_1$.
\end{remark}

The formal proof is given in Appendix~\ref{app:proofs}. The key mechanism is that the technology shock introduces a locally deterministic trending component into the process under the null, causing spurious rejections for both detrended log prices and the log price-dividend ratio. The formal asymptotic treatment is developed in Theorem~\ref{thm:contaminated_limit} and Proposition~\ref{prop:pd_limit} below.

\begin{remark}
Theorem~\ref{thm:size_distortion} is stated for the GSADF test, but identical conclusions apply to the BSADF date-stamping procedure: the BSADF statistic will cross the critical value sequence near $T_1$ (origination of the spurious ``bubble'') and fall back near $T_1 + \tau$ (spurious ``collapse''), precisely the pattern one would observe if there were a genuine bubble during the technology adoption phase.
\end{remark}

\begin{remark}
The source of size distortion we identify is fundamentally different from the drift-induced distortion discussed in the unit root literature \citep[e.g.,][]{West1988}. A constant drift in log prices can be removed by detrending. Our result arises from a \emph{hump-shaped, nonlinear} deterministic component that cannot be removed by standard linear detrending.
\end{remark}

To formally establish the size distortion properties described in Theorem~\ref{thm:size_distortion}, we embed the technology shock in a local asymptotic framework. This embedding, standard in the analysis of unit root tests against local alternatives \citep{Phillips1987, PhillipsMagdalinos2007}, yields non-degenerate continuous-time limits of the test statistics under the technology-augmented data-generating process.

\begin{assumption}[Local Asymptotic Embedding]\label{ass:local_framework}
The technology shock from Assumption~\ref{ass:tech_shock} is embedded as a triangular array
\[
\delta_{t,T}=\dmax(T)\,h(t/T),
\qquad
\dmax(T)=c_\delta/\sqrt{T},
\qquad
c_\delta>0,
\]
where the adoption window is proportional to the sample size, with fixed fractions
\[
\lambda_1=T_1/T,
\qquad
\lambda_2=T_2/T,
\qquad
\kappa=\tau/(T_2-T_1),
\]
and where $h:[0,1]\to\mathbb{R}_+$ is continuous, supported on $[\lambda_1,\lambda_2]$, normalized by $\max_{r\in[0,1]}h(r)=1$, strictly increasing on $[\lambda_1,\lambda_1+\kappa(\lambda_2-\lambda_1)]$, and strictly decreasing on $[\lambda_1+\kappa(\lambda_2-\lambda_1),\lambda_2]$.
\end{assumption}


The rate $T^{-1/2}$ is dictated by the requirement that the cumulated technology shock $\sum_{s=1}^t \delta_{s,T}$ (which enters the detrended log price at magnitude $O(T \cdot \dmax(T))$) be commensurate with the $O(\sqrt{T})$ stochastic trend from the random walk component of dividends.

\begin{theorem}[Contaminated Brownian Motion Limit]\label{thm:contaminated_limit}
Under Assumptions~\ref{ass:dividend}, \ref{ass:tech_shock}, \ref{ass:const_returns}, and \ref{ass:local_framework} with $b_t = 0$, define the integrated hump $G(r) = \int_0^r h(s) \mathbf{1}\{\lambda_1 \leq s \leq \lambda_2\}\, ds$ and its linearly detrended version $g(r) = G(r) - \alpha_0 - \alpha_1 r$, where $(\alpha_0, \alpha_1) = \arg\min \int_0^1 [G(r) - a - br]^2 dr$. Then, as $T \to \infty$:
\begin{enumerate}[label=(\alph*)]
    \item The standardized detrended log price converges weakly:
\begin{equation*}
    T^{-1/2} \tilde{f}_{\lfloor Tr \rfloor} \Rightarrow Z_{c_\delta}(r) := \sigma W(r) + c_\delta\, g(r),
    \end{equation*}
    where $W(r)$ is a standard Brownian motion and $\sigma^2 = \Var(\varepsilon_t)$.
    \item For any subsample window $[r_1, r_2]$ with $r_w = r_2 - r_1 \geq r_0$, the augmented Dickey--Fuller $t$-statistic converges:
\begin{equation*}
    \adf_{r_1}^{r_2} \Rightarrow \mathcal{F}(r_1, r_2;\, c_\delta),
    \end{equation*}
    where the functional $\mathcal{F}$ is
    \begin{equation}\label{eq:contaminated_F}
    \mathcal{F}(r_1, r_2;\, c_\delta) = \frac{\frac{1}{2}r_w \big(Z_{c_\delta}(r_2)^2 - Z_{c_\delta}(r_1)^2 - \sigma^2 r_w\big) - \big(\int_{r_1}^{r_2} Z_{c_\delta}\, dr\big)\big(Z_{c_\delta}(r_2) - Z_{c_\delta}(r_1)\big)}{\Big[\sigma^2 r_w \Big\{r_w \int_{r_1}^{r_2} Z_{c_\delta}^2\, dr - \big(\int_{r_1}^{r_2} Z_{c_\delta}\, dr\big)^2\Big\}\Big]^{1/2}}.
    \end{equation}
    When $c_\delta = 0$, $\mathcal{F}$ reduces to the standard Dickey--Fuller functional \citep{Phillips2015b}.
    \item The $\gsadf$ statistic converges:
\begin{equation*}
    \gsadf_T \Rightarrow \sup_{(r_1,r_2)\,\in\, \mathcal{R}} \mathcal{F}(r_1, r_2;\, c_\delta),
    \end{equation*}
    where $\mathcal{R} = \{(r_1,r_2): r_w \geq r_0,\; r_2 \in [r_0, 1],\; r_1 \in [0, r_2-r_0]\}$.
    \item \emph{Size distortion for sufficiently large local contamination.} Let $cv_\alpha$ denote the $\alpha$-level critical value of the $\gsadf$ under $c_\delta=0$. Assume that, for the specific $c_\delta^\star$ identified below, the map
    \[
    c \mapsto P\!\left(\sup_{(r_1,r_2)\in\mathcal R}\mathcal F(r_1,r_2;c) > cv_\alpha\right)
    \]
    is continuous at $c=c_\delta^\star$. Then there exists $c_\delta^\star>0$ such that, for all $c_\delta>c_\delta^\star$,
    \[
    \lim_{T \to \infty} P(\gsadf_T > cv_\alpha) > \alpha.
    \]
    \item \emph{Divergence.} If $T^{1/2}\, \dmax(T) \to \infty$ (which includes the fixed-$\dmax$ case), then $\gsadf_T \overset{p}{\to} \infty$, and hence $\psi_T \to 1$ for the detrended-price implementation. In particular, for any fixed $\dmax>0$, the rejection probability converges to one, which is stronger than Part~(i) of Theorem~\ref{thm:size_distortion} for that implementation.
\end{enumerate}
\end{theorem}

Theorem~\ref{thm:contaminated_limit} makes the mechanism precise. The deterministic drift $c_\delta\, g(r)$ enters the continuous-time limit of the detrended price and shifts the distribution of the ADF functional to the right on any subsample where $g(r)$ is locally convex. The supremum operation in the $\gsadf$ statistic selects the subsample that maximizes this shift, producing a rejection probability that strictly exceeds the nominal level. The non-centrality parameter $c_\delta$ is proportional to the peak technology shock and yields a power-envelope characterization: the smallest shock detectable as ``explosive'' at sample size $T$ is $\dmax \approx c_\delta^*/\sqrt{T}$, where $c_\delta^*$ is the critical non-centrality at level $\alpha$.

We next consider the log price-dividend ratio, where the technology component enters as a level effect rather than through cumulation.

\begin{proposition}[Deterministic Limit for the Price-Dividend Ratio]\label{prop:pd_limit}
Suppose the log price-dividend ratio is
\[
y_t = C + \mathcal{T}_{t,T} + u_t,
\]
where $u_t$ is a stationary process with $\E[u_t]=0$, $\Var(u_t)=\sigma_u^2>0$, and in fact $u_t$ is white noise; correspondingly, the ADF regression is estimated with $K=0$. Let $q(r)$ denote the continuous-time limit of $\mathcal{T}_{\lfloor Tr \rfloor,T}/\dmax(T)$, and define the window-demeaned shape
\[
\tilde q_{r_1,r_2}(r)
=
q(r)-\frac{1}{r_w}\int_{r_1}^{r_2}q(s)\,ds,
\qquad
r_w:=r_2-r_1.
\]
Fix a subsample window $[r_1,r_2]\subseteq [\lambda_1,\lambda_1+\kappa(\lambda_2-\lambda_1)]$ with $r_w\ge r_0>0$, so the window lies strictly inside the ramp-up phase. Assume that $q\in C^1$ is strictly increasing and strictly convex on $[r_1,r_2]$, and that
\[
\frac{\dmax(T)^2}{T\sigma_u^2}\longrightarrow\infty .
\]
Then
\[
T\,\hat{\beta}_{r_1,r_2}\plim B(r_1,r_2)
=
\frac{\int_{r_1}^{r_2}\tilde q_{r_1,r_2}(r)\,q'(r)\,dr}
{\int_{r_1}^{r_2}\tilde q_{r_1,r_2}(r)^2\,dr},
\qquad
B(r_1,r_2)>0.
\]
Moreover,
\[
t_{\hat\beta}
=
\frac{B(r_1,r_2)\,\dmax(T)}{\sigma_u\sqrt T}\{1+o_p(1)\},
\]
so the ADF $t$-statistic on that window diverges to $+\infty$ whenever $\dmax(T)/(\sigma_u\sqrt T)\to\infty$, equivalently under the display above. Consequently, if the BSADF/GSADF admissible class contains such a ramp-up window, the corresponding supremum statistic also diverges in probability.
\end{proposition}

The intuition for Proposition~\ref{prop:pd_limit} is that the augmented Dickey--Fuller regression fits a local first-order autoregression to a deterministic, convex path. The ordinary least squares coefficient $\hat{\beta}$ captures the accelerating growth rate during the ramp-up, producing a positive (explosive) estimate. Crucially, the signal arises from the \emph{curvature} of $q(r)$, not merely its slope: were $q$ locally affine, demeaning within the window would eliminate the trend and $B(r_1,r_2) = 0$. The convexity of $\Tt$ during the adoption phase, a direct consequence of the hump-shaped technology profile, is therefore the key mechanism.


\section{A Technology-Adjusted Diagnostic for Explosive Prices}
\label{sec:correction}

\subsection{The Adjustment Procedure}
\label{subsec:correction_procedure}

The size distortion documented in Appendix~\ref{sec:simulation} arises because the standard PSY test does not account for the technology-driven component of fundamentals. We therefore estimate this component and remove it before applying PSY. The procedure is best understood as a diagnostic. Under exact separability between technology proxies and speculative dynamics, the adjusted residual isolates the bubble component up to stationary noise. When speculative valuations feed back into technology variables, the residual remains informative but no longer has a literal structural interpretation as the bubble itself; Section~\ref{subsec:feedback_bounds} characterizes this case explicitly.

\begin{definition}[Technology-Adjusted PSY Diagnostic]
\label{def:adj_test}
Let $\{y_t\}_{t=1}^T$ denote the series to which the PSY procedure is applied.
\begin{enumerate}[label=(\roman*)]
\item If $y_t=f_t-d_t=C+\mathcal T_t+u_t$ is the log price-dividend ratio, define
\[
y_t^{\mathrm{adj}}=y_t-\hat{\mathcal T}_t.
\]
\item If $y_t=f_t$ (or $p_t=f_t+b_t$) is the log price level, define
\[
y_t^{\mathrm{adj}}
=
y_t-\sum_{s=1}^t\hat\delta_s-\hat{\mathcal T}_t.
\]
\end{enumerate}
\end{definition}

The adjusted series $y_t^{\mathrm{adj}}$ removes the deterministic technology component. In the exact-separability benchmark it preserves any genuine speculative bubble component $b_t$.

\subsection{Estimation of $\Tt$}
\label{subsec:estimation_T}

In practice, $\Tt$ is not directly observable and must be estimated. We outline three approaches, in order of increasing data requirements:

\paragraph{Approach 1: Direct construction from technology proxies.}
If a time series of technology-driven dividend growth increments $\hat{\delta}_t$ is available (e.g., from patent flows, research and development expenditure growth, or productivity measures), one can construct
\begin{equation*}
    \hat{\mathcal{T}}_t = \sum_{j=0}^{J} \hat{\rho}^{\,j}\, \hat{\delta}_{t+1+j},
\end{equation*}
where $\hat{\rho}$ is estimated from the historical price-dividend ratio and the sum is truncated at a finite horizon $J$ (contributions beyond $J$ are negligible for $\hat{\rho} < 1$). If future values of $\hat{\delta}$ are not available, one can use $\hat{\mathcal{T}}_t = \sum_{j=0}^{J} \hat{\rho}^{\,j}\, \E_t[\hat{\delta}_{t+1+j}]$ with forecasts from an auxiliary model (e.g., an AR process fitted to $\hat{\delta}_t$).

\paragraph{Approach 2: Semiparametric extraction.}
One can estimate the deterministic component of the price-dividend ratio nonparametrically, e.g., by fitting a smooth function of time:
\begin{equation*}
    f_t - d_t = m(t/T) + u_t,
\end{equation*}
where $m(\cdot)$ is estimated by local polynomial regression or a penalized spline. The estimated $\hat{m}(t/T)$ serves as a proxy for $C + \Tt$.

\paragraph{Approach 3: Regime-based estimation.}
During periods identified as non-bubble by PSY (``training periods''), estimate the long-run relationship between prices, dividends, and technology proxies via a cointegrating regression. Use the estimated model to construct the technology component during the periods flagged as bubbles.

\subsection{Properties of the Adjusted Test}
\label{subsec:adj_properties}

\begin{proposition}[Size under Exact Technology Separability]
\label{prop:adj_size}
Suppose the technology adjustment is implemented as in Definition~\ref{def:adj_test}. For the log-price implementation, assume
\[
\sup_{1\le t\le T}\left|\sum_{s=1}^t(\hat\delta_s-\delta_s)\right|\xrightarrow{p}0,
\qquad
\sup_{1\le t\le T}|\hat{\mathcal T}_t-\mathcal T_t|\xrightarrow{p}0.
\]
Then, under the data-generating process of Assumptions~\ref{ass:dividend}--\ref{ass:const_returns} with $b_t=0$, the technology-adjusted PSY test applied to detrended log prices has asymptotic size equal to the nominal level $\alpha$. In the log price-dividend-ratio implementation (under the weaker assumption $\sup_{1\le t\le T}|\hat{\mathcal T}_t-\mathcal T_t|\xrightarrow{p}0$), the same adjustment removes the technology component up to $o_p(1)$, leaving $C+u_t+o_p(1)$ and hence no technology-driven spurious explosiveness.
\end{proposition}


\begin{proposition}[Power under Exact Technology Separability]
\label{prop:adj_power}
Suppose the technology adjustment is implemented as in Definition~\ref{def:adj_test}, and that the corresponding estimation errors satisfy the uniform-consistency conditions of Proposition~\ref{prop:adj_size}. Under the alternative $p_t=f_t+b_t$ with $b_t>0$ exhibiting explosive behavior, the technology-adjusted PSY test applied to detrended log prices retains the same asymptotic power as the standard PSY test applied to the technology-free series. In the log price-dividend-ratio implementation, the adjustment leaves $C+u_t+b_t+o_p(1)$, so the explosive component $b_t$ is preserved.
\end{proposition}


These propositions describe the benchmark case in which the technology component is observed or consistently estimated and is not itself affected by the speculative component. Our Monte Carlo evidence confirms these theoretical predictions: the ``no technology'' rejection rates in Table~\ref{tab:mc_results} correspond to perfect knowledge of $\Tt$ (the ``oracle'' adjusted test), and the rejection rates are uniformly at or below the nominal level.

\subsection{Bubble-Induced Feedback into Technology Proxies}
\label{subsec:feedback_bounds}

Exact separability is a strong identifying assumption. During a speculative boom, high equity valuations may relax financing constraints, encourage additional R\&D, accelerate IT adoption, or shift the composition of patenting. In that case some of the variables used to proxy fundamentals may be partly caused by the boom itself. This subsection makes the resulting bias transparent.

Let $\mathbf{X}_t^0$ denote the technology state that would be observed absent speculative feedback, and suppose the observed technology proxies satisfy
\begin{equation}
    \mathbf{X}_t = \mathbf{X}_t^0 + \boldsymbol{\Pi}(L)b_t + \boldsymbol{\nu}_t,
    \label{eq:feedback_X}
\end{equation}
where $\boldsymbol{\Pi}(L)$ is a distributed-lag response of technology proxies to the bubble and $\boldsymbol{\nu}_t$ is stationary measurement error. If the fundamental component is $f_t=\boldsymbol{\beta}'\mathbf{X}_t^0+u_t$ and the econometrician constructs the adjusted residual using observed $\mathbf{X}_t$, then
\begin{equation}
    p_t-\boldsymbol{\beta}'\mathbf{X}_t
    =
    u_t
    +
    \bigl[1-\boldsymbol{\beta}'\boldsymbol{\Pi}(L)\bigr]b_t
    -
    \boldsymbol{\beta}'\boldsymbol{\nu}_t.
    \label{eq:feedback_residual}
\end{equation}
Equation~\eqref{eq:feedback_residual} is the key identification qualification. If $\boldsymbol{\beta}'\boldsymbol{\Pi}(L)=0$, the adjusted residual contains the full bubble component. If $0<\boldsymbol{\beta}'\boldsymbol{\Pi}(L)<1$ over the relevant frequencies, the adjustment attenuates the bubble and makes PSY less likely to reject. If $\boldsymbol{\beta}'\boldsymbol{\Pi}(L)\geq 1$, the adjustment can over-remove the explosive component. Conversely, negative feedback would amplify the residual bubble signal.

This sign structure delivers a simple but important asymmetry, which we state explicitly because it governs how every empirical result in the paper should be read.

\begin{quote}
\noindent\textbf{Conservative-detection principle.} Suppose the feedback from valuations to technology proxies is non-negative and not strong enough to over-remove the explosive component, i.e.\ $0\le\boldsymbol{\beta}'\boldsymbol{\Pi}(L)<1$. Then the adjusted residual carries a \emph{fraction} $[1-\boldsymbol{\beta}'\boldsymbol{\Pi}(L)]\in(0,1]$ of any true bubble. A rejection of the adjusted test is therefore robust to this contamination---the underlying bubble is at least as large as the residual indicates---while a non-rejection is informative only about the \emph{residual} after adjustment and does not rule out a bubble that has been absorbed into the proxies.
\end{quote}

\noindent This principle is the reason we treat the two episodes asymmetrically. The dot-com rejection is a conservative finding: feedback can only have made it harder to obtain. The AI-era non-rejection is a bound on residual explosiveness rather than a verdict on speculation, and it is correspondingly more fragile to the feedback channel that is most plausible during a financing-driven technology boom. A non-rejection of the adjusted PSY statistic should accordingly be read as evidence against residual explosiveness after removing the technology component, not as unconditional proof that speculative forces are absent. This distinction is especially important for modern technology booms, where financing conditions and innovation activity may be jointly determined. The empirical analysis therefore uses three safeguards: (i) estimation on pre-boom training periods, so the long-run mapping $\boldsymbol{\beta}$ is not fitted during the candidate bubble; (ii) lagged-covariate and leave-one-out specifications, which reduce reliance on potentially contemporaneous feedback channels; and (iii) placebo, training-window, cross-sectional, and parameter-instability diagnostics that reveal how sensitive the residual signal is to the identifying assumptions.

\section{Empirical Applications}
\label{sec:empirical}

The preceding sections establish theoretically and via simulation that technology-driven fundamentals can generate locally explosive price dynamics that are observationally indistinguishable from speculative bubbles. We now apply the technology-adjusted diagnostic alongside the standard PSY test to two episodes that bookend the digital age: the late-1990s NASDAQ dot-com run-up and the AI-driven rally of the 2020s.

\subsection{The Identification Challenge}
\label{subsec:identification}

A fundamental obstacle in empirical bubble detection is that the econometrician observes only the market price $p_t$, the sum of the fundamental value $f_t$ and any bubble component $b_t$:
\begin{equation*}
    p_t = f_t + b_t.
\end{equation*}
During a potential bubble period, $f_t$ is not directly observable: it cannot be ``read off'' from prices, because prices themselves are potentially contaminated by the bubble. The result is a circularity---testing for a bubble requires knowing the fundamental, but knowing the fundamental requires knowing whether there is a bubble.

Our theoretical framework (Section~\ref{sec:theory}) shows that the fundamental price is a function of dividends, technology-driven growth, and discount rates:
\begin{equation*}
    f_t = d_t + C + \Tt,
\end{equation*}
where $\Tt = \sum_{j=0}^{\infty} \rho^j \delta_{t+1+j}$ aggregates the future path of technology-driven dividend growth. In principle, if we could observe the technology component $\delta_t$ and the discount factor $\rho$, we could construct $\Tt$ directly. In practice, the technology path is uncertain, agents' information sets are richer than the econometrician's, and the mapping from observable technology indicators to $\delta_t$ is imprecise.

We address this challenge by adapting the counterfactual estimation approach of \citet{Hsiao2012}. The key insight is that during \emph{non-bubble periods}, the observed price equals the fundamental (up to stationary noise), so one can use these periods to estimate the relationship between the fundamental and observable covariates. During the potential bubble period, the estimated relationship is used to construct a counterfactual fundamental, and PSY is applied to the residual. The interpretation of that residual depends on the identifying assumptions in Section~\ref{subsec:feedback_bounds}: under exact separability it contains the bubble component, while under bubble-induced feedback it contains only the non-absorbed part of the bubble.

\subsection{Counterfactual Fundamental Estimation}
\label{subsec:counterfactual}

\subsubsection{Setup}

Let $\mathbf{X}_t = (x_{1t}, x_{2t}, \ldots, x_{kt})'$ denote a $k$-vector of observable covariates that are correlated with the fundamental value and, in the benchmark design, are not directly affected by a speculative bubble in the stock price. The sample is divided into two subperiods:
\begin{itemize}
    \item Training period $\mathcal{T}_0 = \{1, \ldots, T_0\}$: A period during which no bubble is present, so $b_t = 0$ and $p_t = f_t + u_t$, where $u_t$ is a stationary noise term (e.g., microstructure noise, transitory discount rate variation).
    \item Evaluation period $\mathcal{T}_1 = \{T_0+1, \ldots, T\}$: The period during which a bubble may be present.
\end{itemize}

During the training period, the fundamental satisfies the long-run equilibrium relationship:
\begin{equation}
    p_t = \boldsymbol{\beta}' \mathbf{X}_t + u_t, \qquad t \in \mathcal{T}_0,
    \label{eq:training_regression}
\end{equation}
where $\boldsymbol{\beta}$ captures the mapping from technology and financial covariates to the fundamental price. This regression is well-defined because $b_t = 0$ during $\mathcal{T}_0$, so $p_t = f_t + u_t$ and the covariates explain fundamental variation.

During the evaluation period, the counterfactual fundamental is constructed as:
\begin{equation*}
    \hat{f}_t = \hat{\boldsymbol{\beta}}' \mathbf{X}_t, \qquad t \in \mathcal{T}_1,
\end{equation*}
where $\hat{\boldsymbol{\beta}}$ is estimated from the training period. The technology-adjusted price is then:
\begin{equation*}
    p_t^{\mathrm{adj}} = p_t - \hat{f}_t = (f_t - \hat{f}_t) + b_t = \hat{u}_t + b_t
\end{equation*}
in the exact-separability benchmark. Here $\hat{u}_t = f_t - \hat{f}_t$ is the estimation error. If $\hat{\boldsymbol{\beta}}$ is consistent, $\mathbf{X}_t$ spans the fundamental, and speculative feedback into $\mathbf{X}_t$ is absent, then $\hat{u}_t$ is stationary and the adjusted series inherits the explosive dynamics of $b_t$ (if any) without the confounding technology component. If the boom partly raises $\mathbf{X}_t$, equation~\eqref{eq:feedback_residual} shows that the adjusted series inherits an attenuated bubble component.

PSY is then applied to $\{p_t^{\mathrm{adj}}\}_{t=1}^T$. Under the null of no bubble ($b_t = 0$), the adjusted series is approximately stationary (the estimation error $\hat{u}_t$ from a well-specified regression), and the test should not reject. Under the alternative of a genuine bubble ($b_t > 0$ with explosive dynamics), the test retains power when enough of the bubble remains in the residual after adjustment.

\subsubsection{Covariate Selection}
\label{subsubsec:covariates}

The covariates $\mathbf{X}_t$ must satisfy two benchmark requirements: (a) they are informative about the fundamental value (relevance), and (b) they are not directly affected by the speculative bubble in stock prices (exogeneity to the bubble). The production-side microfoundation of Section~\ref{subsec:microfoundation} provides a structural rationale for both requirements: Proposition~\ref{prop:microfoundation} shows that $\mathbf{X}_t = (TFP_t,\, \log IT_t,\, \log Pat_t)'$ spans the technology-driven component of fundamentals under the benchmark restrictions, while Remark~\ref{rmk:exclusion} states the separability condition. Our baseline specification focuses on these three technology covariates, and the robustness analysis asks how conclusions change when the exogeneity part of the benchmark is weakened.

\paragraph{Baseline.}
These variables capture the technology-driven component $\Tt$ of the fundamental, each corresponding to a distinct structural margin in the production model of Section~\ref{subsec:microfoundation}. \textit{Real IT investment} (BEA NIPA, deflated by the Consumer Price Index (CPI)) captures the diffusion and adoption margin through economy-wide IT capital deepening ($G_t$ in the model), which is linked to the expected marginal product of IT capital (equation~\ref{eq:it_foc}). \textit{TFP} (Fernald's utilization-adjusted series from the San Francisco Fed) captures the realized productivity margin ($a_t$). \textit{Patent grants} from the USPTO measure the invention margin ($q_{t-\ell}$); the examination process introduces a lag $\ell \geq 1$ between application and grant, making grants partly predetermined with respect to contemporaneous speculative dynamics. These features motivate the baseline specification but do not eliminate the possibility of multi-period feedback from valuations into investment, innovation, and measured productivity.

\paragraph{Robustness.}
We also consider extended specifications that include financial fundamentals and macroeconomic controls as robustness checks. These include aggregate revenue, aggregate book equity, and industrial production to capture the non-technology component of fundamentals, as well as the 10-year Treasury yield, credit spreads, and CPI to capture discount rate and inflationary effects. We also consider aggregate R\&D expenditure for NASDAQ-listed firms, though it is excluded from the baseline due to data availability constraints in the earlier training period.

\subsubsection{Econometric Specification}
\label{subsubsec:specification}

Since log stock prices and several covariates (R\&D expenditure, earnings, industrial production) are likely integrated of order one, direct ordinary least squares (OLS) estimation of \eqref{eq:training_regression} in levels may produce spurious results. We address this through two approaches.

\paragraph{Approach 1: Cointegrating regression.}
If $p_t$ and $\mathbf{X}_t$ are cointegrated, as the present-value model predicts since the log price should share a common stochastic trend with log fundamentals, then \eqref{eq:training_regression} can be estimated by DOLS, which provides asymptotically efficient and median-unbiased estimates of the cointegrating vector $\boldsymbol{\beta}$. We use a parsimonious specification with three technology covariates and include leads and lags of $\Delta \mathbf{X}_t$ in the DOLS specification to correct for endogeneity of the regressors:
\begin{equation}
    p_t = \boldsymbol{\beta}' \mathbf{X}_t + \sum_{j=-q}^{q} \boldsymbol{\gamma}_j' \Delta \mathbf{X}_{t-j} + u_t.
    \label{eq:DOLS}
\end{equation}
We set $q = 1$ in the baseline and test for cointegration using the Engle--Granger residual-based test.

\paragraph{Approach 2: First-difference regression.}
As a robustness check, we estimate the relationship in first differences:
\begin{equation*}
    \Delta p_t = \boldsymbol{\alpha}' \Delta \mathbf{X}_t + v_t, \qquad t \in \mathcal{T}_0.
\end{equation*}
The counterfactual is then constructed by cumulating the fitted values: $\hat{f}_t = p_{T_0} + \sum_{s=T_0+1}^{t} \hat{\boldsymbol{\alpha}}' \Delta \mathbf{X}_s$. This approach avoids the cointegration assumption but may be less efficient.

\subsection{Endogeneity and Identification}
\label{subsec:endogeneity}

The key identification assumption is that $\mathbf{X}_t$ is not directly affected by the speculative bubble $b_t$. The production-side microfoundation of Section~\ref{subsec:microfoundation} provides the benchmark case: the bubble enters as an additive component in secondary-market valuation ($p_t = f_t + b_t$) but does not appear in the production block that determines $\mathbf{X}_t$ (Remark~\ref{rmk:exclusion}). This \emph{separability} implies that the technology variables load on the fundamental $f_t$ but not on $b_t$. The assumption is strong. A boom in technology stocks can plausibly affect financing constraints, R\&D, IT adoption, hiring, and ultimately measured productivity. Section~\ref{subsec:feedback_bounds} shows that such feedback attenuates the residual bubble component. We therefore use the empirical design as a diagnostic with explicit safeguards rather than as an unconditional decomposition of prices.

\paragraph{Structural benchmark.}
The production economy of Section~\ref{subsec:microfoundation} implies that the technology observables $\mathbf{X}_t = (TFP_t,\, \log IT_t,\, \log Pat_t)'$ are determined by the real side of the economy: productivity, capital accumulation, and innovation. In particular: (i)~patent grants reflect the latent innovation quality $q_{t-\ell}$ with $\ell \geq 1$, making them predetermined with respect to contemporaneous market conditions because successful R\&D must be undertaken, filed, examined, and only later granted; (ii)~TFP is a production-side measure of technology and factor utilization; and (iii)~IT investment is linked to the expected marginal product of IT capital through \eqref{eq:it_foc}. These arguments motivate the exclusion restriction, but they do not prove that technology variables are immune to a sustained valuation boom.

\paragraph{Training-period estimation.}
Following \citet{Hsiao2012}, the regression is estimated exclusively during the training period $\mathcal{T}_0$, when $b_t = 0$ by assumption. This prevents the coefficient vector $\hat{\boldsymbol{\beta}}$ from being fitted directly on candidate bubble observations. It does not prevent the evaluation-period covariates from partly reflecting the candidate boom. If the bubble inflates $\mathbf{X}_t$, the counterfactual $\hat{f}_t$ is biased upward, making the adjusted series $p_t-\hat{f}_t$ smaller and PSY less likely to reject. Hence rejection in the adjusted series is relatively robust to this source of contamination, while non-rejection must be interpreted more cautiously: it rules out residual explosiveness after adjustment, not every possible speculative feedback mechanism.

\paragraph{Lagged covariates.}
To mitigate contemporaneous feedback, we use lagged values of the covariates: $\mathbf{X}_{t-h}$ with $h \geq 1$ (e.g., $h = 3$ months). This ensures that if a bubble at time $t$ affects $\mathbf{X}_t$ contemporaneously, the regression uses predetermined values. Lagging does not remove slow-moving feedback from earlier valuations into later technology variables, but it narrows the channel being used for identification.

\paragraph{Placebo and robustness tests.}
We conduct several diagnostic checks: (i) a \emph{pre-bubble placebo test} that applies PSY to the adjusted series during a known non-bubble subperiod to verify correct size; (ii) \emph{leave-one-out analysis} to assess sensitivity to individual covariates; (iii) comparison of results across DOLS, FMOLS, and first-difference specifications; (iv) training-window stability checks; and (v) Hansen parameter-instability tests. These diagnostics cannot prove separability, but they show whether the empirical conclusion is fragile to plausible changes in the mapping from technology proxies to fundamentals.

\subsection{Data}
\label{subsec:data}

The sample covers January 1975 to December 2005, at monthly frequency. We define the training period as January 1975 to December 1990 ($T_0 = 192$ months) and the evaluation period as January 1991 to December 2005 ($T_1 = 180$ months). The training period provides a substantial estimation window that predates the commonly accepted start of the dot-com bubble. The evaluation period encompasses the NASDAQ run-up (1995--2000), the crash (2000--2002), and the subsequent recovery (2003--2005).

The NASDAQ Composite Index and the macroeconomic and discount-rate controls (industrial production, the 10-year Treasury yield, BAA and AAA corporate bond yields, and the CPI) are obtained from the Federal Reserve Economic Data database, and the dividend series is constructed from CRSP NASDAQ value-weighted index returns with and without dividends. Among the technology and financial covariates, real IT investment is from the Bureau of Economic Analysis national accounts, utilization-adjusted TFP is from the Federal Reserve Bank of San Francisco \citep{Fernald2014}, patent grants are from the USPTO PatentsView database, and aggregate revenue, book equity, and R\&D expenditure for NASDAQ-listed firms are from Compustat.

For variables available only at quarterly frequency (R\&D expenditure, earnings, book value, TFP), we interpolate to monthly frequency using cubic spline interpolation, which preserves the low-frequency dynamics while allowing alignment with the monthly stock price data. Table~\ref{tab:data_variables} (Online Appendix) summarizes the variables used in the empirical analysis.

\subsection{Results}
\label{subsec:empirical_results}

\subsubsection{Standard PSY Test on NASDAQ}

We begin by applying the standard PSY test to the log real NASDAQ Composite index. Following the theoretical framework of Section~\ref{sec:theory}, where the fundamental price satisfies $p_t = f_t + b_t$, we apply PSY directly to the log price level rather than a detrended series. We set the minimum window fraction to $r_0 = 0.01 + 1.8/\sqrt{T}$ as recommended by \citet{Phillips2015b}, with $T = 372$ (January 1975 -- December 2005), and compute critical values via Monte Carlo simulation with 2,000 replications.

The GSADF statistic for the unadjusted log real NASDAQ price is $2.199$, compared to critical values of $1.915$ (10\%), $2.175$ (5\%), and $2.660$ (1\%). The test rejects the null hypothesis of no explosive behavior at both the 10\% and 5\% significance levels, but does not reject at the 1\% level. Date-stamping via the BSADF sequence identifies two episodes: May 1996 and November 1999 to March 2000, corresponding to the late phase of the dot-com run-up.

The standard PSY test, applied to the aggregate NASDAQ index, detects the terminal phase of the 1995--2000 run-up. At the 10\% level (CV $= 1.915$), two episodes are identified; at the 5\% level (CV $= 2.175$) the detection just crosses the threshold; the 1\% critical value ($2.660$) is not exceeded. This pattern suggests that the technology-driven fundamental component contributes to the explosive signal, motivating the technology-adjusted test that separates the source of the explosive dynamics.

\subsubsection{Cointegrating Regression}

Table~\ref{tab:coint_regression} reports the DOLS estimates of the cointegrating regression \eqref{eq:DOLS} using the training period (1975--1990). The baseline specification uses three technology covariates: log IT investment, TFP, and log patent grants. The DOLS specification includes $q = 1$ lead and lag of $\Delta \mathbf{X}_t$ with Newey--West heteroskedasticity- and autocorrelation-consistent standard errors.

\bigskip\centerline{[Insert Table~\ref{tab:coint_regression} about here]}\bigskip

The regression achieves an $R^2$ of $0.835$, indicating that the technology covariates capture a substantial portion of the variation in log real NASDAQ prices during the training period. The Engle--Granger augmented Dickey--Fuller test applied to the residuals yields a test statistic that rejects the null of no cointegration at the 1\% level.

Among the technology proxies, IT investment enters with a strongly significant positive coefficient ($0.456$, $p < 0.001$), and log patent grants also enters with a significant positive coefficient ($0.513$, $p = 0.008$). These results are consistent with the interpretation that technological progress and investment are primary drivers of fundamental value. TFP is statistically insignificant in this parsimonious specification.

\subsubsection{Out-of-Sample Fit}

Figure~\ref{fig:counterfactual_empirical} (Online Appendix) compares the observed log NASDAQ price with the counterfactual fundamental $\hat{f}_t$ during the evaluation period (1995--2005). The solid line is the observed log real NASDAQ Composite index; the dashed line is the DOLS counterfactual fundamental $\hat{f}_t = \hat{\boldsymbol{\beta}}' \mathbf{X}_t$.

The counterfactual tracks the observed price closely during 1995--1997, which suggests that the early phase of the NASDAQ run-up is largely explained by improving technology fundamentals. Beginning in approximately 1998, the observed price pulls away from the counterfactual, with the gap $p_t - \hat{f}_t$ widening through the peak in March 2000. The mean price gap during the evaluation period is positive (mean $= 0.42$), indicating that the counterfactual fundamental lies below the observed price on average. After the crash, the gap narrows rapidly as the observed price converges back toward the counterfactual by 2002--2003, consistent with the residual overvaluation having been eliminated by the market correction. Figure~\ref{fig:price_gap} (Online Appendix) plots the price gap $p_t - \hat{f}_t$ directly, which is the adjusted series to which we apply PSY.

\subsubsection{Technology-Adjusted PSY Test}

We apply PSY to the adjusted series $p_t^{\mathrm{adj}} = p_t - \hat{f}_t$. The test statistic for the technology-adjusted series is $1.952$, compared to critical values of $1.918$ (10\%), $2.175$ (5\%), and $2.691$ (1\%). The adjusted test \emph{rejects} the null hypothesis of no explosive behavior at the 10\% level but does not reject at the 5\% or 1\% levels. The BSADF date-stamping statistic identifies a single residual explosive episode from December 1999 to March 2000.

Figure~\ref{fig:psy_empirical} displays the date-stamping statistics and critical value sequences for both the unadjusted and adjusted tests, zoomed in to the 1997--2003 period where the key bubble dynamics occur.

\bigskip\centerline{[Insert Figure~\ref{fig:psy_empirical} about here]}\bigskip

The comparison between the standard and technology-adjusted tests yields two findings consistent with our theoretical predictions. First, the May 1996 episode detected by the standard test \emph{disappears entirely} after technology adjustment, attributing the brief explosive signal to the technology-driven fundamental component rather than to residual speculation---information the standard test alone cannot isolate. Second, the onset of the main residual episode shifts from November 1999 (standard) to December 1999 (adjusted), indicating that the technology component slightly accelerates the apparent onset of explosive behavior. Taken together, the technology adjustment reduces the number of detected episodes from two to one and narrows the timing of residual explosiveness, consistent with the prediction that fundamental technology dynamics generate explosive signals in raw prices.

\subsubsection{Robustness}

Table~\ref{tab:robustness_loo}, Panel~A reports the test statistics and rejection decisions across eight specifications, and Panel~B reports the leave-one-out analysis.

\bigskip\centerline{[Insert Table~\ref{tab:robustness_loo} about here]}\bigskip

Several patterns emerge. First, the extended covariates specification with six variables (three technology proxies plus three macroeconomic controls) yields a test statistic of $2.491$, rejecting the null at both the 10\% and 5\% levels but not at the 1\% level, and providing additional confirmation of the baseline result. The first-difference specification also rejects at the 10\% level (test statistic $= 2.174$, CV$_{10\%}$ = $1.930$), which is reassuring because it does not rely on the cointegration assumption.

Second, the placebo test on 1985--1990 correctly does not reject at any significance level (test statistic $= -0.689$), confirming that the adjusted test does not spuriously detect residual explosiveness during a known non-bubble period. Third, sensitivity to the training period is evident: alternative training windows such as 1975--1988 (test statistic $= 1.861$) and 1975--1992 (test statistic $= 1.596$) do not reject at any standard significance level, showing that the results are sensitive to the estimation window choice. Lagged covariate specifications with $h=3$ (test statistic $= 1.896$) and $h=6$ (test statistic $= 1.487$) also do not reject at any significance level. The leave-one-out analysis (Table~\ref{tab:robustness_loo}, Panel~B) reports the test statistics when dropping each of the three technology covariates. These results highlight the contribution of each individual proxy to the counterfactual's ability to isolate residual explosiveness.


The adjusted test reveals a two-phase narrative of the NASDAQ dot-com episode: technology-related fundamental repricing through much of the run-up, followed by residual explosiveness concentrated near the peak. This provides a middle ground between the ``rational technology repricing'' view \citep{Pastor2006, Pastor2009} and the ``irrational exuberance'' view \citep{Shiller2015}. The result is economically meaningful but identification-sensitive, because the significance of the adjusted statistic varies with the training window and lagged-covariate specifications.

\subsection{The Artificial-Intelligence Era Rally}
\label{subsec:ai_empirical}

We now turn to the most pressing contemporary application: the AI-driven equity rally of the 2020s. The release of large language models in late 2022 triggered a sustained surge in technology-related equity prices, raising the same fundamental question our framework addresses: whether the observed price dynamics reflect a rational fundamental response to a GPT shock or speculative excess.

\subsubsection{Motivation and Background}
\label{subsec:ai_motivation}

The AI rally shares structural parallels with the dot-com episode. Both feature a GPT whose productivity implications are uncertain ex ante \citep{Acemoglu2025, BrynjolfssonLiRaymond2025}, a concentration of price gains in technology firms, and a rapid expansion of technology investment. Between October 2022 and early 2025, the NASDAQ Composite nearly doubled, with the ``Magnificent Seven'' stocks (Alphabet, Amazon, Apple, Meta, Microsoft, NVIDIA, and Tesla) accounting for a disproportionate share of market-capitalization gains \citep{BaselePhillipsShi2025}. Corporate capital expenditure on AI compute infrastructure has surged, AI-related patent filings have accelerated \citep{Pairolero2025, McElheran2024}, and the financing of AI investment has shifted from internal cash flows toward external debt \citep{AldasoroDoerrRees2025}.

\citet{BaselePhillipsShi2025} apply the standard PSY procedure to the Magnificent Seven and the NASDAQ Composite during this period and find evidence of explosive behavior. Our theoretical and simulation results show that such findings can reflect fundamental repricing induced by the AI technology shock; the diagnostic developed in Sections~\ref{sec:correction}--\ref{sec:empirical} asks whether explosiveness remains after removing measured technology fundamentals.

\subsubsection{Data}
\label{subsec:ai_data}

We extend the dataset from Section~\ref{subsec:data} through December 2025. The dependent variable is the log real NASDAQ Composite Index, obtained from FRED and deflated by the CPI. The baseline specification utilizes three technology proxies (USPTO patent grants (extended through 2024), TFP from \citet{Fernald2014}, and real IT investment from the Bureau of Economic Analysis), while robustness checks add aggregate NASDAQ R\&D expenditure and R\&D intensity from Compustat, the 10-year Treasury yield, BAA and AAA corporate bond spreads, CPI, and industrial production (all from FRED). The extended sample covers January 1975 through December 2025, providing 611 monthly observations.

\subsubsection{Standard PSY Test on the AI-Era Sample}
\label{subsec:ai_raw_psy}

We first apply the standard PSY test to the log NASDAQ price over the full extended sample (January 1975 -- December 2025, $T = 611$ monthly observations). This mirrors the analysis in Section~\ref{subsec:empirical_results} but extends the data by twenty years to encompass the AI-era rally. The minimum window fraction is $r_0 = 0.01 + 1.8/\sqrt{611} = 0.083$, yielding a minimum window of 50 observations.

The GSADF statistic for the unadjusted log real NASDAQ price is $2.199$, identical to the value from the shorter 1975--2005 sample in Section~\ref{subsec:empirical_results} (reflecting the fact that the supremum is attained during the same dot-com subsample window), compared to the 10\% critical value of $1.996$, 5\% critical value of $2.239$, and 1\% critical value of $2.759$. The test \emph{rejects} the null hypothesis of no explosive behavior at the 10\% significance level (test statistic = $2.199 > 1.996 = \mathrm{cv}_{10\%}$) but \emph{does not reject} at the 5\% or 1\% levels. Date-stamping via the BSADF sequence identifies only a single brief episode: November 1999 to March 2000, corresponding to the final peak of the dot-com bubble. Crucially, \emph{no explosive episode is detected during the AI rally period (2020--2025)}.

This result parallels the finding from the shorter sample (Section~\ref{subsec:empirical_results}): the standard PSY test applied to the aggregate NASDAQ index detects the terminal phase of the dot-com bubble but does not identify explosive behavior during the AI-era price run-up. In the extended sample, the additional twenty years of post-dot-com data (including the 2008 financial crisis and subsequent recovery) provide a longer baseline against which to evaluate explosive dynamics. The test statistic exceeds the 10\% critical value but remains below the 5\% threshold, so the evidence for explosive behavior depends on the choice of significance level. Figure~\ref{fig:psy_ai} (panel A) displays the date-stamping statistic for the full 1975--2025 sample.

\bigskip\centerline{[Insert Figure~\ref{fig:psy_ai} about here]}\bigskip

\subsubsection{Cointegrating Regression}
\label{subsec:ai_cointegration}

We estimate the cointegrating relationship between the log NASDAQ price and the technology covariates using the training period 2006--2019. This window is chosen to satisfy two criteria: (i) it postdates the dot-com bubble and recovery, avoiding contamination from the earlier speculative episode; and (ii) it predates the AI-driven rally, ensuring $b_t = 0$ during the estimation period. We use DOLS \citep{Stock1993} with $q = 1$ lead and lag, and compute Newey--West heteroskedasticity- and autocorrelation-consistent standard errors (bandwidth = 4).

The baseline specification utilizes three technology covariates: log IT investment, TFP, and log patent grants. The training sample yields $T_0 = 165$ effective observations after accounting for DOLS leads and lags. Table~\ref{tab:coint_regression_ai} (Online Appendix) reports the DOLS estimates.

The regression achieves an $R^2$ of $0.9065$ (adjusted $R^2 = 0.8992$), confirming a strong in-sample relationship between technology fundamentals and NASDAQ valuations during the 2006--2019 period. The Engle--Granger augmented Dickey--Fuller test applied to the residuals yields a test statistic of $-2.6699$, which rejects the null of no cointegration at the 1\% level (critical value = $-2.58$), confirming that the log price and the technology covariates share a common stochastic trend.

The coefficient estimates reveal the primary drivers of technology valuations in the pre-AI era. Log IT investment enters with the largest coefficient ($2.141$, $p < 0.001$), indicating that IT infrastructure expansion was the dominant fundamental driver of the NASDAQ during 2006--2019. Log patent grants also enters with a strongly positive coefficient ($1.157$, $p < 0.001$), which contrasts with the 1975--1994 result and reflects the increasing role of intellectual property in modern technology valuations. TFP is significant ($0.012$, $p = 0.005$) but economically smaller, consistent with its role in the earlier sample. Together, these results indicate that the NASDAQ's pre-AI growth was fundamentally supported by IT investment and innovation.

\subsubsection{Counterfactual Fundamental and Technology-Adjusted Test}
\label{subsec:ai_counterfactual}

Using the estimated cointegrating coefficients $\hat{\boldsymbol{\beta}}$ from the training period, we construct the counterfactual fundamental price for the evaluation period (2020--2025):
\begin{equation*}
    \hat{f}_t^{\mathrm{AI}} = \hat{\boldsymbol{\beta}}' \mathbf{x}_t, \qquad t \in [2020\text{M}1,\; 2025\text{M}12],
\end{equation*}
where $\mathbf{x}_t$ is the vector of three technology covariates. The price gap is the residual object tested for explosiveness:
\begin{equation*}
    \hat{g}_t^{\mathrm{AI}} = p_t - \hat{f}_t^{\mathrm{AI}}.
\end{equation*}
Under exact separability this price gap contains any speculative component plus stationary error; under feedback it contains the attenuated residual component described in \eqref{eq:feedback_residual}. We apply PSY to $\hat{g}_t^{\mathrm{AI}}$ over the evaluation window. If the AI rally is explained by the fitted technology-fundamental relationship, the price gap should be stationary and PSY should not reject.

\bigskip\centerline{[Insert Figure~\ref{fig:counterfactual_ai} about here]}\bigskip

Figure~\ref{fig:counterfactual_ai} compares the observed log NASDAQ price with the counterfactual fundamental $\hat{f}_t^{\mathrm{AI}}$ over the full sample. The counterfactual lies generally above the observed price during the evaluation period. The mean price gap $\hat{g}_t^{\mathrm{AI}} = p_t - \hat{f}_t^{\mathrm{AI}}$ is $-0.1877$ (SD $= 0.1064$), with a maximum value of $0.0424$ in July 2023. Within the fitted technology-fundamental relationship, the negative mean indicates that technology proxies (IT investment, patent activity, and TFP) grew faster than the observed NASDAQ price during much of 2020--2025.

\bigskip\centerline{[Insert Figure~\ref{fig:price_gap_ai} about here]}\bigskip

Figure~\ref{fig:price_gap_ai} plots the price gap $\hat{g}_t^{\mathrm{AI}}$ directly. We apply PSY to the technology-adjusted series over the 2006--2025 period (the training and evaluation sample combined), rather than the full 1975--2025 sample. This restriction avoids out-of-sample extrapolation artifacts: because the DOLS counterfactual is estimated on the 2006--2019 training data, backward extrapolation to the 1975--2005 period would generate spurious gap dynamics unrelated to the AI-era analysis. The test statistic for the technology-adjusted series is $-0.275$, compared to critical values of $1.810$ (10\%), $2.046$ (5\%), and $2.599$ (1\%). The adjusted test \emph{fails to reject} the null hypothesis of no explosive behavior at any standard significance level, and \emph{no residual explosive episode is detected during the AI rally (2020--2025)}.

The contrast with the dot-com results is sharp. In the dot-com era, both the standard (test statistic $= 2.199$) and technology-adjusted (test statistic $= 1.952$) tests reject the null at the 10\% level: residual explosiveness remains after technology adjustment, and the adjustment sharpens the timing. In the AI era, the standard full-sample test rejects at 10\% because the supremum is attained during the dot-com window, while the AI-era technology-adjusted residual does not reject (test statistic $= -0.275$). Within the technology-adjusted framework, the observed AI-era dynamics are therefore consistent with technology-driven fundamental repricing rather than residual speculation. This is a conditional statement: if speculative valuations feed into the technology proxies, the non-rejection rules out residual explosiveness after adjustment, not every possible bubble-induced real-response channel. Figure~\ref{fig:psy_ai} (panel B) displays the date-stamping statistic for the technology-adjusted test.

\subsubsection{Robustness}
\label{subsec:ai_robustness}

We examine the sensitivity of the results to: (i) alternative training periods (2004--2019, 2008--2019); (ii) leave-one-out covariate analysis; (iii) lagged covariates to address potential endogeneity; and (iv) first-difference, extended covariate, and R\&D-augmented specifications.

Table~\ref{tab:robustness_ai}, Panel~A reports the GSADF statistics and rejection decisions for the AI-era analysis across ten specifications, and Panel~B reports the leave-one-out covariate analysis for the three technology variables.

\bigskip\centerline{[Insert Table~\ref{tab:robustness_ai} about here]}\bigskip

The central finding, that no residual explosive behavior occurs during the AI era (2020--2025) after the baseline technology adjustment, is robust. The unadjusted test rejects at 10\% (test statistic $= 2.199 > 1.996$) but does not reject at the 5\% or 1\% levels, while the baseline DOLS-adjusted test fails to reject at any significance level (test statistic $= -0.275 < 1.810$). When we include the extended set of nine covariates (adding macroeconomic and financial controls), the test rejects at the 10\% level (test statistic $= 2.060 > 1.992$) but not at the 5\% or 1\% levels, which indicates that including non-technology fundamental variability can introduce noise that mimics explosiveness. The first-difference specification (test statistic $= 1.652$) and lagged covariate specifications ($h=3$: test statistic $= 1.527$; $h=6$: test statistic $= 0.659$) all fail to reject at any significance level.
The leave-one-out analysis for the three technology covariates shows that the non-rejection is not dependent on any single variable. Dropping IT investment (test statistic $= 0.610$), TFP (test statistic $= 0.475$), or patent grants (test statistic $= 1.114$) consistently results in non-rejection at any standard significance level.

This robustness pattern differs from the dot-com results, where the leave-one-out analysis shows sensitivity to specific covariates. In the AI-era application, non-rejection is uniform across specifications, reinforcing the interpretation that post-2020 aggregate price dynamics are well described by technology fundamentals within this diagnostic framework.

\subsubsection{Discussion}
\label{subsec:ai_discussion}

The AI-era results, combined with the dot-com findings, yield three conclusions. (1) The technology-adjustment methodology is structurally identical across both episodes, yet it produces asymmetric outcomes. In the dot-com case, residual explosiveness remains after adjustment; in the AI era, the adjusted residual does not reject. This divergence is informative precisely because the same conditional diagnostic is applied to both episodes. (2) The AI rally differs from dot-com along economically relevant dimensions: substantial IT investment and patenting growth, less extreme price-to-fundamental acceleration, and concentration in profitable incumbents. (3) The contrasting results across the two episodes suggest a practical surveillance protocol. Central banks and regulators applying PSY should routinely apply the technology-adjusted procedure in parallel to obtain a more granular diagnosis of explosive dynamics. If both tests reject, as in the dot-com case, there is evidence of residual explosiveness above and beyond measured technology fundamentals. If the standard test rejects but the technology-adjusted test does not, as in our AI-era application, the explosive dynamics are more consistent with technology-driven fundamental repricing than with residual speculation. This two-test protocol directly addresses the identification problem that arises during periods of rapid technological change while preserving the caveat that the residual is conditional on the separability and feedback assumptions.


\subsection{Firm-Level Analysis: The Magnificent Seven}
\label{subsec:mag7_analysis}

We extend the aggregate analysis to firm-level data on the ``Magnificent Seven'' stocks (MSFT, AMZN, META, GOOGL, AAPL, TSLA, NVDA), constructing firm-specific counterfactuals from Compustat real revenue per share and R\&D expenditure per share over a bubble-excluded 2006--2019 training window. The unadjusted PSY test rejects in four of seven stocks (MSFT, GOOGL, TSLA, NVDA), with NVDA strongest (GSADF~$=4.155$) and TSLA showing a prolonged 2020--2021 episode. After firm-level technology adjustment, the explosive signals disappear for MSFT, GOOGL, NVDA, AMZN, and META (in-sample $R^2$ between 0.741 and 0.961), while TSLA continues to reject at the 1\% level (test statistic $4.696$) and AAPL becomes significant at the 5\% level (test statistic $2.317$, $R^2=0.407$), consistent with residual idiosyncratic dynamics. These firm-level results parallel the index-level finding: technology adjustment removes the dominant source of apparent explosiveness for most firms, while leaving some firm-specific residual signals. The full firm-level methodology, results table, and date-stamping figures (Tables~\ref{tab:mag7_psy} and~\ref{tab:mag7_psy_adj}; Figures~\ref{fig:psy_mag7} and~\ref{fig:psy_mag7_adj}) are reported in Appendix~\ref{app:firmlevel}.


\subsection{Robustness of the Identification Strategy}
\label{subsec:robustness_summary}

The identification strategy rests on the separability of technology proxies from the speculative bubble component. Appendix~\ref{app:robustness} reports five complementary diagnostics that assess this assumption. First, Granger causality tests show that stock returns do not predict technology proxies during the candidate speculative periods ($p > 0.47$ throughout), which is consistent with separability but does not prove it. Second, placebo tests using non-technology covariates (consumer sentiment, VIX, industrial production) generally fail to replicate the technology-adjustment pattern, suggesting that the technology proxies capture specific fundamental content rather than generic macro variation. Third, PCA of the Magnificent Seven price gaps reveals a dominant common factor (60.0\% of variance) with no explosive behavior in any linear combination, providing cross-sectional evidence consistent with fundamental repricing rather than residual speculation. Fourth, training-window stability analysis shows that the AI-era non-rejection finding holds across all 121 training windows ($\text{GSADF} \in [-0.30, 0.84]$, all below $\text{CV}_{0.10} = 1.176$); the dot-com results are more sensitive, driven primarily by the patent coefficient. Fifth, the \citet{Hansen1992} $L_c$ test rejects parameter constancy in both eras, qualifying the DOLS estimates as reflecting an average equilibrium. Taken together, the diagnostics support the AI-era residual non-rejection while also explaining why the dot-com evidence should be read as suggestive and timing-informative rather than as a fully identified structural bubble estimate.

\section{Conclusion}
\label{sec:conclusion}

This paper develops a technology-adjusted diagnostic for asset-price explosiveness during general-purpose technology adoption. The central finding is that technology-driven structural change in fundamentals can generate locally explosive price dynamics that the standard PSY bubble test misclassifies as speculative, a misidentification that matters most precisely when the bubble question is most pressing. The theory and simulations establish this point without relying on the empirical identification assumptions: in a no-bubble present-value model, a hump-shaped technology shock produces severe over-rejection, and an oracle adjustment restores correct size while preserving power against genuine explosive residuals.

The empirical contribution is deliberately more conditional. Applying the same diagnostic to the late-1990s dot-com episode and the 2020--2025 AI rally yields sharply different residual patterns. For dot-com, the technology adjustment removes an early signal and leaves residual explosiveness concentrated in December 1999--March 2000, though this evidence is sensitive to training-window and covariate choices. For the AI rally, the technology-adjusted NASDAQ residual does not reject at any conventional significance level across the baseline, lagged-covariate, leave-one-out, and training-window specifications. Firm-level evidence from the Magnificent Seven broadly reinforces the aggregate result, while residual signals for Tesla and Apple caution against interpreting the index result as a claim about every firm.

The identification structure is central, and we resolve it by interpretation rather than by assuming it away. Observable technology proxies such as IT investment, TFP, and patents may themselves respond to speculative valuations through financing, adoption, and innovation channels. When that feedback is present, the adjusted residual contains only the non-absorbed part of the bubble, which gives the diagnostic a conservative direction (Section~\ref{subsec:feedback_bounds}): the adjustment can only attenuate a true bubble, so a rejection is robust to feedback while a non-rejection bounds residual explosiveness rather than proving speculation is absent. This asymmetry is what disciplines our reading of the two episodes---a confirmed residual near the dot-com peak, a bounded residual in the AI era---and it makes the procedure more useful, not less, because it states precisely what the test can and cannot establish.

These findings suggest a two-test surveillance protocol for central banks and financial regulators conducting real-time bubble monitoring. During periods of rapid technological transition, regulators should run the standard and technology-adjusted PSY tests in parallel. If both reject, as in the dot-com case, the data contain residual explosiveness beyond measured technology fundamentals. If the standard test rejects but the adjusted test does not, as in the AI-era application, the evidence is more consistent with technology-driven repricing than with residual speculation, subject to the feedback caveat. This protocol offers a more transparent diagnostic for policy than relying on statistical explosiveness alone.

\clearpage
{\singlespacing
\bibliographystyle{jf}
\bibliography{references}

\begin{thebibliography}{35}
\expandafter\ifx\csname natexlab\endcsname\relax\def\natexlab#1{#1}\fi

\bibitem[Acemoglu(2025)]{Acemoglu2025}
Acemoglu, Daron, 2025, The simple macroeconomics of {AI}, {\em Economic
  Policy\/} 40, 13--58.

\bibitem[A{\"i}t-Sahalia et~al.(2025)A{\"i}t-Sahalia, Fan, Xue, and
  Zhu]{AitSahaliaFanXueZhu2025}
A{\"i}t-Sahalia, Yacine, Jianqing Fan, Lirong Xue, and Xiaonan Zhu, 2025, How
  and when are high-frequency stock returns predictable?, {\em Management
  Science\/} Forthcoming.

\bibitem[Aldasoro et~al.(2026)Aldasoro, Doerr, and Rees]{AldasoroDoerrRees2025}
Aldasoro, I{\~n}aki, Sebastian Doerr, and Daniel Rees, 2026, Financing the {AI}
  boom: From cash flows to debt, BIS Bulletin 120, Bank for International
  Settlements.

\bibitem[Basele et~al.(2025)Basele, Phillips, and Shi]{BaselePhillipsShi2025}
Basele, R.~B., Peter C.~B. Phillips, and Shuping Shi, 2025, Speculative bubbles
  in the recent {AI} boom: {N}asdaq and the {M}agnificent {S}even, {\em Journal
  of Time Series Analysis\/} 46, 814--828.

\bibitem[Bierens and Martins(2010)]{BierensMartins2010}
Bierens, Herman~J., and Luis~F. Martins, 2010, Time-varying cointegration, {\em
  Econometric Theory\/} 26, 1453--1490.

\bibitem[Blanchard and Watson(1982)]{Blanchard1982}
Blanchard, Olivier~J., and Mark~W. Watson, 1982, Bubbles, rational expectations
  and financial markets, in Paul Wachtel, ed., {\em Crises in the Economic and
  Financial Structure: Bubbles, Bursts, and Shocks\/},  295--316 (Lexington
  Books, Lexington, MA).

\bibitem[Bresnahan and Trajtenberg(1995)]{Bresnahan1995}
Bresnahan, Timothy~F., and Manuel Trajtenberg, 1995, General purpose
  technologies: ``{E}ngines of growth''?, {\em Journal of Econometrics\/} 65,
  83--108.

\bibitem[Brynjolfsson et~al.(2025)Brynjolfsson, Li, and
  Raymond]{BrynjolfssonLiRaymond2025}
Brynjolfsson, Erik, Danielle Li, and Lindsey~R. Raymond, 2025, Generative {AI}
  at work, {\em Quarterly Journal of Economics\/} 140, 889--942.

\bibitem[Campbell and Shiller(1988)]{Campbell1988}
Campbell, John~Y., and Robert~J. Shiller, 1988, The dividend-price ratio and
  expectations of future dividends and discount factors, {\em Review of
  Financial Studies\/} 1, 195--228.

\bibitem[Cochrane(2008)]{Cochrane2008}
Cochrane, John~H., 2008, The dog that did not bark: A defense of return
  predictability, {\em Review of Financial Studies\/} 21, 1533--1575.

\bibitem[Cochrane(2011)]{Cochrane2011}
Cochrane, John~H., 2011, Presidential address: Discount rates, {\em Journal of
  Finance\/} 66, 1047--1108.

\bibitem[Diba and Grossman(1988)]{Diba1988}
Diba, Behzad~T., and Herschel~I. Grossman, 1988, Explosive rational bubbles in
  stock prices?, {\em American Economic Review\/} 78, 520--530.

\bibitem[Evans(1991)]{Evans1991}
Evans, George~W., 1991, Pitfalls in testing for explosive bubbles in asset
  prices, {\em American Economic Review\/} 81, 922--930.

\bibitem[Fan et~al.(2026)Fan, Liu, Wang, and Zheng]{FanLiuWangZheng2026}
Fan, Jianqing, Qingfu Liu, Bo~Wang, and Kaixin Zheng, 2026, Unearthing
  financial statement fraud: Insights from news coverage analysis, {\em
  Management Science\/} 72, 4200--4230.

\bibitem[Fernald(2014)]{Fernald2014}
Fernald, John~G., 2014, A quarterly, utilization-adjusted series on total
  factor productivity, Working Paper 2012-19, Federal Reserve Bank of San
  Francisco.

\bibitem[Hansen(1992)]{Hansen1992}
Hansen, Bruce~E., 1992, Tests for parameter instability in regressions with
  {I(1)} processes, {\em Journal of Business \& Economic Statistics\/} 10,
  321--335.

\bibitem[Harvey et~al.(2015)Harvey, Leybourne, and Sollis]{Harvey2016}
Harvey, David~I., Stephen~J. Leybourne, and Robert Sollis, 2015, Recursive
  right-tailed unit root tests for an explosive asset price bubble, {\em
  Journal of Financial Econometrics\/} 13, 166--187.

\bibitem[Harvey et~al.(2016)Harvey, Leybourne, Sollis, and Taylor]{Harvey2020}
Harvey, David~I., Stephen~J. Leybourne, Robert Sollis, and A.~M.~Robert Taylor,
  2016, Tests for explosive financial bubbles in the presence of non-stationary
  volatility, {\em Journal of Empirical Finance\/} 38, 548--574.

\bibitem[Hsiao et~al.(2012)Hsiao, Ching, and Wan]{Hsiao2012}
Hsiao, Cheng, H.~Steve Ching, and Shui~Ki Wan, 2012, A panel data approach for
  program evaluation: Measuring the benefits of political and economic
  integration of {H}ong {K}ong with mainland {C}hina, {\em Journal of Applied
  Econometrics\/} 27, 705--740.

\bibitem[Jovanovic and Rousseau(2005)]{Jovanovic2005}
Jovanovic, Boyan, and Peter~L. Rousseau, 2005, General purpose technologies, in
  {\em Handbook of Economic Growth\/}, volume~1B,  1181--1224 (Elsevier).

\bibitem[Kogan et~al.(2017)Kogan, Papanikolaou, Seru, and Stoffman]{Kogan2017}
Kogan, Leonid, Dimitris Papanikolaou, Amit Seru, and Noah Stoffman, 2017,
  Technological innovation, resource allocation, and growth, {\em Quarterly
  Journal of Economics\/} 132, 665--712.

\bibitem[McElheran et~al.(2024)McElheran, Li, Brynjolfsson, Kroff, Dinlersoz,
  Foster, and Zolas]{McElheran2024}
McElheran, Kristina, J.~Frank Li, Erik Brynjolfsson, Zachary Kroff, Emin
  Dinlersoz, Lucia Foster, and Nikolas Zolas, 2024, {AI} adoption in {A}merica:
  Who, what, and where, {\em Journal of Economics \& Management Strategy\/} 33,
  375--415.

\bibitem[Monschang and Wilfling(2021)]{Monschang2021}
Monschang, Verena, and Bernd Wilfling, 2021, Sup-{ADF}-style bubble-detection
  methods under test, {\em Empirical Economics\/} 61, 145--172.

\bibitem[Pairolero et~al.(2025)Pairolero, Giczy, Torres, Islam~Erana,
  Finlayson, and Toole]{Pairolero2025}
Pairolero, Nicholas~A., Alexander~V. Giczy, Gerard Torres, Tisa Islam~Erana,
  Mark~A. Finlayson, and Andrew~A. Toole, 2025, The artificial intelligence
  patent dataset ({AIPD}) 2023 update, {\em Journal of Technology Transfer\/}
  50, 2587--2610.

\bibitem[P{\'a}stor and Veronesi(2006)]{Pastor2006}
P{\'a}stor, {\v{L}}ubo{\v{s}}, and Pietro Veronesi, 2006, Was there a {N}asdaq
  bubble in the late 1990s?, {\em Journal of Financial Economics\/} 81,
  61--100.

\bibitem[P{\'a}stor and Veronesi(2009)]{Pastor2009}
P{\'a}stor, {\v{L}}ubo{\v{s}}, and Pietro Veronesi, 2009, Technological
  revolutions and stock prices, {\em American Economic Review\/} 99,
  1451--1483.

\bibitem[Pavlidis et~al.(2016)Pavlidis, Yusupova, Paya, Peel,
  Mart{\'\i}nez-Garc{\'\i}a, Mack, and Grossman]{Pavlidis2016}
Pavlidis, Efthymios, Alisa Yusupova, Ivan Paya, David Peel, Enrique
  Mart{\'\i}nez-Garc{\'\i}a, Adrienne Mack, and Valerie Grossman, 2016,
  Episodes of exuberance in housing markets: In search of the smoking gun, {\em
  Journal of Real Estate Finance and Economics\/} 53, 419--449.

\bibitem[Phillips(1987)]{Phillips1987}
Phillips, Peter C.~B., 1987, Towards a unified asymptotic theory for
  autoregression, {\em Biometrika\/} 74, 535--547.

\bibitem[Phillips and Magdalinos(2007)]{PhillipsMagdalinos2007}
Phillips, Peter C.~B., and Tassos Magdalinos, 2007, Limit theory for moderate
  deviations from a unit root, {\em Journal of Econometrics\/} 136, 115--130.

\bibitem[Phillips et~al.(2015{\natexlab{a}})Phillips, Shi, and
  Yu]{Phillips2015b}
Phillips, Peter C.~B., Shuping Shi, and Jun Yu, 2015{\natexlab{a}}, Testing for
  multiple bubbles: Historical episodes of exuberance and collapse in the
  {S\&P}~500, {\em International Economic Review\/} 56, 1043--1078.

\bibitem[Phillips et~al.(2015{\natexlab{b}})Phillips, Shi, and
  Yu]{Phillips2015a}
Phillips, Peter C.~B., Shuping Shi, and Jun Yu, 2015{\natexlab{b}}, Testing for
  multiple bubbles: Limit theory of real-time detectors, {\em International
  Economic Review\/} 56, 1079--1134.

\bibitem[Phillips et~al.(2011)Phillips, Wu, and Yu]{Phillips2011}
Phillips, Peter C.~B., Yangru Wu, and Jun Yu, 2011, Explosive behavior in the
  1990s {NASDAQ}: When did exuberance escalate asset values?, {\em
  International Economic Review\/} 52, 201--226.

\bibitem[Shiller(2015)]{Shiller2015}
Shiller, Robert~J., 2015, {\em Irrational Exuberance\/}, third edition
  (Princeton University Press, Princeton, NJ).

\bibitem[Stock and Watson(1993)]{Stock1993}
Stock, James~H., and Mark~W. Watson, 1993, A simple estimator of cointegrating
  vectors in higher order integrated systems, {\em Econometrica\/} 61,
  783--820.

\bibitem[West(1988)]{West1988}
West, Kenneth~D., 1988, Bubbles, fads and stock price volatility tests: A
  partial evaluation, {\em Journal of Finance\/} 43, 639--656.

\end{thebibliography}
}


\clearpage
\begin{figure}
\centerline{\includegraphics[width=0.95\textwidth]{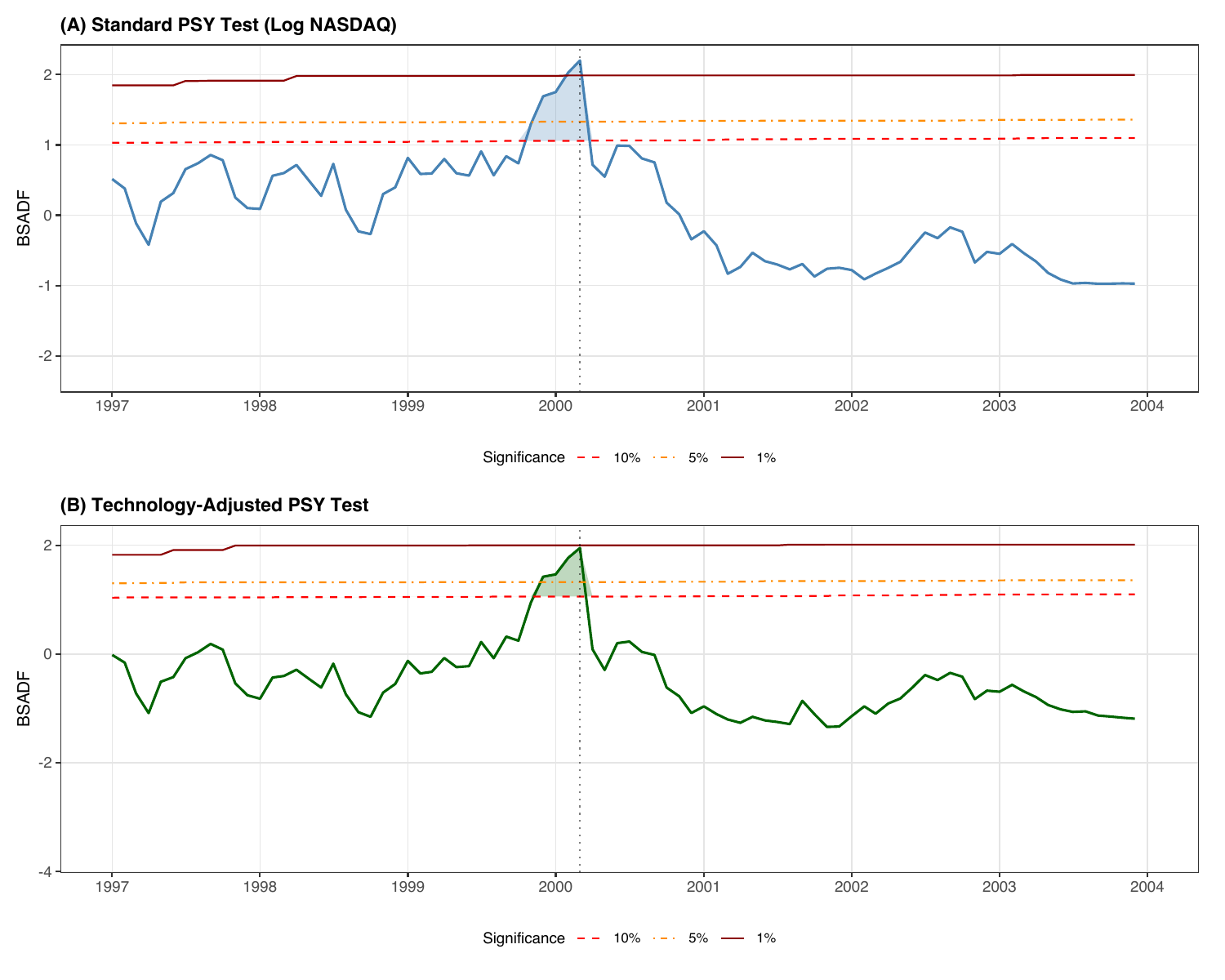}}
\noindent\caption{{\bf Technology adjustment removes the early false signal in the dot-com episode.} The horizontal axis is calendar time from 1997 to 2003, the solid line in each panel is the BSADF date-stamping statistic, and the dashed lines are the 10, 5, and 1 percent critical value sequences. The unadjusted panel detects May 1996 and November 1999 to March 2000, while the adjusted panel retains only December 1999 to March 2000.\label{fig:psy_empirical}}
\end{figure}

\clearpage
\begin{figure}
\centerline{\includegraphics[width=0.95\textwidth]{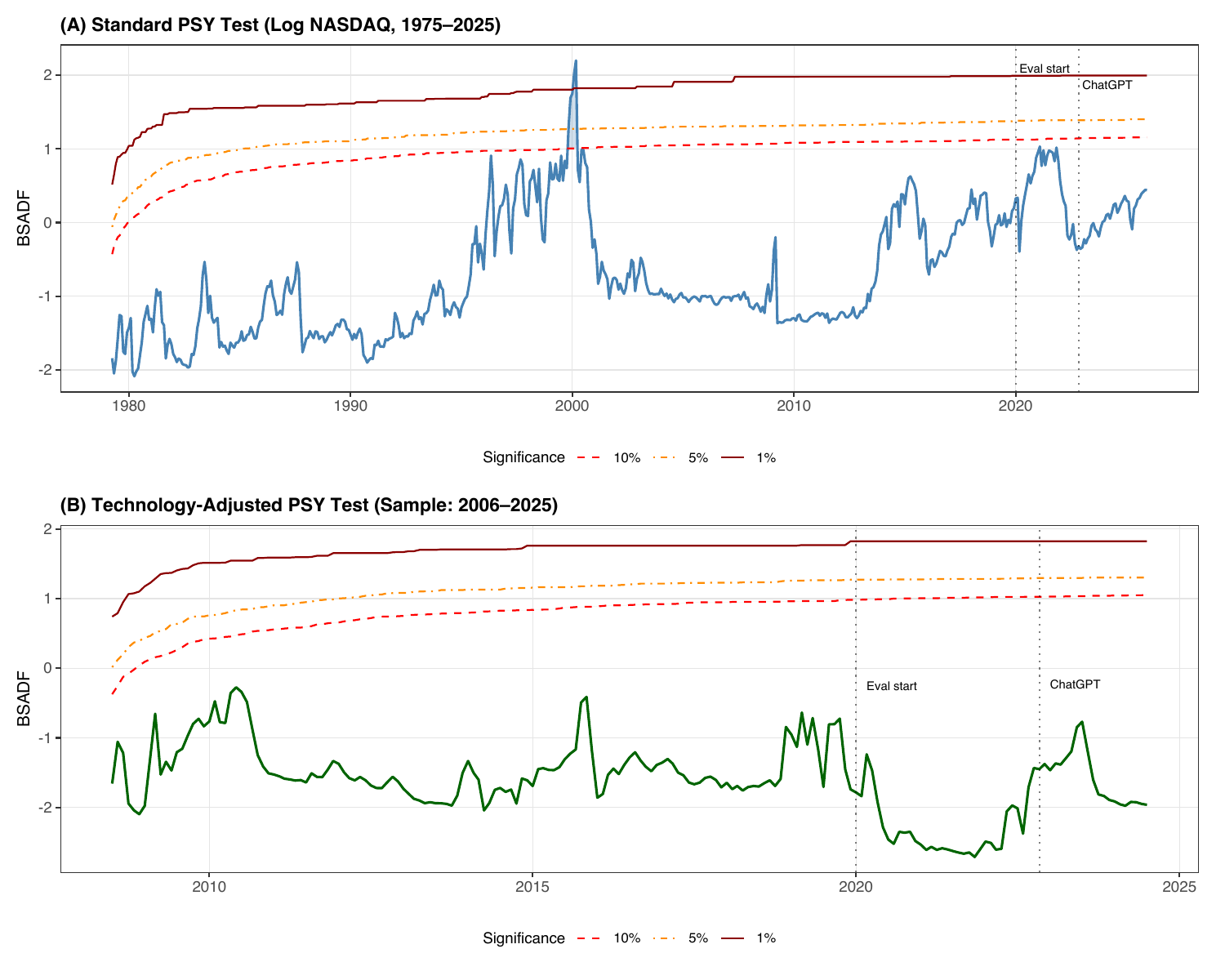}}
\noindent\caption{{\bf The extended-sample test still points to dot-com explosiveness, not to the 2020--2025 AI rally.} Panel A plots the unadjusted BSADF date-stamping statistic and its 10, 5, and 1 percent critical value sequences for the full 1975--2025 sample. Panel B plots the technology-adjusted date-stamping statistic and its 10, 5, and 1 percent critical value sequences for the 2006--2025 subsample. The key result is that only the dot-com peak crosses the threshold, while the adjusted AI-era series remains well below it.\label{fig:psy_ai}}
\end{figure}

\clearpage
\begin{figure}
\centering
\begin{subfigure}[t]{0.95\textwidth}
\centering
\includegraphics[width=\textwidth]{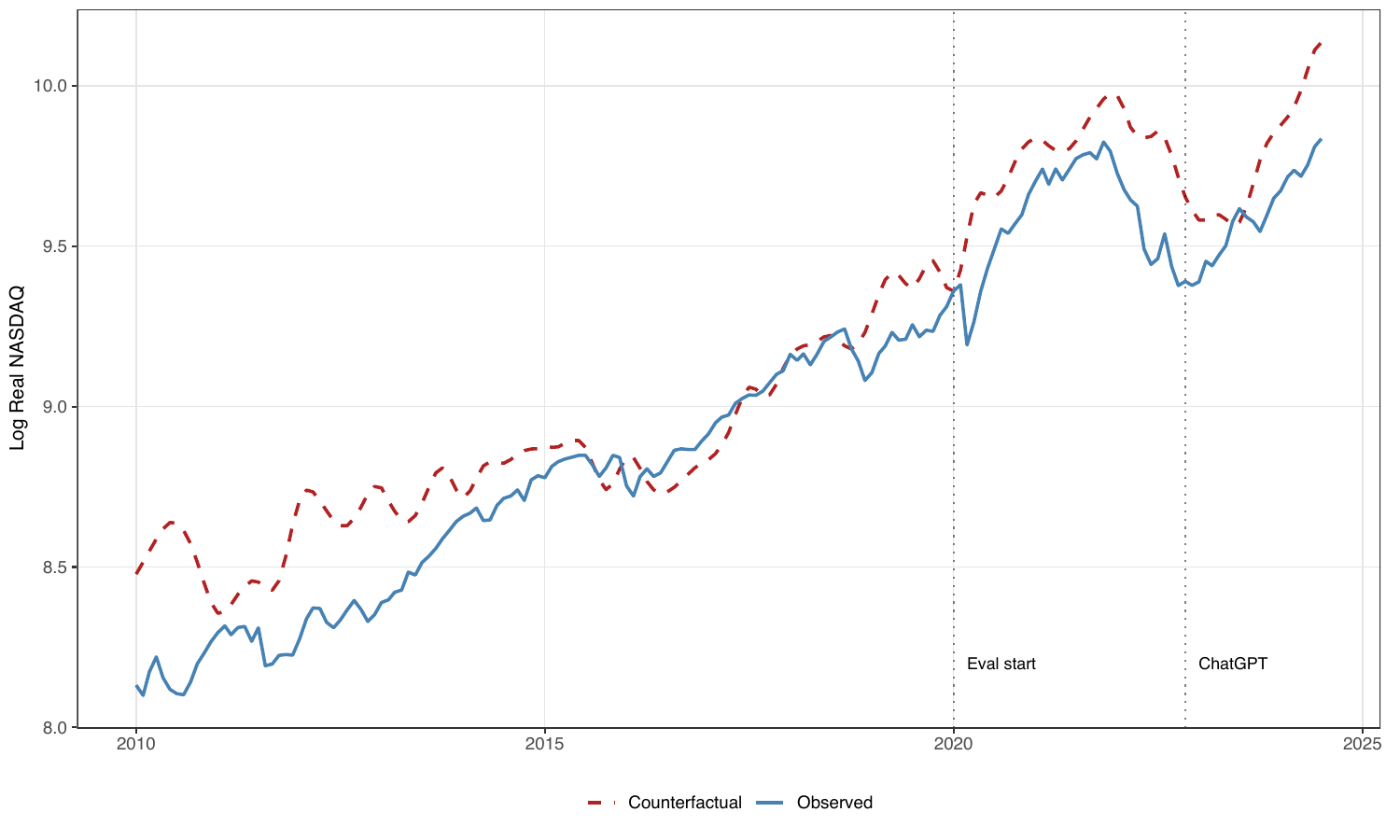}
\caption{Observed vs.\ counterfactual fundamental log NASDAQ price, 1975--2025. Solid line: observed log real NASDAQ Composite. Dashed line: DOLS counterfactual $\hat{f}_t^{\mathrm{AI}}$ estimated from the 2006--2019 training period using log IT investment, TFP, and log patent grants.}
\label{fig:counterfactual_ai}
\end{subfigure}

\vspace{1em}

\begin{subfigure}[t]{0.95\textwidth}
\centering
\includegraphics[width=\textwidth]{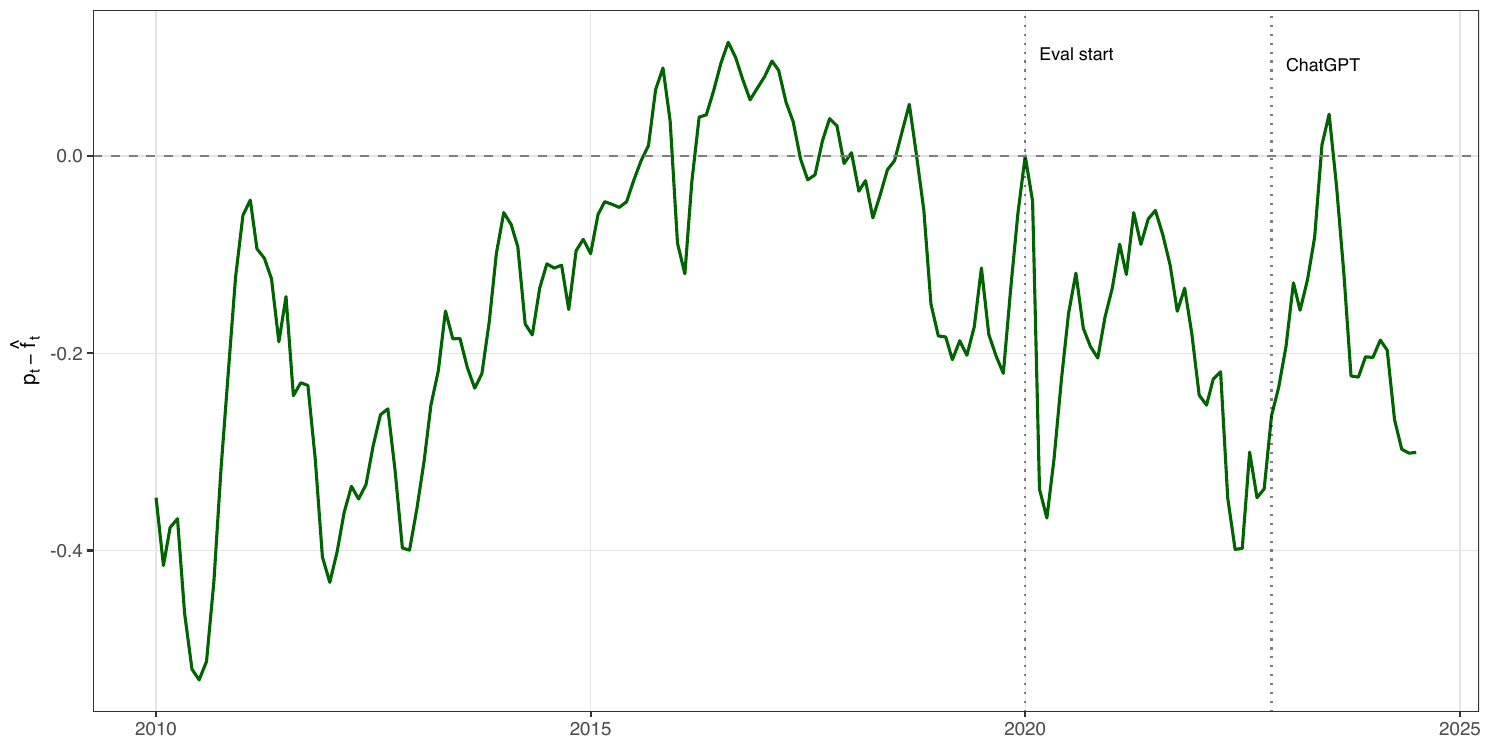}
\caption{Price gap $p_t - \hat{f}_t^{\mathrm{AI}}$ (observed minus counterfactual), 2020--2025. The technology-adjusted series to which PSY is applied. The mean gap of $-0.1877$ indicates that technology fundamentals grew faster than the observed NASDAQ price during the evaluation period.}
\label{fig:price_gap_ai}
\end{subfigure}
\noindent\caption{{\bf Counterfactual fundamental and price gap for the AI rally.} Panel~(a) compares observed and counterfactual prices; Panel~(b) plots the corresponding price gap. Together, the two panels show that technology-driven fundamentals outpaced the observed NASDAQ price during most of 2020--2025.\label{fig:counterfactual_and_gap_ai}}
\end{figure}

\clearpage
\begin{table}
\noindent\caption{{\bf The training-period cointegrating regression links NASDAQ valuations primarily to IT investment and patenting.} The table reports DOLS coefficient estimates for log real NASDAQ prices over 1975--1990, with one lead and one lag of differenced covariates and Newey--West standard errors. Log IT investment and log patent grants enter positively and significantly, while TFP is not significant in this specification. All nominal variables are deflated by CPI, and quarterly variables are interpolated to monthly frequency by cubic spline.\label{tab:coint_regression}}
\centering
\begin{tabular}{lcccc}
\toprule
\textbf{Variable} & \textbf{Coef.} & \textbf{Std. Err.} & \textbf{$t$-Stat.} & \textbf{$p$-Val.} \\
\midrule
Intercept & $-1.188$ & 1.477 & $-0.804$ & 0.423 \\
Log IT investment & 0.456 & 0.068 & 6.713 & $<0.001$ \\
TFP & $-0.000$ & 0.005 & $-0.088$ & 0.930 \\
Log patent grants & 0.513 & 0.191 & 2.693 & 0.008 \\
\midrule
$R^2$ & \multicolumn{4}{c}{0.835} \\
Engle--Granger ADF rejects at 1\% & \multicolumn{4}{c}{Yes} \\
Train.\ obs. & \multicolumn{4}{c}{192 months (Jan 1975--Dec 1990)} \\
\bottomrule
\end{tabular}
\end{table}

\clearpage
\begin{landscape}
\begin{table}
\noindent\caption{{\bf The dot-com result survives a broad set of alternative specifications, but the training window matters.} Panel A reports GSADF test statistics and rejection decisions at the 10, 5, and 1 percent levels across nine specifications. Panel B reports leave-one-out sensitivity after dropping one technology covariate at a time. The table shows that the technology-adjusted result remains concentrated in the late 1999 to 2000 episode, while some alternatives weaken or eliminate rejection. Critical values are computed by Monte Carlo with 2,000 replications, and the placebo row applies the test to the adjusted series during 1985--1990, a known non-bubble period.\label{tab:robustness_loo}}
\centering
\begin{tabular}{lccccccc}
\toprule
\textbf{Specification} & \textbf{GSADF} & \textbf{CV$_{0.10}$} & \textbf{CV$_{0.05}$} & \textbf{CV$_{0.01}$} & \textbf{Rej.\ 10\%} & \textbf{Rej.\ 5\%} & \textbf{Rej.\ 1\%} \\
\midrule
\multicolumn{8}{l}{\textit{Panel A: Specification Robustness}} \\
Unadjusted (standard PSY) & 2.199 & 1.915 & 2.175 & 2.660 & Yes & Yes & No \\
Baseline (technology-only) & 1.952 & 1.918 & 2.175 & 2.691 & Yes & No & No \\
Extended covariates (technology and macroeconomic) & 2.491 & 1.918 & 2.175 & 2.691 & Yes & Yes & No \\
First-difference & 2.174 & 1.930 & 2.175 & 2.681 & Yes & No & No \\
Lagged covariates ($h=3$) & 1.896 & 1.918 & 2.175 & 2.691 & No & No & No \\
Lagged covariates ($h=6$) & 1.487 & 1.930 & 2.175 & 2.681 & No & No & No \\
Placebo (1985--1990) & $-0.689$ & 1.334 & 1.623 & 2.403 & No & No & No \\
Alternative training (1975--1988) & 1.861 & 1.918 & 2.175 & 2.691 & No & No & No \\
Alternative training (1975--1992) & 1.596 & 1.918 & 2.175 & 2.691 & No & No & No \\
\midrule
\multicolumn{8}{l}{\textit{Panel B: Leave-One-Out}} \\
\textbf{Dropped Variable} & \textbf{GSADF} & \textbf{CV$_{0.10}$} & \textbf{CV$_{0.05}$} & \textbf{CV$_{0.01}$} & \textbf{Rej.\ 10\%} & \textbf{Rej.\ 5\%} & \textbf{Rej.\ 1\%} \\
Log IT investment & 0.613 & 1.918 & 2.175 & 2.691 & No & No & No \\
TFP & 2.012 & 1.918 & 2.175 & 2.691 & Yes & No & No \\
Log patent grants & 1.723 & 1.903 & 2.164 & 2.660 & No & No & No \\
\bottomrule
\end{tabular}
\end{table}
\end{landscape}

\clearpage
\begin{landscape}
\begin{table}
\noindent\caption{{\bf The AI-era non-rejection remains intact across all major specification changes.} Panel A reports GSADF statistics and rejection decisions at the 10, 5, and 1 percent levels across ten specifications. Panel B reports leave-one-out sensitivity after dropping one technology covariate at a time. The table shows that the post-2020 price dynamics remain consistent with technology fundamentals across all adjustments. Critical values are computed by Monte Carlo with 2,000 replications, and the placebo row applies the test to the adjusted series during 2010--2019, a known non-bubble period.\label{tab:robustness_ai}}
\centering
\begin{tabular}{lccccccc}
\toprule
\textbf{Specification} & \textbf{GSADF} & \textbf{CV$_{0.10}$} & \textbf{CV$_{0.05}$} & \textbf{CV$_{0.01}$} & \textbf{Rej.\ 10\%} & \textbf{Rej.\ 5\%} & \textbf{Rej.\ 1\%} \\
\midrule
\multicolumn{8}{l}{\textit{Panel A: Specification Robustness}} \\
Unadjusted (standard PSY, 1975--2025) & 2.199 & 1.996 & 2.239 & 2.759 & Yes & No & No \\
Baseline (technology-only, train 2006--2019) & $-0.275$ & 1.810 & 2.046 & 2.599 & No & No & No \\
Extended covariates (technology and macroeconomic) & 2.060 & 1.992 & 2.239 & 2.711 & Yes & No & No \\
First-difference & 1.652 & 1.996 & 2.246 & 2.711 & No & No & No \\
Lagged covariates ($h=3$) & 1.527 & 1.992 & 2.239 & 2.711 & No & No & No \\
Lagged covariates ($h=6$) & 0.659 & 1.992 & 2.239 & 2.711 & No & No & No \\
Placebo (2010--2019) & 0.053 & 1.687 & 1.970 & 2.540 & No & No & No \\
Alternative training (2004--2019) & 0.431 & 1.992 & 2.239 & 2.711 & No & No & No \\
Alternative training (2008--2019) & 0.021 & 1.992 & 2.239 & 2.711 & No & No & No \\
With research and development (train 2006--2019) & $-1.617$ & 1.955 & 2.233 & 2.660 & No & No & No \\
\midrule
\multicolumn{8}{l}{\textit{Panel B: Leave-One-Out}} \\
\textbf{Dropped Variable} & \textbf{GSADF} & \textbf{CV$_{0.10}$} & \textbf{CV$_{0.05}$} & \textbf{CV$_{0.01}$} & \textbf{Rej.\ 10\%} & \textbf{Rej.\ 5\%} & \textbf{Rej.\ 1\%} \\
Log IT investment & 0.610 & 1.992 & 2.239 & 2.711 & No & No & No \\
TFP & 0.475 & 1.992 & 2.239 & 2.711 & No & No & No \\
Log patent grants & 1.114 & 1.996 & 2.239 & 2.759 & No & No & No \\
\bottomrule
\end{tabular}
\end{table}
\end{landscape}


\clearpage
\appendix
\begin{center}
  {\bf \Large Internet Appendix}
\end{center}

\pagenumbering{arabic}
\renewcommand*{\thepage}{Appendix-\arabic{page}}
 
\counterwithin*{table}{section}
\counterwithin*{figure}{section}
\counterwithin*{equation}{section}
\setcounter{table}{0}
\setcounter{figure}{0}
\setcounter{equation}{0}
\renewcommand{\thetable}{\Alph{section}.\arabic{table}}
\renewcommand{\thefigure}{\Alph{section}.\arabic{figure}}
\renewcommand{\theequation}{\Alph{section}.\arabic{equation}}

\counterwithin*{assumption}{section}
\counterwithin*{proposition}{section}
\counterwithin*{theorem}{section}
\counterwithin*{lemma}{section}
\counterwithin*{corollary}{section}
\counterwithin*{definition}{section}
\counterwithin*{remark}{section}
\counterwithin*{condition}{section}
\counterwithin*{obs}{section}
\counterwithin*{hypothesis}{section}
\setcounter{assumption}{0}
\setcounter{proposition}{0}
\setcounter{theorem}{0}
\setcounter{lemma}{0}
\setcounter{corollary}{0}
\setcounter{definition}{0}
\setcounter{remark}{0}
\renewcommand{\theassumption}{\Alph{section}.\arabic{assumption}}
\renewcommand{\theproposition}{\Alph{section}.\arabic{proposition}}
\renewcommand{\thetheorem}{\Alph{section}.\arabic{theorem}}
\renewcommand{\thelemma}{\Alph{section}.\arabic{lemma}}
\renewcommand{\thecorollary}{\Alph{section}.\arabic{corollary}}
\renewcommand{\thedefinition}{\Alph{section}.\arabic{definition}}
\renewcommand{\theremark}{\Alph{section}.\arabic{remark}}
\renewcommand{\thecondition}{\Alph{section}.\arabic{condition}}
\renewcommand{\theobs}{\Alph{section}.\arabic{obs}}
\renewcommand{\thehypothesis}{\Alph{section}.\arabic{hypothesis}}

\section{Microfoundation: A Production Economy with Innovation}
\label{app:microfoundation}

This appendix provides the production-economy microfoundation for the technology-adjusted PSY diagnostic summarized in Section~\ref{subsec:microfoundation}. We derive the dividend-growth representation $\Delta d_t = c + \boldsymbol{\gamma}'\mathbf{X}_t + \varepsilon_t$, the closed-form expression for the technology present-value term $\Tt = \boldsymbol{\beta}'\mathbf{X}_t + \eta_t$, and the benchmark identification conditions (R1)--(R3) under which the bubble component is orthogonal to the observable technology state $\mathbf{X}_t = (TFP_t, \log IT_t, \log Pat_t)'$.

\subsection{Environment and Preferences}
\label{subsubsec:env_prefs}

Time is discrete, $t=0,1,2,\dots$. The economy is populated by a unit mass of ex ante identical, infinitely lived households who consume a single final good, supply labor inelastically, and trade shares in a representative firm and a one-period riskless bond. Let $C_t$ denote aggregate consumption and $\mathcal{F}_t$ the information set available at date $t$. Households have preferences over consumption streams represented by the stochastic discount factor (SDF) $\{M_{t,t+j}\}_{j\geq 0}$, with one-period SDF $M_{t+1} \equiv M_{t,t+1}$ satisfying the Euler equations
\begin{equation}
\mathbb{E}_t\!\left[M_{t+1}\,R_{t+1}^{i}\right] = 1, \qquad i \in \{\text{equity},\text{bond}\},
\label{eq:euler_micro}
\end{equation}
for each traded asset with gross return $R_{t+1}^{i}$. We assume throughout this subsection that $M_{t+1} > 0$ almost surely and that $\mathbb{E}_t[M_{t+1}]$ is a strictly positive, $\mathcal{F}_t$-measurable random variable bounded away from $0$ and $\infty$. Under the log-linear approximation of Campbell and Shiller (1988) that we invoke in Section~\ref{subsubsec:loglin}, \eqref{eq:euler_micro} collapses to the constant expected log-return condition $\mathbb{E}_t[r_{t+1}] = \bar r$ used in Assumption~\ref{ass:const_returns}, so the two layers are consistent. The standard transversality condition
\begin{equation*}
\lim_{j\to\infty} \mathbb{E}_t\!\left[M_{t,t+j}\, P_{t+j}\right] = 0
\end{equation*}
is imposed on the fundamental component; the bubble component $b_t$, by definition, is the solution of the martingale equation $\mathbb{E}_t[M_{t+1} b_{t+1}] = b_t$ that violates transversality when $b_t \neq 0$.

\subsection{Production Technology}
\label{subsubsec:production}

A representative firm produces output using a constant-returns Cobb--Douglas technology that combines physical capital $K_t$, information-technology (IT) capital $G_t$, and labor $L_t$ with disembodied total factor productivity $A_t$:
\begin{equation}
Y_t \;=\; A_t\, K_t^{\alpha}\, G_t^{\eta}\, L_t^{1-\alpha-\eta},\qquad \alpha, \eta \in (0,1),\quad \alpha + \eta < 1.
\label{eq:production}
\end{equation}
The two capital stocks evolve according to
\begin{equation}
K_{t+1} = (1-\delta_K)\,K_t + I_t^{K}, \qquad
G_{t+1} = (1-\delta_G)\,G_t + I_t^{IT},
\label{eq:capital_laws}
\end{equation}
where $\delta_K, \delta_G \in (0,1)$ are depreciation rates and $I_t^{K}, I_t^{IT} \geq 0$ are gross investments. Disembodied productivity follows the process
\begin{equation}
\Delta a_{t+1} \;=\; \mu_A \;+\; \alpha_q\, q_t \;+\; \alpha_I\, i_t^{IT} \;+\; u_{t+1}^{A},\qquad \alpha_q,\alpha_I > 0,
\label{eq:tfp_law}
\end{equation}
where $a_t \equiv \log A_t$, $i_t^{IT} \equiv \log I_t^{IT}$, $q_t$ is a patent-quality index (defined below), and $u_{t+1}^{A}$ is a mean-zero innovation orthogonal to $\mathcal{F}_t$ with finite variance. The coefficients $(\alpha_q,\alpha_I)$ parameterize the two complementary channels through which innovation lifts TFP: the \emph{invention} margin (quality of newly granted patents) and the \emph{diffusion} margin (adoption and use of IT capital).

We assume a patent-grant lag $\ell \geq 1$: if $Q_t$ denotes the true (unobserved) innovation-quality shock at date $t$, the granted patent index $Pat_{t+\ell}$ observed at date $t+\ell$ is a deterministic function of $\{Q_s\}_{s\leq t}$ and of pre-$t$ variables only. This makes $Pat_t$ a predetermined variable with respect to $\mathcal{F}_{t-1}$, a fact we use in Section~\ref{subsubsec:identification} to rule out contemporaneous feedback from the bubble into the patent measure. The \emph{observable technology vector} is
\begin{equation*}
\mathbf{X}_t \;\equiv\; (TFP_t,\; \log IT_t,\; \log Pat_t)',
\end{equation*}
with $TFP_t$ the Solow residual constructed from \eqref{eq:production}, $IT_t = I_t^{IT}$, and $Pat_t$ the granted-patent index.

\subsection{Firm Problem and First-Order Conditions}
\label{subsubsec:firm}

The firm chooses contingent plans for labor, physical investment, IT investment, and innovation spending to maximize the present value of dividends paid to shareholders under the SDF $M_{t,t+j}$:
\begin{equation}
\max_{\{L_t, I_t^{K}, I_t^{IT}, q_t\}} \;\; \mathbb{E}_0 \sum_{t=0}^{\infty} M_{0,t}\, D_t,
\label{eq:firm_problem}
\end{equation}
subject to \eqref{eq:production}--\eqref{eq:tfp_law}, where the dividend (residual cash flow) is
\begin{equation}
D_t \;=\; Y_t \;-\; w_t\, L_t \;-\; I_t^{K} \;-\; I_t^{IT} \;-\; \Phi\!\left(I_t^{IT}/G_t\right)\,G_t \;-\; \Xi(q_t).
\label{eq:dividend_identity}
\end{equation}
Here $w_t$ is the real wage; $\Phi(\cdot)$ is a $C^{2}$, strictly convex IT-investment adjustment-cost function normalized so that $\Phi(\delta_G)=0$ and $\Phi'(\delta_G)=0$ at the steady-state rate $\delta_G$; and $\Xi(\cdot)$ is a $C^{2}$, strictly convex innovation cost. Attaching Lagrange multipliers to the two capital laws in \eqref{eq:capital_laws} and differentiating, the firm's first-order conditions are:

\emph{Labor.} $\;(1-\alpha-\eta)\,Y_t/L_t = w_t$.

\emph{Physical capital (standard $Q$-equation).} $\;1 = \mathbb{E}_t\!\left[M_{t+1}\left(\alpha\, Y_{t+1}/K_{t+1} + 1 - \delta_K\right)\right]$.

\emph{IT capital (the key Euler equation).}
\begin{equation}
1 \;+\; \Phi'\!\left(I_t^{IT}/G_t\right) \;=\; \mathbb{E}_t\!\left[M_{t+1}\left(\eta\,Y_{t+1}/G_{t+1} \;+\; 1 - \delta_G\right)\right].
\label{eq:it_foc}
\end{equation}

\emph{Innovation quality.} $\;\Xi'(q_t) = \mathbb{E}_t\!\left[\sum_{j=1}^{\infty} M_{t,t+j}\, \alpha_q\, Y_{t+j}/A_{t+j}\right]$, i.e., the marginal cost of innovation equals the expected discounted marginal product of the resulting patent quality operating through \eqref{eq:tfp_law}.

Equation~\eqref{eq:it_foc} is the central identifying restriction of the microfoundation: in the benchmark model it ties IT investment to the \emph{real} marginal product of IT capital, $\eta Y_{t+1}/G_{t+1}$, and to the SDF---not to the secondary-market price of equity, and hence not to the bubble $b_t$. This separability is what permits us in Section~\ref{subsubsec:identification} to exclude $b_t$ from the structural regression of $f_t$ on $\mathbf{X}_t$ under the benchmark restrictions.

\subsection{Balanced-Growth Path}
\label{subsubsec:bgp}

A deterministic balanced-growth path (BGP) is a sequence $\{A_t^{\ast}, K_t^{\ast}, G_t^{\ast}, L_t^{\ast}, Y_t^{\ast}, D_t^{\ast}, w_t^{\ast}\}$ along which $A_t^{\ast}, K_t^{\ast}, G_t^{\ast}, Y_t^{\ast}, D_t^{\ast}, w_t^{\ast}$ grow at a common gross rate $e^{\mu_A/(1-\alpha-\eta)}$ and $L_t^{\ast}$ is constant, with innovation quality $q_t^{\ast}$ and IT-investment rate $I_t^{IT,\ast}/G_t^{\ast} = \delta_G$ both constant. Existence and uniqueness of such a BGP under \eqref{eq:production}--\eqref{eq:tfp_law} and standard Inada-type conditions on $\Phi$ and $\Xi$ is a routine application of the arguments in King, Plosser, and Rebelo (1988); we take it as given. The BGP pins down the steady-state expenditure shares
\begin{equation*}
s_Y \equiv \frac{Y^{\ast}}{D^{\ast}},\quad s_w \equiv \frac{w^{\ast} L^{\ast}}{D^{\ast}},\quad s_K \equiv \frac{I^{K,\ast}}{D^{\ast}},\quad s_{IT} \equiv \frac{I^{IT,\ast}}{D^{\ast}},\quad s_\Phi \equiv \frac{\Phi(\delta_G)\, G^{\ast}}{D^{\ast}},\quad s_\Xi \equiv \frac{\Xi(q^{\ast})}{D^{\ast}},
\end{equation*}
which satisfy the accounting identity $s_Y - s_w - s_K - s_{IT} - s_\Phi - s_\Xi = 1$ (with $s_\Phi = 0$ because $\Phi(\delta_G)=0$ by normalization). Logarithmic deviations from trend are denoted $\widehat{x}_t \equiv \log X_t - \log X_t^{\ast}$ for any variable $X_t$.

\subsection{Log-Linearization of the Dividend Process}
\label{subsubsec:loglin}

We now derive the dividend representation rigorously by log-linearizing \eqref{eq:dividend_identity} around the BGP. Step 1. Take logarithms of \eqref{eq:production}:
\begin{equation}
y_t \;=\; a_t \;+\; \alpha\, k_t \;+\; \eta\, g_t \;+\; (1-\alpha-\eta)\,\ell_t,
\label{eq:log_production}
\end{equation}
so that $\widehat{y}_t = \widehat{a}_t + \alpha\widehat{k}_t + \eta\widehat{g}_t + (1-\alpha-\eta)\widehat{\ell}_t$. Step 2. First-order Taylor expansion of \eqref{eq:dividend_identity} around the BGP yields
\begin{equation}
\widehat{d}_t \;=\; s_Y\,\widehat{y}_t \;-\; s_w\,(\widehat{w}_t + \widehat{\ell}_t) \;-\; s_K\,\widehat{I}_t^{K} \;-\; s_{IT}\,\widehat{I}_t^{IT} \;-\; s_\Phi\,(\widehat{I}_t^{IT} - \widehat{g}_t) \;-\; s_\Xi\,\widehat{q}_t.
\label{eq:d_linearized}
\end{equation}
Here the $s_\Phi$ term arises from differentiating $\Phi(I_t^{IT}/G_t)\,G_t$ and using the BGP normalization $\Phi'(\delta_G)=0$ (so the coefficient on $\widehat{I}_t^{IT}-\widehat{g}_t$ enters only through the second-order term, which we collect into the residual if $\Phi''(\delta_G)>0$; under the first-order approximation $s_\Phi$ vanishes at the BGP and we retain the placeholder for bookkeeping). Step 3. First-difference \eqref{eq:d_linearized} and substitute: (a) \eqref{eq:log_production} for $\Delta\widehat{y}_t$; (b) the TFP law~\eqref{eq:tfp_law} for $\Delta a_{t+1}$, which expresses $\Delta a_t$ as an affine function of $q_{t-1}$ and $i_{t-1}^{IT}$; (c) the IT-investment Euler equation~\eqref{eq:it_foc}, which under log-linearization ties $\widehat{I}_t^{IT} - \widehat{g}_t$ to the expected real marginal product of IT, $\eta\,\widehat{y}_{t+1} - \widehat{g}_{t+1}$, and hence to the TFP component $TFP_{t+1}$; and (d) the labor FOC $(1-\alpha-\eta)Y_t/L_t = w_t$, which after log-linearization eliminates $\widehat{w}_t + \widehat{\ell}_t$ in favor of $\widehat{y}_t$. Step 4. Collecting all resulting terms and projecting onto the span of $\mathbf{X}_t = (TFP_t,\log IT_t,\log Pat_t)'$ yields the affine representation
\begin{equation}
\Delta d_t \;=\; c \;+\; \boldsymbol{\gamma}'\,\mathbf{X}_t \;+\; \varepsilon_t,
\label{eq:div_micro}
\end{equation}
where $\boldsymbol{\gamma} \equiv (\gamma_{TFP},\gamma_{IT},\gamma_{Pat})'$ is a vector of reduced-form elasticities with $\gamma_{TFP}, \gamma_{IT}, \gamma_{Pat} > 0$ under the sign conventions $\alpha_q,\alpha_I>0$ and $\alpha,\eta\in(0,1)$, the constant $c$ absorbs the deterministic BGP growth $\mu_A/(1-\alpha-\eta)$ and any steady-state contributions from $\boldsymbol{\gamma}'\mathbb{E}[\mathbf{X}_t]$ (so $\mathbf{X}_t$ in~\eqref{eq:div_micro} is understood as demeaned), and $\varepsilon_t$ collects (i) the TFP innovation $u_t^{A}$, (ii) the log-linearization remainder of order $O_p(\Vert\widehat{\mathbf{x}}_t\Vert^{2})$, and (iii) any projection residual orthogonal to $\mathbf{X}_t$. By construction, $\mathbb{E}[\varepsilon_t \mid \mathbf{X}_t] = 0$.

\subsection{Technology-State VAR and Closed-Form $\mathcal{T}_t$}
\label{subsubsec:var}

To obtain a closed form for the technology present-value term $\mathcal{T}_t = \sum_{j=0}^{\infty}\rho^{\,j}\,\delta_{t+1+j}$, we impose the following reduced-form dynamics on the observable state.

\begin{assumption}[Technology-State VAR]
\label{ass:var}
The observable technology vector $\mathbf{X}_t$ follows a first-order VAR,
\[
\mathbf{X}_{t+1} = \boldsymbol{\Phi}\,\mathbf{X}_t + \mathbf{v}_{t+1},
\]
where $\{\mathbf{v}_{t+1},\mathcal{F}_t\}$ is a strictly stationary, square-integrable martingale-difference sequence satisfying
\[
\mathbb{E}[\mathbf{v}_{t+1}\mid \mathcal{F}_t]=\mathbf{0},
\qquad
\mathbb{E}\|\mathbf{v}_{t+1}\|^{2}<\infty,
\qquad
\mathbb{E}[\mathbf{v}_{t+1}\mathbf{v}_{t+1}']=\boldsymbol{\Sigma}_v,
\]
with $\boldsymbol{\Sigma}_v$ finite, positive semi-definite, and time-invariant. The spectral radius of $\boldsymbol{\Phi}$ satisfies $\rho(\boldsymbol{\Phi}) < \rho^{-1}$, where $\rho \in (0,1)$ is the Campbell--Shiller log-linearization constant.
\end{assumption}

Under Assumption~\ref{ass:var}, we derive a closed-form expression for $\mathcal{T}_t$ by decomposing each future $\mathbf{X}_{t+1+j}$ into its conditional mean and the cumulated innovation noise, and then summing the two components separately. Substituting \eqref{eq:div_micro} (which gives $\delta_{t+1+j} = \boldsymbol{\gamma}'\mathbf{X}_{t+1+j}$) into the definition of $\mathcal{T}_t$ and adding and subtracting the conditional expectation,
\begin{align*}
\mathcal{T}_t
&\;=\; \sum_{j=0}^{\infty} \rho^{\,j}\,\delta_{t+1+j}
\;=\; \sum_{j=0}^{\infty} \rho^{\,j}\,\boldsymbol{\gamma}'\,\mathbf{X}_{t+1+j} \\
&\;=\; \sum_{j=0}^{\infty} \rho^{\,j}\,\boldsymbol{\gamma}'\,\Big(\mathbb{E}_t[\mathbf{X}_{t+1+j}] \;+\; \big(\mathbf{X}_{t+1+j} - \mathbb{E}_t[\mathbf{X}_{t+1+j}]\big)\Big) \\
&\;=\; \underbrace{\sum_{j=0}^{\infty} \rho^{\,j}\,\boldsymbol{\gamma}'\,\mathbb{E}_t[\mathbf{X}_{t+1+j}]}_{\text{predictable part}}
\;+\; \underbrace{\sum_{j=0}^{\infty} \rho^{\,j}\,\boldsymbol{\gamma}'\,\big(\mathbf{X}_{t+1+j} - \mathbb{E}_t[\mathbf{X}_{t+1+j}]\big)}_{\text{innovation noise}}.
\end{align*}
Iterating the VAR gives $\mathbb{E}_t[\mathbf{X}_{t+1+j}] = \boldsymbol{\Phi}^{\,j+1}\mathbf{X}_t$ for each $j \geq 0$, while iterated substitution of $\mathbf{X}_{s+1} = \boldsymbol{\Phi}\mathbf{X}_s + \mathbf{v}_{s+1}$ yields the moving-average representation $\mathbf{X}_{t+1+j} - \mathbb{E}_t[\mathbf{X}_{t+1+j}] = \sum_{l=0}^{j}\boldsymbol{\Phi}^{\,l}\,\mathbf{v}_{t+1+j-l}$. Substituting both into the display above,
\begin{equation*}
\mathcal{T}_t \;=\; \sum_{j=0}^{\infty} \rho^{\,j}\,\boldsymbol{\gamma}'\,\boldsymbol{\Phi}^{\,j+1}\,\mathbf{X}_t
\;+\; \sum_{j=0}^{\infty} \rho^{\,j}\,\boldsymbol{\gamma}'\,\Bigg(\sum_{l=0}^{j}\boldsymbol{\Phi}^{\,l}\,\mathbf{v}_{t+1+j-l}\Bigg).
\end{equation*}
The first (predictable) sum is a matrix geometric series that converges because $\rho(\rho\boldsymbol{\Phi}) = \rho\cdot\rho(\boldsymbol{\Phi}) < 1$ by Assumption~\ref{ass:var}, and evaluates to
\begin{equation*}
\sum_{j=0}^{\infty}\rho^{\,j}\,\boldsymbol{\gamma}'\,\boldsymbol{\Phi}^{\,j+1}\,\mathbf{X}_t
\;=\; \boldsymbol{\gamma}'\,\boldsymbol{\Phi}\sum_{j=0}^{\infty}(\rho\boldsymbol{\Phi})^{\,j}\,\mathbf{X}_t
\;=\; \boldsymbol{\gamma}'\,\boldsymbol{\Phi}\,(\mathbf{I}-\rho\boldsymbol{\Phi})^{-1}\,\mathbf{X}_t.
\end{equation*}
Hence we obtain the decomposition
\begin{equation}
\mathcal{T}_t \;=\; \boldsymbol{\beta}'\,\mathbf{X}_t \;+\; \eta_t, \qquad \boldsymbol{\beta} \;\equiv\; (\mathbf{I}-\rho\boldsymbol{\Phi})^{-\top}\boldsymbol{\Phi}^{\top}\boldsymbol{\gamma},
\label{eq:Tt_closed_form}
\end{equation}
where the stationary innovation term
\begin{equation}
\eta_t \;\equiv\; \sum_{j=0}^{\infty} \rho^{\,j}\,\boldsymbol{\gamma}'\,\Bigg(\sum_{l=0}^{j}\boldsymbol{\Phi}^{\,l}\,\mathbf{v}_{t+1+j-l}\Bigg)
\label{eq:eta_t_def}
\end{equation}
is a mean-zero, $\rho$-discounted weighted sum of the VAR innovations $\{\mathbf{v}_{s}\}_{s\geq t+1}$. Under Assumption~\ref{ass:var} the process $\{\eta_t\}$ is stationary and has finite variance: applying the submultiplicative matrix norm $\|\cdot\|$ and the spectral-radius bound $\rho(\boldsymbol{\Phi}) < \rho^{-1}$ gives $\|\boldsymbol{\Phi}^{\,l}\| \leq C_0\,r^{\,l}$ for some $r < \rho^{-1}$, so the double sum in~\eqref{eq:eta_t_def} is absolutely convergent in mean square.

The predictable part $\boldsymbol{\beta}'\mathbf{X}_t$ is the economically meaningful component of $\mathcal{T}_t$: it is the projection of $\mathcal{T}_t$ onto the information set spanned by the observable state $\mathbf{X}_t$ and is what the empirical cointegrating regression identifies. The innovation component $\eta_t$ is orthogonal to $\mathbf{X}_t$ (by construction, as a sum of future innovations $\mathbf{v}_{t+1+j-l}$ that are $\mathcal{F}_t$-unpredictable) and is therefore absorbed into the regression residual when \eqref{eq:Tt_closed_form} is combined with Proposition~\ref{prop:fundamental}.

Adding the bubble $b_t$ to the fundamental price and using Proposition~\ref{prop:fundamental} delivers the structural cointegrating regression
\begin{equation}
p_t \;=\; f_t + b_t \;=\; \beta_0 \;+\; \boldsymbol{\beta}_X'\,\mathbf{X}_t \;+\; u_t, \qquad u_t \;=\; b_t \;+\; \eta_t \;+\; \zeta_t,
\label{eq:coint_structural}
\end{equation}
where $\beta_0 = C$ is the Campbell--Shiller constant, $\boldsymbol{\beta}_X$ collects $\boldsymbol{\beta}$ from~\eqref{eq:Tt_closed_form} together with the contemporaneous loading of $d_t$ on $\mathbf{X}_t$ implied by~\eqref{eq:div_micro}, and the residual $u_t$ pools three orthogonal components in the benchmark model: the bubble $b_t$, the stationary VAR-innovation term $\eta_t$ from~\eqref{eq:eta_t_def}, and a further $\mathcal{F}_t$-measurable noise $\zeta_t \perp \mathbf{X}_t$ that collects the dividend-growth residual $\varepsilon_t$ from~\eqref{eq:div_micro} and higher-order log-linearization terms. Equation~\eqref{eq:coint_structural} is the structural counterpart of the empirical cointegrating regression estimated in the technology-adjusted PSY diagnostic; the explosive content of $b_t$ is therefore recovered from the residual $\widehat u_t$ only under identification conditions (R1)--(R3) below. When speculative valuations affect $\mathbf{X}_t$, the residual instead contains the attenuated component described in Section~\ref{subsec:feedback_bounds}.

\subsection{Identification of the Cointegrating Regression}
\label{subsubsec:identification}

Three benchmark identification requirements are sufficient for \eqref{eq:coint_structural} to deliver a consistent estimate of $\boldsymbol{\beta}_X$ despite the presence of the bubble $b_t$ in the residual $u_t$.

\emph{(R1) Separability of the production and valuation blocks.} The bubble enters the share price $p_t = f_t + b_t$ additively and satisfies the martingale equation $\mathbb{E}_t[M_{t+1} b_{t+1}] = b_t$, but does not enter the production function~\eqref{eq:production}, the capital laws~\eqref{eq:capital_laws}, the TFP law~\eqref{eq:tfp_law}, or any of the firm's first-order conditions in Section~\ref{subsubsec:firm}. In this benchmark model, $\mathbf{X}_t$ is generated by a dynamic system that is \emph{block-separable} from $b_t$, so $\mathrm{Cov}(\mathbf{X}_t, b_t) = 0$ in population.

\emph{(R2) Predetermined patent index.} The patent-grant lag $\ell \geq 1$ implies that $Pat_t$ is $\mathcal{F}_{t-\ell}$-measurable. Even if shareholders' speculative exuberance affects innovation \emph{financing} contemporaneously, the resulting change in realized patent grants is observed only at date $t+\ell$ and enters $\mathbf{X}_{t+\ell}$, not $\mathbf{X}_t$. Thus patent grants help with contemporaneous feedback. They do not, by themselves, rule out multi-period feedback from a sustained valuation boom into future patenting.

\emph{(R3) Real IT Euler equation.} The IT-investment FOC~\eqref{eq:it_foc} ties $I_t^{IT}$ to the expected real marginal product of IT capital under the SDF, not to the market price of the firm's equity. Formally, \eqref{eq:it_foc} depends on $(Y_{t+1}, G_{t+1}, M_{t+1})$ only; since none of these objects is a function of $b_t$ under (R1), contemporaneous IT investment is orthogonal to the bubble in the benchmark model.

Under (R1)--(R3), OLS applied to \eqref{eq:coint_structural} consistently estimates $\boldsymbol{\beta}_X$ in the absence of an explosive bubble, and the regression residual $\widehat u_t$ inherits any explosive component of $b_t$---the feature that the technology-adjusted PSY diagnostic is designed to detect. If these restrictions fail because speculative valuations affect $\mathbf{X}_t$, the residual inherits only the non-absorbed component of the bubble, as shown in Section~\ref{subsec:feedback_bounds}. The empirical analysis should therefore be interpreted as a conditional exercise rather than as unconditional structural identification of $b_t$.

\begin{remark}[Theoretical vs.\ Empirical Technology Components]
\label{rmk:theo_vs_emp}
The theoretical analysis in Sections~\ref{subsec:tech_dividends}--\ref{subsec:explosive} uses a deterministic technology component $\delta_t$ (Assumption~\ref{ass:tech_shock}) for analytical tractability: the hump-shaped, bounded-support structure enables closed-form characterization of the fundamental dynamics and sharp statements about local explosiveness. The microfoundation of this appendix derives $\Delta d_t = c + \boldsymbol{\gamma}'\,\mathbf{X}_t + \varepsilon_t$, where $\boldsymbol{\gamma}'\,\mathbf{X}_t$ is the \emph{empirical} counterpart of $\delta_t$---a stochastic, observable proxy driven by the VAR dynamics of $\mathbf{X}_t$. Formally, the two are linked by the identification $\delta_t = \mathbb{E}[\boldsymbol{\gamma}'\mathbf{X}_t \mid \text{regime}]$: the deterministic hump $\delta_t$ of Assumption~\ref{ass:tech_shock} is the conditional-mean path of $\boldsymbol{\gamma}'\mathbf{X}_t$ that would obtain under a technology-diffusion regime with adoption at $T_1$ and maturation at $T_2$, while the stochastic component $\boldsymbol{\gamma}'\mathbf{X}_t - \delta_t$ is absorbed into the residual noise in the empirical specification. Under this mapping, Propositions~\ref{prop:explosive} and~\ref{prop:pd_limit} describe the signal component of the data-generating process that the PSY test erroneously attributes to explosive dynamics, and the technology-adjusted PSY test removes exactly this signal component by projecting on $\mathbf{X}_t$.
\end{remark}

\begin{remark}[Separability and Exclusion of Bubbles]
\label{rmk:exclusion}
The identification conditions (R1)--(R3) above formalize what was informally called ``separability'' in earlier drafts. The bubble enters as an additive component in secondary-market valuation ($p_t = f_t + b_t$, $\mathbb{E}_t[M_{t+1}b_{t+1}]=b_t$) but is excluded from the production block determining $\mathbf{X}_t$ by the structure of \eqref{eq:production}--\eqref{eq:tfp_law} and the FOCs~\eqref{eq:it_foc}. Patents are predetermined at short horizons (grant lag $\ell\geq 1$), TFP is a production-side measure, and IT investment satisfies \eqref{eq:it_foc}; these three facts motivate the exclusion restriction. They should not be read as eliminating every economically plausible feedback channel from valuations into future innovation or adoption. Where such feedback exists, the adjusted residual is attenuated as in \eqref{eq:feedback_residual}.
\end{remark}

\clearpage
\section{Extensions}
\label{app:extensions}

This appendix presents the three extensions summarized in Section~\ref{subsec:extensions}: time-varying expected returns, stochastic technology shocks under Bayesian learning, and cross-sectional heterogeneity. Each extension relaxes one of the simplifying restrictions of the baseline theory and shows that the spurious-explosiveness mechanism is preserved or strengthened.

\subsection{Time-Varying Expected Returns}

The first extension asks whether the explosive pattern survives when expected returns respond to the technology state. High-valuation regimes are typically associated with low expected returns \citep{Campbell1988,Cochrane2008}, so one might worry that a technology-induced decline in $\bar r$ could offset the rise in $\mathcal{T}_t$ and mute the spurious explosiveness. The following result shows the opposite under the natural parameter region $\phi < 1$: a technology-induced decline in expected returns \emph{amplifies}, rather than attenuates, the locally explosive pattern, because a lower discount rate magnifies the present value of future technology payoffs.

Relaxing Assumption~\ref{ass:const_returns}, suppose expected returns respond to the technology shock:
\begin{equation}
    \E_t[r_{t+1+j}] = \rbar + \phi\, \delta_{t+1+j} \qquad \text{for all } j \geq 0, \qquad \phi \in \mathbb{R}.
    \label{eq:tv_returns}
\end{equation}
This horizon-by-horizon specification is consistent with consumption-based asset-pricing environments in which the technology state $\delta_t$ enters the stochastic discount factor.

\begin{proposition}
\label{prop:tv_returns}
Under \eqref{eq:tv_returns},
\[
f_t = d_t + C + (1-\phi)\,\Tt.
\]
If $\phi<1$, the technology component of the log price-dividend ratio, $f_t-d_t-C=(1-\phi)\Tt$, has the same sign, monotonicity, and convexity properties as $\Tt$. If $\phi<0$, this component is amplified by the factor $1-\phi>1$. The drift of $f_t$ becomes $c+\delta_t+(1-\phi)\Delta\Tt$, so any statement about drift relative to $c$ requires the corresponding sign condition on $\Delta\Tt$ on the window of interest.
\end{proposition}

\subsection{Stochastic Technology Shocks}
\label{subsec:stochastic_tech}

The second extension relaxes determinism: in practice the ultimate cumulative impact of a new general-purpose technology (the eventual productivity gains from electrification, the internet, or artificial intelligence) is not known ex ante, and agents learn about it gradually through observation---the mechanism \citet{Pastor2009} identify as central to technology-related asset price dynamics. We now formalize this uncertainty. Under the local parametrization of Assumption~\ref{ass:local_framework} with a fixed discount factor $\rho\in[0,1)$, Bayesian learning affects the statistic only through its pathwise contribution at finite $T$: asymptotically, the learning-innovation term is negligible at PSY scale, so the first-order limit theory is governed by the deterministic channel of Assumption~\ref{ass:tech_shock}. The news-variance parameter $\sigma_\xi$ therefore enters the rejection probability as a \emph{finite-sample} refinement; its quantitative effect is documented empirically via Monte Carlo in Appendix~\ref{subsec:mc_stochastic}.

\begin{assumption}[Stochastic Technology Uncertainty]
\label{ass:stochastic_tech}
The technology shock is
\begin{equation*}
	\delta_t = \bar{\delta} \cdot g_{\star}(t) + \sigma_\xi\, \xi_t, \qquad t \in [T_1, T_2],
\end{equation*}
where $\bar{\delta} \sim N(\mu_\delta, \sigma_\delta^2)$ is the uncertain total technology impact, $g_{\star}:[T_1,T_2]\to\mathbb{R}_+$ is the full normalized hump shape obtained from the deterministic profile of Assumption~\ref{ass:tech_shock} by rescaling so that $\sum_{t=T_1}^{T_2} g_{\star}(t) = 1$ (note: $g_{\star}$ extends over the entire adoption--maturation window $[T_1,T_2]$ and should not be confused with the ramp-only function $g:[0,\tau]\to\mathbb{R}_+$ of Assumption~\ref{ass:tech_shock}), $\xi_t \overset{\text{i.i.d.}}{\sim} N(0,1)$ represents period-by-period technology news, and $\sigma_\xi > 0$ scales the news variance. The prior mean $\mu_\delta$ and variance $\sigma_\delta^2$ are common knowledge at $t = T_1$.
\end{assumption}

Under Bayesian updating, agents observe $\delta_t$ and update their beliefs about $\bar{\delta}$. The signal $\delta_t - \sigma_\xi \xi_t = \bar{\delta}\, g_{\star}(t)$ is contaminated by the news shock $\xi_t$, so the conditional expectation $\hat{\delta}_t \equiv \E_t[\bar{\delta}]$ evolves as a martingale:
\begin{equation*}
	\hat{\delta}_t = \hat{\delta}_{t-1} + K_t\, \nu_t, \qquad \nu_t = \delta_t - \hat{\delta}_{t-1}\, g_{\star}(t),
\end{equation*}
where $\nu_t$ denotes the Kalman innovation (to avoid collision with the VAR innovation $\eta_t$ of~\eqref{eq:eta_t_def}), $K_t = P_{t-1}\, g_{\star}(t)\, [g_{\star}(t)^2\, P_{t-1} + \sigma_\xi^2]^{-1}$ is the Kalman gain and $P_t = P_{t-1} - K_t\, g_{\star}(t)\, P_{t-1}$ is the posterior variance, initialized at $P_0 = \sigma_\delta^2$. The technology present-value term becomes
\begin{equation*}
	\Tt = \sum_{j=0}^{\infty} \rho^j\, \E_t[\delta_{t+1+j}] = \hat{\delta}_t \sum_{j=0}^{\infty} \rho^j\, g_{\star}(t+1+j) + \sigma_\xi \sum_{j=0}^{\infty} \rho^j\, \E_t[\xi_{t+1+j}].
\end{equation*}
Since $\E_t[\xi_{t+1+j}] = 0$, the second sum vanishes and $\Tt = \hat{\delta}_t\, G_\rho(t)$, where $G_\rho(t) = \sum_{j=0}^{\infty} \rho^j\, g_{\star}(t+1+j)$ is a deterministic weight. Revisions to beliefs about $\bar{\delta}$ propagate directly into $\Tt$:
\begin{equation*}
	\Tt - \E_{t-1}[\Tt] = K_t\, \nu_t\, G_\rho(t).
\end{equation*}
Positive news ($\nu_t > 0$) generates a discrete upward jump in $\Tt$ scaled by $G_\rho(t)$, which is largest during the early adoption phase when the remaining future technology path is long. These jumps compound the deterministic hump-shaped rise, producing price dynamics more explosive than the deterministic baseline.

The following local-asymptotic scaling assumption is used for the limit theory of Proposition~\ref{prop:stochastic_tech}.

\begin{assumption}[Local Stochastic Technology Scaling]\label{ass:local_stochastic_scaling}
Under the local framework of Assumption~\ref{ass:local_framework}, the stochastic technology process from Assumption~\ref{ass:stochastic_tech} satisfies the triangular-array scaling
\[
g_{\star,T}(s)=T^{-1}h_\star(s/T)+o(T^{-1})
\]
for some bounded continuous $h_\star:[0,1]\to[0,1]$ supported on $[\lambda_1,\lambda_2]$, with prior mean $\mu_\delta=O(\sqrt{T})$, prior standard deviation $\sigma_\delta=O(\sqrt{T})$, and prior variance $P_0=T\tau_\delta^2$ for some $\tau_\delta>0$.
\end{assumption}

\begin{proposition}[Stochastic Technology Amplification]
\label{prop:stochastic_tech}
Under Assumptions~\ref{ass:dividend}, \ref{ass:const_returns}, \ref{ass:stochastic_tech}, \ref{ass:local_framework}, and \ref{ass:local_stochastic_scaling} with $b_t = 0$:
\begin{enumerate}[label=(\roman*)]
	\item Jumps compound explosive dynamics. Positive news about $\bar{\delta}$ generates upward jumps in $\Tt$ of magnitude $K_t\, \nu_t\, G_\rho(t)$. During the early adoption phase $[T_1, T_1+\tau]$, $K_t$ is high (prior uncertainty is greatest) and $G_\rho(t)>0$ throughout the adoption window (indeed bounded away from zero on a non-empty subinterval near the peak). The news innovations $\nu_t$ form a martingale difference sequence (hence are serially uncorrelated), but the posterior mean $\hat{\delta}_t$ is a martingale whose levels are highly persistent. It is this persistence in levels---not serial correlation of innovations---that propagates each positive jump into all future values of $\Tt$, amplifying the locally explosive pattern relative to the deterministic baseline.
    \item Asymptotic negligibility of the learning channel; finite-sample dependence on $\sigma_\xi$. Under the local parametrization of Assumption~\ref{ass:local_framework} with fixed discount factor $\rho\in[0,1)$, the partial-sum process of Bayesian learning innovations is asymptotically negligible at PSY scale: $T^{-1/2}\sum_{s\le\lfloor Tr\rfloor} K_s\,\nu_s\,G_\rho(s) \Rightarrow 0$ in $D[0,1]$. Consequently the gsADF limit law coincides with the deterministic-$\delta$ limit law of Theorem~\ref{thm:contaminated_limit}, and the asymptotic size distortion is \emph{not} altered by $\sigma_\xi$ at first order. Nontrivial dependence of the test's rejection probability on $\sigma_\xi$ is therefore a \emph{finite-sample} phenomenon---operating through the pathwise contribution of $K_t\,\nu_t\,G_\rho(t)$ to $\tilde y_t$ at empirically relevant sample sizes---and its sign and magnitude are documented numerically in Appendix~\ref{subsec:mc_stochastic}, where Panel~A of Table~\ref{tab:mc_stochastic} shows substantial amplification in the price-dividend implementation and Panel~B shows near-invariance in the detrended log-price implementation.
    	\item Adjusted test remains valid (implementation-specific). If the econometrician uniformly consistently estimates the implementation-specific technology component required by Definition~\ref{def:adj_test}---namely $\mathcal T_t$ in the log price-dividend-ratio implementation and $\sum_{s=1}^{t}\delta_s+\mathcal T_t$ in the detrended log-price implementation---then the technology-adjusted PSY statistic has the same first-order asymptotics as in the corresponding technology-free benchmark. In particular, for detrended log prices its asymptotic size remains $\alpha$, and for the price-dividend-ratio implementation the stochastic technology component is removed along with the deterministic one (so any residual rejection cannot be attributed to the hump-shaped technology term).
\end{enumerate}
\end{proposition}

\begin{remark}[Connection to \citet{Pastor2009}]
\label{rmk:pastor_veronesi}
The stochastic extension nests a simplified version of the \citet{Pastor2009} framework. In their model, uncertainty about a new technology's productivity generates high prices during the adoption phase (the ``bubble'') and price declines when uncertainty resolves unfavorably (the ``crash''). Our decomposition identifies precisely what PSY misattributes: the high prices during the learning phase reflect the conditional expectation $\hat{\delta}_t$ entering the fundamental $f_t = d_t + C + \hat{\delta}_t\, G_\rho(t)$, not a speculative component $b_t$. When uncertainty resolves favorably ($\hat{\delta}_t$ rises), prices jump, a rational Bayesian response. When it resolves unfavorably, prices decline. Both outcomes are fundamental, yet PSY may classify them as bubble origination and collapse. The technology-adjusted test correctly attributes these dynamics to $f_t$.
\end{remark}

\subsection{Cross-Sectional Heterogeneity}

The third extension moves from a representative asset to a cross-section. General-purpose technologies affect firms asymmetrically: technology-intensive industries load heavily on $\delta_t$ while traditional industries load lightly, so the bubble-like pattern that the spurious-explosiveness mechanism generates at the asset level aggregates into a \emph{portfolio-level} pattern whose magnitude depends on the concentration of high-exposure firms in the index. This observation is empirically important because PSY tests are commonly applied to sector indices and to the broad market aggregate, where cross-sectional concentration matters directly.

When firms differ in technology exposure ($\delta_{i,t} = \lambda_i \delta_t$ with $\lambda_i \geq 0$), the aggregate market fundamental $f_t^{\mathrm{mkt}} = \sum_i w_i f_{i,t}$ inherits the technology effect weighted by $\sum_i w_i \lambda_i$. Industries with high technology exposure and large market capitalization weight (e.g., the late-1990s NASDAQ) are most susceptible to spurious bubble detection.

\clearpage
\section{Firm-Level Analysis: The Magnificent Seven}
\label{app:firmlevel}

This appendix reports the full firm-level analysis of the Magnificent Seven stocks summarized in Section~\ref{subsec:mag7_analysis}.

The aggregate NASDAQ index masks substantial heterogeneity across individual stocks. \citet{BaselePhillipsShi2025} applied PSY to the ``Magnificent Seven'' technology stocks to evaluate firm-level speculative dynamics; we extend this analysis by applying our technology-adjusted test at the firm level. We use log dividend-adjusted monthly closing prices from Yahoo Finance, applying the same PSY specification as the index-level analysis: a minimum window $r_0$ defined by the formula $0.01 + 1.8/\sqrt{T}$, lag length of 1, and Monte Carlo critical values with 2,000 replications. Sample periods vary by initial public offering date, extending through early 2026.

Table~\ref{tab:mag7_psy} and Figure~\ref{fig:psy_mag7} present the results. We find that four of the seven stocks reject the null hypothesis of no explosive behavior at the 10\% level: MSFT, GOOGL, TSLA, and NVDA. Three of these (MSFT, TSLA, and NVDA) also reject at the 5\% level, and two (TSLA and NVDA) reject at the stringent 1\% level. NVDA exhibits the strongest signal by far, with a GSADF statistic of 4.155 (well above the 1\% critical value of $2.657$) and extensive explosive episodes from January 2024 through the end of the sample. TSLA shows a prolonged explosive episode from July 2020 to December 2021, coinciding with the broader electric-vehicle and technology euphoria. MSFT shows nearly continuous explosive behavior from February 2019 through January 2026, spanning both the pre- and post-large-language-model periods.

The remaining three stocks (AMZN, META, and AAPL) do not show statistically significant explosive behavior despite their status as market leaders. This cross-sectional heterogeneity is a key finding: even within the Magnificent Seven, explosive price dynamics are not universal.

\subsection{Technology-Adjusted Firm-Level Tests}

To assess whether the detected explosive episodes remain after accounting for firm-level fundamentals, we apply the technology-adjusted PSY test to each Magnificent Seven stock. Unlike the index-level analysis, which uses aggregate technology covariates (log IT investment, TFP, and log patent grants), the firm-level analysis constructs individual counterfactuals from each firm's own Compustat fundamentals. Specifically, for each stock~$i$ we estimate the OLS regression
\[
\log P_{i,t} = \alpha_i + \beta_{1i}\, \log\!\bigl(\text{RevPS}_{i,t}^{\text{real}}\bigr) + \beta_{2i}\, \log\!\bigl(1 + \text{RDPS}_{i,t}^{\text{real}}\bigr) + u_{i,t},
\]
where $\text{RevPS}_{i,t}^{\text{real}} = \text{revtq}_{i,t} / (\text{CPI}_t \times \text{cshoq}_{i,t})$ is CPI-deflated revenue per share and $\text{RDPS}_{i,t}^{\text{real}} = \text{xrdq}_{i,t} / (\text{CPI}_t \times \text{cshoq}_{i,t})$ is CPI-deflated R\&D expenditure per share, both constructed from quarterly Compustat data interpolated to monthly frequency via last observation carried forward (LOCF). Revenue per share captures realized commercialization and firm scale, while R\&D expenditure per share proxies firm-specific technology investment. The training window is the 2006--2019 period, consistent with the index-level AI-era analysis of Section~\ref{subsec:ai_cointegration}, which postdates the dot-com bubble and predates the AI-driven rally. Crucially, because several Magnificent Seven stocks exhibit PSY-detected explosive episodes within the 2006--2019 window (Figure~\ref{fig:psy_mag7}), we exclude each stock's own bubble months from its training sample. This ensures that the counterfactual relationship is estimated exclusively from non-bubble observations, consistent with the requirement that the training period should not be contaminated by speculative dynamics. We then compute the price gap $\hat{g}_{i,t} = p_{i,t} - \hat{f}_{i,t}$ and apply PSY to the residual.

Table~\ref{tab:mag7_psy_adj} and Figure~\ref{fig:psy_mag7_adj} present the technology-adjusted results for all seven stocks. The results are consistent with our central thesis. After technology adjustment, the explosive signals weaken substantially for most stocks. Of the four stocks that rejected the null in the unadjusted test, three (MSFT (test statistic $= 0.605$), GOOGL (test statistic $= 0.633$), and NVDA (test statistic $= 1.531$)) no longer reject at any standard significance level after technology adjustment. Notably, the extensive NVDA explosive episodes spanning 2016--2018 and 2024--2026 largely disappear once the counterfactual is estimated from firm-level fundamentals on bubble-free training data (134 of 168 months retained after excluding 34 bubble months), leaving only minor episodes in 2016--2017. TSLA remains highly significant after adjustment (test statistic $= 4.696$, rejecting at the 1\% level), with episodes concentrated in the 2020--2021 period coinciding with electric-vehicle euphoria; its relatively high $R^2$ of~0.898 indicates that the firm-level fundamentals capture a substantial share of price variation even for Tesla, while residual explosive behavior remains. AAPL, which did not reject in the unadjusted test, shows a significant adjusted test statistic ($2.317$), rejecting at the 10\% and 5\% levels but not at the 1\% level (CV$_{1\%} = 2.660$), with detected episodes in the early-to-mid 2000s, a period outside the 2006--2019 training window where the counterfactual extrapolates less reliably; its lower $R^2$ of~0.407 suggests that the two-covariate specification captures less of Apple's price dynamics, possibly because Apple's transformation from a niche computer maker to a consumer electronics and services platform involved drivers beyond revenue and R\&D growth. The remaining stocks (MSFT, AMZN, META, GOOGL) show no explosive episodes after adjustment, with $R^2$ values ranging from 0.741 (MSFT) to 0.961 (AMZN). These firm-level results parallel the index-level finding: the technology adjustment removes the dominant source of apparent explosiveness in most AI-exposed prices, while leaving firm-specific residual signals.

These firm-level results reinforce the aggregate finding: the technology-adjusted test, used alongside the standard PSY test, reveals that much of the apparent speculative behavior in Magnificent Seven stocks is technology-driven fundamental repricing, a diagnostic decomposition that neither test offers in isolation.

\clearpage
\begin{figure}
\centerline{\includegraphics[width=0.95\textwidth]{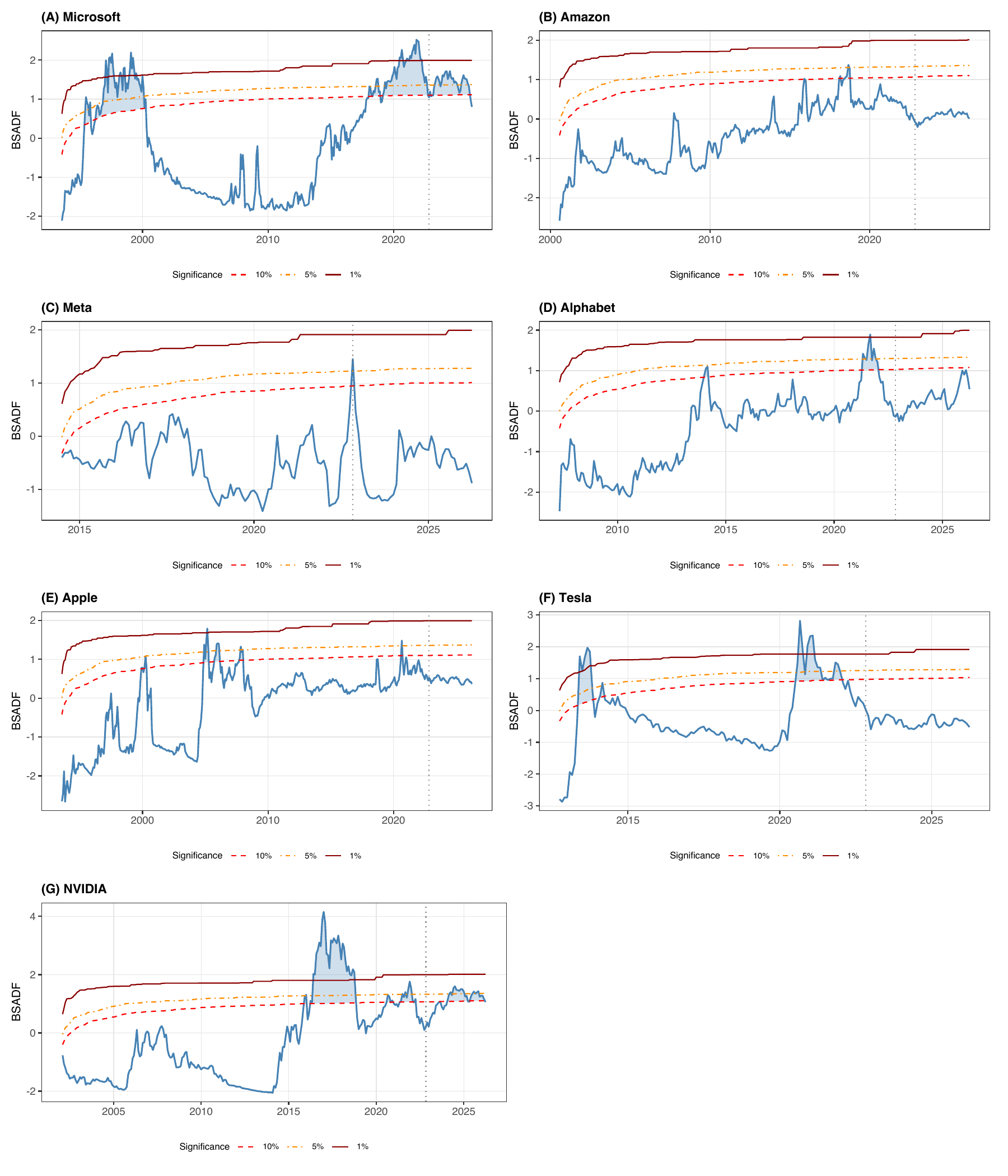}}
\noindent\caption{{\bf Unadjusted firm-level date-stamping highlights the strongest explosive episodes in NVIDIA and Tesla.} In each panel, the horizontal axis is calendar time, the solid line is the BSADF date-stamping statistic, and the dashed lines are the 10, 5, and 1 percent critical value sequences for log dividend-adjusted monthly prices. Shaded regions mark detected explosive episodes, which are most persistent for NVIDIA.\label{fig:psy_mag7}}
\end{figure}

\clearpage
\begin{figure}
\centerline{\includegraphics[width=0.95\textwidth]{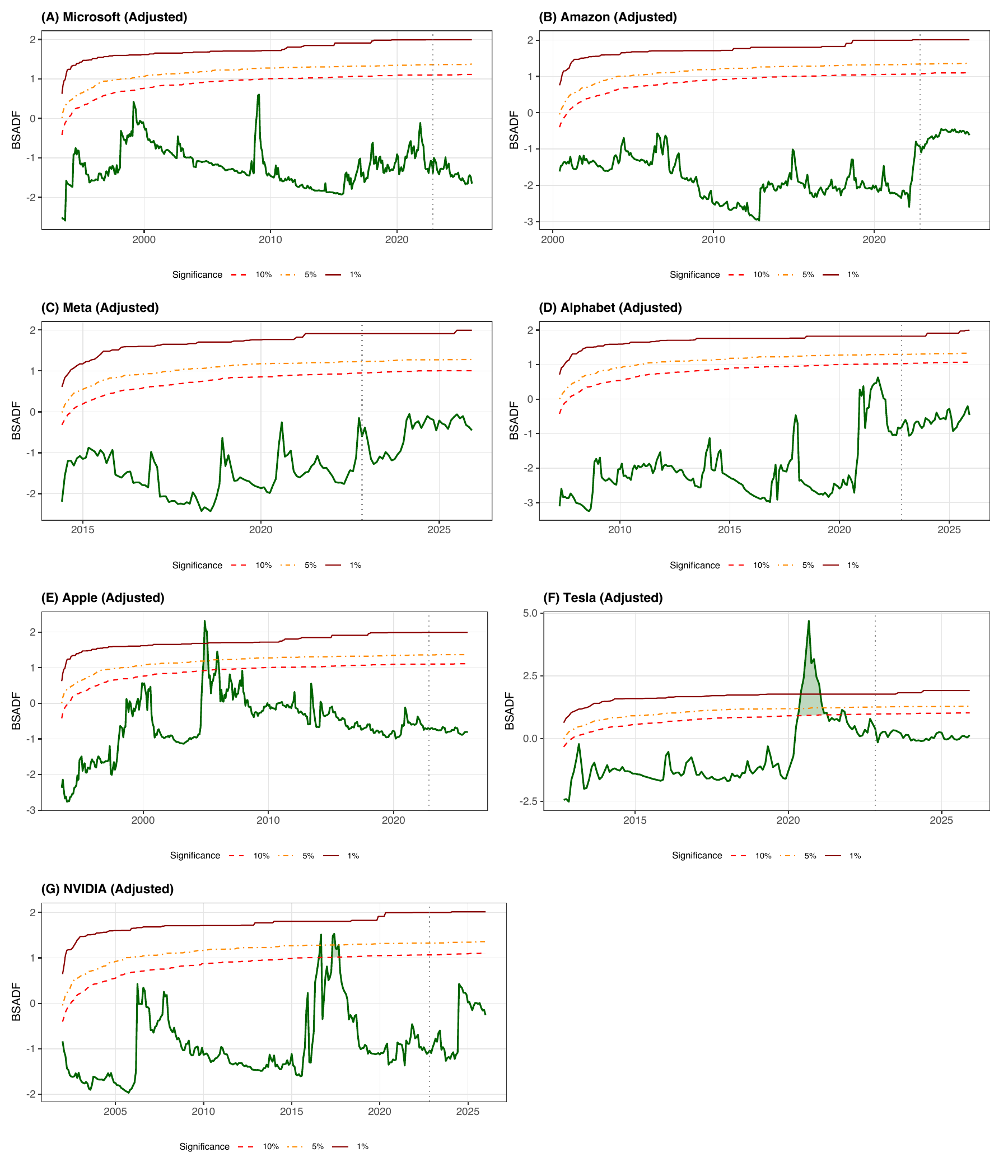}}
\noindent\caption{{\bf Firm-level technology adjustment using Compustat fundamentals weakens most explosive signals.} In each panel, the horizontal axis is calendar time, the solid line is the BSADF date-stamping statistic for the technology-adjusted price gap (based on firm-specific real revenue per share and real R\&D expenditure per share), and the dashed lines are the 10, 5, and 1 percent critical value sequences. Relative to the unadjusted results, most series remain below the threshold for most of the sample.\label{fig:psy_mag7_adj}}
\end{figure}

\clearpage
\begin{landscape}
\begin{table}
\noindent\caption{{\bf Four Magnificent Seven stocks show unadjusted explosive signals, with NVIDIA and Tesla strongest.} The table reports GSADF statistics, critical values, rejection decisions, and detected episodes for log dividend-adjusted monthly prices. Microsoft, Alphabet, Tesla, and NVIDIA reject at the 10 percent level, and Tesla and NVIDIA also reject at the 1 percent level. Sample periods vary by listing date, all extending through March 2026, and critical values are computed by Monte Carlo with 2,000 replications.\label{tab:mag7_psy}}
\centering
\begin{tabular}{lcccccccccc}
\toprule
\textbf{Stock} & \textbf{$n$} & \textbf{Period} & \textbf{GSADF} & \textbf{CV$_{0.10}$} & \textbf{CV$_{0.05}$} & \textbf{CV$_{0.01}$} & \textbf{Rej.\ 10\%} & \textbf{Rej.\ 5\%} & \textbf{Rej.\ 1\%} & \textbf{Episodes} \\
\midrule
Microsoft & 435 & 1990--2026 & 2.522 & 1.957 & 2.233 & 2.660 & Yes & Yes & No & 7 \\
Amazon & 347 & 1997--2026 & 1.374 & 1.885 & 2.128 & 2.660 & No & No & No & 3 \\
Meta & 167 & 2012--2026 & 1.455 & 1.747 & 2.003 & 2.580 & No & No & No & 1 \\
Alphabet & 260 & 2004--2026 & 1.894 & 1.826 & 2.101 & 2.609 & Yes & No & No & 2 \\
Apple & 435 & 1990--2026 & 1.793 & 1.957 & 2.233 & 2.660 & No & No & No & 6 \\
Tesla & 190 & 2010--2026 & 2.817 & 1.768 & 2.041 & 2.576 & Yes & Yes & Yes & 3 \\
NVIDIA & 327 & 1999--2026 & 4.155 & 1.903 & 2.138 & 2.657 & Yes & Yes & Yes & 8 \\
\bottomrule
\end{tabular}
\end{table}
\end{landscape}

\clearpage
\begin{landscape}
\begin{table}
\noindent\caption{{\bf Technology adjustment removes most firm-level explosive signals in the Magnificent Seven.} The table reports GSADF statistics, critical values, rejection decisions, detected episodes, and in-sample fit for price gaps defined as observed prices minus counterfactual fundamentals. The counterfactual for each stock is estimated via OLS of log price on log real revenue per share and log$(1 +$ real R\&D per share$)$, constructed from quarterly Compustat data interpolated to monthly frequency. After adjustment, Microsoft, Amazon, Meta, Alphabet, and NVIDIA no longer reject, while Tesla rejects at the 1 percent level and Apple rejects at the 5 percent level. The reported $R^2$ values refer to the in-sample fit of the counterfactual regressions, and sample periods vary by listing date.\label{tab:mag7_psy_adj}}
\centering
\begin{tabular}{lccccccccc}
\toprule
\textbf{Stock} & \textbf{GSADF} & \textbf{CV$_{0.10}$} & \textbf{CV$_{0.05}$} & \textbf{CV$_{0.01}$} & \textbf{Rej.\ 10\%} & \textbf{Rej.\ 5\%} & \textbf{Rej.\ 1\%} & \textbf{Episodes} & \textbf{$R^2$} \\
\midrule
Microsoft & 0.605 & 1.952 & 2.223 & 2.660 & No & No & No & 0 & 0.741 \\
Amazon & $-0.445$ & 1.910 & 2.150 & 2.672 & No & No & No & 0 & 0.961 \\
Meta & $-0.052$ & 1.743 & 1.990 & 2.576 & No & No & No & 0 & 0.908 \\
Alphabet & 0.633 & 1.826 & 2.101 & 2.609 & No & No & No & 0 & 0.809 \\
Apple & 2.317 & 1.952 & 2.223 & 2.660 & Yes & Yes & No & 2 & 0.407 \\
Tesla & 4.696 & 1.764 & 2.035 & 2.576 & Yes & Yes & Yes & 2 & 0.898 \\
NVIDIA & 1.531 & 1.902 & 2.133 & 2.657 & No & No & No & 2 & 0.782 \\
\bottomrule
\end{tabular}
\end{table}
\end{landscape}

\clearpage
\section{Proofs}
\label{app:proofs}

This appendix collects the formal proofs of all propositions, lemmas, and theorems stated in the main text.

\subsection*{Proof of Proposition~\ref{prop:fundamental} (Fundamental Price Decomposition)}

Under Assumptions~\ref{ass:dividend}--\ref{ass:const_returns} and $b_t = 0$, the present-value identity \eqref{eq:PV_full} becomes
\begin{equation*}
	f_t = d_t + \frac{\kappa}{1-\rho} + \sum_{j=0}^{\infty} \rho^{j}\, \E_t[\Delta d_{t+1+j} - r_{t+1+j}].
\end{equation*}
Substituting the dividend growth process $\Delta d_{t+1+j} = c + \delta_{t+1+j} + \varepsilon_{t+1+j}$ from Assumption~\ref{ass:dividend} and the constant expected return $\E_t[r_{t+1+j}] = \rbar$ from Assumption~\ref{ass:const_returns}:
\begin{align}
	\E_t[\Delta d_{t+1+j} - r_{t+1+j}]
	&= \E_t[c + \delta_{t+1+j} + \varepsilon_{t+1+j}] - \rbar \notag \\
	&= (c - \rbar) + \delta_{t+1+j},
\end{align}
where we used the fact that $\delta_{t+1+j}$ is deterministic (Assumption~\ref{ass:tech_shock}) and $\E_t[\varepsilon_{t+1+j}] = 0$. Substituting:
\begin{align}
	f_t &= d_t + \frac{\kappa}{1-\rho} + \sum_{j=0}^{\infty} \rho^{j}\, \big[(c - \rbar) + \delta_{t+1+j}\big] \notag \\
	&= d_t + \frac{\kappa}{1-\rho} + \frac{c - \rbar}{1-\rho} + \sum_{j=0}^{\infty} \rho^{j}\, \delta_{t+1+j} \notag \\
	&= d_t + C + \Tt,
\end{align}
where $C = \frac{\kappa}{1-\rho} + \frac{c - \rbar}{1-\rho}$ and $\Tt = \sum_{j=0}^{\infty} \rho^{j}\, \delta_{t+1+j}$. The geometric sum $\sum_{j=0}^{\infty} \rho^j = (1-\rho)^{-1}$ converges since $\rho \in (0,1)$, and $\Tt$ converges since $\delta_t$ has bounded support $[T_1, T_2]$. \qed

\subsection*{Proof of Proposition~\ref{prop:microfoundation} (Observable Technology Proxies and Fundamental Value)}

Denote the technology state vector $\mathbf{X}_t \equiv (TFP_t,\, \log IT_t,\, \log Pat_t)'$ and suppose it follows a first-order vector autoregressive process $\mathbf{X}_{t+1} = \boldsymbol{\Phi}\, \mathbf{X}_t + \mathbf{v}_{t+1}$, where the spectral radius of $\boldsymbol{\Phi}$ satisfies $\rho(\boldsymbol{\Phi}) < \rho^{-1}$ (the persistence of the technology state is dominated by the discount factor). We prove each part in turn.

\medskip\noindent\textit{Part (i): Log dividend growth.}
The production function is $Y_t = A_t\, K_t^{\alpha}\, G_t^{\eta}\, L_t^{1-\alpha-\eta}$. Log-linearizing output around the balanced-growth path and using the optimality conditions for labor, physical capital, and information technology capital yields $\Delta \log Y_t = \Delta a_t + \alpha\, \Delta \log K_t + \eta\, \Delta \log G_t + (1-\alpha-\eta)\, \Delta \log L_t$. Dividends are residual cash flow: $D_t = Y_t - w_t L_t - I_t^K - I_t^{IT} - \Phi(I_t^{IT}/G_t)G_t - \Xi(q_t)$. To a first-order approximation around the balanced-growth path, the log-linearized dividend growth is affine in the technology state:
\begin{equation*}
    \Delta d_t = c + \boldsymbol{\gamma}'\, \mathbf{X}_t + \varepsilon_t,
\end{equation*}
where $\boldsymbol{\gamma} = (\gamma_{TFP},\, \gamma_{IT},\, \gamma_{Pat})'$ collects the elasticities of dividend growth with respect to each technology proxy and $\varepsilon_t$ captures idiosyncratic shocks orthogonal to $\mathbf{X}_t$.

The key structural ingredients are: (a)~total factor productivity enters through the disembodied productivity channel $\Delta a_{t+1} = \mu_A + \alpha_q\, q_t + \alpha_I\, i_t^{IT} + u_{t+1}^A$; (b)~information technology investment enters through the information technology capital accumulation $G_{t+1} = (1-\delta_G)\, G_t + I_t^{IT}$ and the optimality condition \eqref{eq:it_foc}; (c)~patent grants enter through the innovation quality channel with grant lag $\ell \geq 1$.

\medskip\noindent\textit{Part (ii): $\Tt = \boldsymbol{\beta}'\, \mathbf{X}_t + \eta_t$.}
From Proposition~\ref{prop:fundamental}, the technology present-value term is $\Tt = \sum_{j=0}^{\infty} \rho^j\, \delta_{t+1+j}$, and from Part~(i) the identification $\delta_s = \boldsymbol{\gamma}'\mathbf{X}_s$ (to first order, after absorbing the zero-mean dividend-growth residual $\varepsilon_s$ into the noise) yields
\begin{equation}
\Tt \;=\; \sum_{j=0}^{\infty} \rho^{\,j}\,\boldsymbol{\gamma}'\,\mathbf{X}_{t+1+j}.
\label{eq:Tt_raw}
\end{equation}
\emph{Step 1: Decompose each $\mathbf{X}_{t+1+j}$ into predictable part and innovation noise.} Add and subtract the conditional expectation to obtain
\begin{align}
\Tt
&\;=\; \sum_{j=0}^{\infty} \rho^{\,j}\,\boldsymbol{\gamma}'\,\Big(\E_t[\mathbf{X}_{t+1+j}] \;+\; \big(\mathbf{X}_{t+1+j} - \E_t[\mathbf{X}_{t+1+j}]\big)\Big) \notag \\
&\;=\; \underbrace{\sum_{j=0}^{\infty} \rho^{\,j}\,\boldsymbol{\gamma}'\,\E_t[\mathbf{X}_{t+1+j}]}_{\text{(A) predictable part}}
\;+\; \underbrace{\sum_{j=0}^{\infty} \rho^{\,j}\,\boldsymbol{\gamma}'\,\big(\mathbf{X}_{t+1+j} - \E_t[\mathbf{X}_{t+1+j}]\big)}_{\text{(B) innovation noise}}.
\label{eq:Tt_decomp}
\end{align}
\emph{Step 2: Evaluate the predictable part (A).} Iterating the VAR $\mathbf{X}_{s+1} = \boldsymbol{\Phi}\mathbf{X}_s + \mathbf{v}_{s+1}$ forward from date $t$ and taking conditional expectations, $\E_t[\mathbf{X}_{t+1+j}] = \boldsymbol{\Phi}^{\,j+1}\mathbf{X}_t$ for each $j \geq 0$. (The unconditional mean is zero by demeaning, so no constant term arises.) Substituting,
\begin{equation*}
\text{(A)} \;=\; \sum_{j=0}^{\infty} \rho^{\,j}\,\boldsymbol{\gamma}'\,\boldsymbol{\Phi}^{\,j+1}\,\mathbf{X}_t
\;=\; \boldsymbol{\gamma}'\,\boldsymbol{\Phi}\,\Bigg(\sum_{j=0}^{\infty} (\rho\,\boldsymbol{\Phi})^{\,j}\Bigg)\,\mathbf{X}_t.
\end{equation*}
The matrix geometric series $\sum_{j=0}^{\infty}(\rho\boldsymbol{\Phi})^{\,j}$ converges in operator norm because the spectral radius satisfies $\rho(\rho\boldsymbol{\Phi}) = \rho\cdot\rho(\boldsymbol{\Phi}) < \rho\cdot\rho^{-1} = 1$ by Assumption~\ref{ass:var}, and the standard Neumann-series identity gives $\sum_{j=0}^{\infty}(\rho\boldsymbol{\Phi})^{\,j} = (\mathbf{I}-\rho\boldsymbol{\Phi})^{-1}$. Therefore
\begin{equation}
\text{(A)} \;=\; \boldsymbol{\gamma}'\,\boldsymbol{\Phi}\,(\mathbf{I}-\rho\boldsymbol{\Phi})^{-1}\,\mathbf{X}_t
\;=\; \big[(\mathbf{I}-\rho\boldsymbol{\Phi})^{-\top}\,\boldsymbol{\Phi}^{\top}\,\boldsymbol{\gamma}\big]' \,\mathbf{X}_t
\;\equiv\; \boldsymbol{\beta}'\,\mathbf{X}_t,
\label{eq:predictable_part}
\end{equation}
where we have used the identity $(\mathbf{A}'\mathbf{B})' = \mathbf{B}'\mathbf{A}$ applied to $\mathbf{A}=\boldsymbol{\gamma}$ and $\mathbf{B}=\boldsymbol{\Phi}(\mathbf{I}-\rho\boldsymbol{\Phi})^{-1}\mathbf{X}_t$, and we have defined the coefficient vector $\boldsymbol{\beta} \equiv (\mathbf{I}-\rho\boldsymbol{\Phi})^{-\top}\boldsymbol{\Phi}^{\top}\boldsymbol{\gamma}$.

\emph{Step 3: Evaluate the innovation term (B).} Iterating the VAR forward from date $t$ via $\mathbf{X}_{t+1} = \boldsymbol{\Phi}\mathbf{X}_t + \mathbf{v}_{t+1}$, $\mathbf{X}_{t+2} = \boldsymbol{\Phi}^{2}\mathbf{X}_t + \boldsymbol{\Phi}\mathbf{v}_{t+1} + \mathbf{v}_{t+2}$, and more generally,
\begin{equation*}
\mathbf{X}_{t+1+j} \;=\; \boldsymbol{\Phi}^{\,j+1}\mathbf{X}_t \;+\; \sum_{l=0}^{j}\boldsymbol{\Phi}^{\,l}\,\mathbf{v}_{t+1+j-l}.
\end{equation*}
Subtracting $\E_t[\mathbf{X}_{t+1+j}] = \boldsymbol{\Phi}^{\,j+1}\mathbf{X}_t$ isolates the innovation component $\mathbf{X}_{t+1+j} - \E_t[\mathbf{X}_{t+1+j}] = \sum_{l=0}^{j}\boldsymbol{\Phi}^{\,l}\mathbf{v}_{t+1+j-l}$. Substituting into~\eqref{eq:Tt_decomp} and defining
\begin{equation}
\eta_t \;\equiv\; \text{(B)} \;=\; \sum_{j=0}^{\infty} \rho^{\,j}\,\boldsymbol{\gamma}'\,\Bigg(\sum_{l=0}^{j}\boldsymbol{\Phi}^{\,l}\,\mathbf{v}_{t+1+j-l}\Bigg),
\label{eq:eta_appendix_def}
\end{equation}
we must verify: (a) absolute convergence in mean square; (b) mean zero; (c) stationarity; and (d) orthogonality to $\mathbf{X}_t$.

\emph{Convergence.} Because $\rho(\boldsymbol{\Phi}) < \rho^{-1}$, by Gelfand's formula there exist constants $C_0>0$ and $r\in(\rho(\boldsymbol{\Phi}),\rho^{-1})$ such that $\|\boldsymbol{\Phi}^{\,l}\| \leq C_0\, r^{\,l}$ for all $l \geq 0$ (in any submultiplicative matrix norm). Using $\E\|\mathbf{v}_s\|^2 \leq \sigma_v^2 < \infty$ and uncorrelatedness of the $\mathbf{v}_s$ across~$s$, the double sum defining $\eta_t$ is bounded in $L^2$ by exchanging orders of summation:
\begin{equation*}
\E\|\eta_t\|^{2} \;\leq\; \|\boldsymbol{\gamma}\|^{2}\, C_{0}^{2}\,\sigma_{v}^{2}\sum_{j=0}^{\infty}\sum_{l=0}^{j}\rho^{\,j}\,r^{\,l}
\;=\; \|\boldsymbol{\gamma}\|^{2}\, C_{0}^{2}\,\sigma_{v}^{2}\sum_{l=0}^{\infty} r^{\,l}\sum_{j=l}^{\infty}\rho^{\,j}
\;=\; \frac{\|\boldsymbol{\gamma}\|^{2}\, C_{0}^{2}\,\sigma_{v}^{2}}{(1-\rho)(1-\rho r)},
\end{equation*}
where the final equality uses $\sum_{j=l}^{\infty}\rho^{\,j} = \rho^{\,l}/(1-\rho)$ and $\sum_{l=0}^{\infty}(\rho r)^{\,l} = 1/(1-\rho r)$, both finite because $\rho r < \rho\cdot\rho^{-1} = 1$. Hence $\eta_t$ converges absolutely in $L^{2}$.

\emph{Mean zero.} Each $\mathbf{v}_{t+1+j-l}$ is $\mathcal{F}_{t+j-l}$-measurable and satisfies $\E[\mathbf{v}_{s}]=0$, hence $\E[\eta_t]=0$.

\emph{Stationarity.} The joint distribution of $\{\mathbf{v}_s\}_{s\geq t+1}$ depends on $t$ only through the index shift $s\mapsto s-t$, and the deterministic weights $\rho^{\,j}\boldsymbol{\gamma}'\boldsymbol{\Phi}^{\,l}$ depend only on $(j,l)$. Hence $\eta_t$ is strictly stationary when $\{\mathbf{v}_s\}$ is strictly stationary (and covariance-stationary otherwise), with finite second moment by the convergence argument above.

\emph{Orthogonality to $\mathbf{X}_t$.} Each innovation $\mathbf{v}_{t+1+j-l}$ (for $j\geq 0$, $0\leq l\leq j$, which ranges over indices $t+1-l+j\geq t+1$) is $\mathcal{F}_t$-unpredictable, i.e., $\E[\mathbf{v}_{t+1+j-l}\mid\mathcal{F}_t]=0$. Since $\mathbf{X}_t\in\mathcal{F}_t$, the conditional covariance $\mathrm{Cov}(\mathbf{X}_t, \eta_t\mid\mathcal{F}_t)=0$, and hence $\mathrm{Cov}(\mathbf{X}_t,\eta_t)=0$ in population.

\emph{Step 4: Assemble.} Combining~\eqref{eq:Tt_decomp} with~\eqref{eq:predictable_part} and~\eqref{eq:eta_appendix_def}, we obtain the decomposition
\begin{equation*}
\Tt \;=\; \boldsymbol{\beta}'\,\mathbf{X}_t \;+\; \eta_t, \qquad
\boldsymbol{\beta} \;=\; (\mathbf{I}-\rho\boldsymbol{\Phi})^{-\top}\boldsymbol{\Phi}^{\top}\boldsymbol{\gamma},
\end{equation*}
with $\eta_t$ mean-zero, stationary, and orthogonal to $\mathbf{X}_t$, as claimed.

\medskip\noindent\textit{Part (iii): $f_t = d_t + \beta_0 + \boldsymbol{\beta}'\, \mathbf{X}_t + \eta_t$.}
From Proposition~\ref{prop:fundamental}, $f_t = d_t + C + \Tt$. Substituting Part~(ii) yields
\begin{equation*}
    f_t \;=\; d_t \;+\; C \;+\; \boldsymbol{\beta}'\, \mathbf{X}_t \;+\; \eta_t.
\end{equation*}
Defining $\beta_0 \equiv C$ gives the stated representation. By Part~(ii), $\eta_t$ is mean-zero, stationary, and orthogonal to $\mathbf{X}_t$. This argument keeps $d_t$ explicit; the reduced-form coefficient $\boldsymbol{\beta}_X$ in \eqref{eq:coint_structural} arises only after absorbing the contemporaneous loading of $d_t$ on $\mathbf{X}_t$. Hence $\mathbf{X}_t$ spans the technology-driven present-value component of fundamentals. \qed

\subsection*{Proof of Proposition~\ref{prop:dynamics} (Fundamental Dynamics)}

From Proposition~\ref{prop:fundamental}, $f_t = d_t + C + \Tt$. First-differencing:
\begin{align}
    f_t - f_{t-1} &= (d_t - d_{t-1}) + (\Tt - \Ttm) \notag \\
    &= \Delta d_t + \Delta \Tt \notag \\
    &= c + \delta_t + \varepsilon_t + \Delta \Tt.
\end{align}
Using the recursion from Lemma~\ref{lem:T_recursion}:
\begin{align}
    \Delta \Tt &= \Tt - \Ttm = (\rho^{-1}\, \Ttm - \rho^{-1}\, \delta_t) - \Ttm \notag \\
    &= (\rho^{-1} - 1)\, \Ttm - \rho^{-1}\, \delta_t.
\end{align}
Substituting back:
\begin{align}
    f_t - f_{t-1} &= c + \delta_t + \varepsilon_t + (\rho^{-1} - 1)\, \Ttm - \rho^{-1}\, \delta_t \notag \\
    &= c + (1 - \rho^{-1})\, \delta_t + (\rho^{-1} - 1)\, \Ttm + \varepsilon_t \notag \\
    &= \mu_t + \varepsilon_t,
\end{align}
where $\mu_t = c + (1 - \rho^{-1})\, \delta_t + (\rho^{-1} - 1)\, \Ttm$. \qed

\subsection*{Proof of Proposition~\ref{prop:explosive} (Local Explosiveness)}

We verify each part of the proposition.

\medskip\noindent\textit{Part (i): Pre-adoption ($t < T_1$).} For $t < T_1$, $\delta_s = 0$ for all $s \leq t$. Writing $\Tt = \sum_{j=0}^{\infty} \rho^j \delta_{t+1+j}$ and noting that $\delta_{t+1+j} = 0$ for $j = 0, \ldots, T_1 - t - 2$ (when $t+1 < T_1$), only indices $j \geq T_1 - t - 1$ contribute. Using the uniform bound $0 \leq \delta_s \leq \dmax$ on the support of $\delta$ (Assumption~\ref{ass:tech_shock}), we obtain the explicit inequality
\begin{equation*}
0 \;\leq\; \Tt \;=\; \sum_{j=T_1-t-1}^{\infty} \rho^j\, \delta_{t+1+j} \;\leq\; \dmax \sum_{j=T_1-t-1}^{\infty} \rho^j \;=\; \dmax\,\frac{\rho^{T_1-t-1}}{1-\rho}.
\end{equation*}
Hence $\Tt$ decays geometrically in the gap $T_1 - t$ at rate $\rho$. Combined with the drift identity $\mu_t = c + (\rho^{-1}-1)\Ttm - (\rho^{-1}-1)\delta_t$ (Proposition~\ref{prop:dynamics}) and $\delta_t = 0$, we obtain $\mu_t - c = (\rho^{-1}-1)\Ttm$, so
\begin{equation*}
|\mu_t - c| \;=\; (\rho^{-1}-1)\,\Ttm \;\leq\; (\rho^{-1}-1)\cdot\frac{\dmax\,\rho^{\,T_1-t}}{1-\rho} \;=\; \frac{\dmax}{\rho}\cdot\rho^{\,T_1-t},
\end{equation*}
using $(\rho^{-1}-1) = (1-\rho)/\rho$. Thus $\mu_t \to c$ geometrically at rate $\rho$ as $T_1 - t \to \infty$, so $f_t$ is asymptotically a random walk with drift $c$. In many applications (and in our Monte Carlo design), we set $\delta_t = 0$ for $t < T_1$ exactly and $\Tt$ begins rising from zero only as $t$ approaches $T_1 - 1$ from below.

\medskip\noindent\textit{Part (ii): Early adoption ($t \in [T_1, T_1 + \tau)$).} We use the compact identity $\mu_t - c = (\rho^{-1}-1)(\mathcal{T}_{t-1} - \delta_t)$, which follows from the drift formula $\mu_t = c + (1-\rho^{-1})\delta_t + (\rho^{-1}-1)\mathcal{T}_{t-1}$ in Proposition~\ref{prop:dynamics}. Substituting the definition $\mathcal{T}_{t-1} = \sum_{j=0}^{\infty}\rho^j \delta_{t+j} = \delta_t + \sum_{j=1}^{\infty}\rho^j \delta_{t+j}$ gives
\begin{equation*}
\mu_t - c \;=\; (\rho^{-1}-1)\sum_{j=1}^{\infty}\rho^j\, \delta_{t+j}.
\end{equation*}
Since $\rho^{-1} > 1$, $\rho^j > 0$, and $\delta_{t+j} \geq 0$ with at least one strictly positive term for $t \in [T_1, T_1+\tau)$ (because future shocks in the adoption window are strictly positive by Assumption~\ref{ass:tech_shock}), we conclude $\mu_t > c$ strictly throughout this phase. This proves the ``above-drift-$c$'' claim without any sign restriction on $\Delta \mathcal{T}_t$.

\medskip\noindent\textit{Convexity of $\mathcal{T}_t$ on $[T_1, T_1+\tau^\star)$ under Assumption~\ref{ass:T_convexity}.} We derive a useful identity for $\Delta^2 \mathcal{T}_t$ in terms of second differences of $\delta$, which both motivates Assumption~\ref{ass:T_convexity} and underpins the verification Lemma~\ref{lem:T_convex_verify}. By Lemma~\ref{lem:T_recursion}, $\Delta \mathcal{T}_t = (\rho^{-1}-1)\mathcal{T}_{t-1} - \rho^{-1}\Delta \delta_t$. Applying the same recursion once more and rearranging,
\begin{equation*}
\Delta^2 \mathcal{T}_t \;=\; \sum_{j=1}^{\infty} \rho^{j-1}\, \Delta^2 \delta_{t+j},
\end{equation*}
which one verifies by substituting the definition $\mathcal{T}_t = \sum_{j\geq 0}\rho^j \delta_{t+1+j}$ and telescoping. Thus $\Delta^2 \mathcal{T}_t$ is a $\rho$-discounted weighted average of future second differences of $\delta$. Assumption~\ref{ass:T_convexity} ($\Delta^2 \mathcal{T}_t > 0$ on $[T_1, T_1+\tau^\star)$) is therefore an endogenous property of the hump shape that one checks by direct calculation; it holds for strictly-convex $g$ and standard smooth-hump profiles, but fails for the piecewise-linear triangular profile (Lemma~\ref{lem:T_convex_verify}). Granting Assumption~\ref{ass:T_convexity}, $\mathcal{T}_t$ is strictly convex on $[T_1, T_1+\tau^\star)$, and the log price--dividend ratio $f_t - d_t = C + \mathcal{T}_t$ inherits strict convexity. This is the convexity invoked in the proof of Proposition~\ref{prop:pd_limit}. No monotonicity claim for $\mathcal{T}_t$ follows from $\mu_t > c$ alone; such a claim would require additional restrictions on the discounted future-shock profile.

\medskip\noindent\textit{Part (iii): Late adoption ($t \in (T_1+\tau, T_2]$).} After the peak, $\delta_t$ is positive but decreasing ($h$ is strictly decreasing on $[\tau, T_2-T_1]$ by Assumption~\ref{ass:tech_shock}). The identity $\mu_t - c = (\rho^{-1}-1)\sum_{j=1}^{\infty}\rho^j \delta_{t+j}$ derived in Part~(ii) remains valid. Its right-hand side is a $\rho$-weighted sum of non-negative future $\delta$ values with at least one strictly positive term whenever $t < T_2$, so $\mu_t > c$ strictly on $(T_1+\tau, T_2]$. To show monotone decline of $\mu_t$ to $c$, compute
\begin{equation*}
(\mu_{t+1}-c)-(\mu_t-c) \;=\; (\rho^{-1}-1)\sum_{j=1}^{\infty}\rho^j\,\big[\delta_{t+1+j}-\delta_{t+j}\big].
\end{equation*}
For $t \geq T_1+\tau$ and each $j\geq 1$, the index $t+j$ lies in the descending branch of $h$ (until it exits the support $[T_1,T_2]$, after which the term is identically zero), so $\delta_{t+1+j}-\delta_{t+j} \leq 0$ with strict inequality for at least one $j$ as long as $t < T_2$. Hence $\mu_{t+1} < \mu_t$ strictly, and as $t \to T_2$ the sum shrinks to zero, giving $\mu_t \downarrow c$ monotonically on this phase.

\medskip\noindent\textit{Part (iv): Post-maturation ($t > T_2$).} For $t > T_2$, $\delta_s = 0$ for all $s > T_2$, so $\Tt = \sum_{j=0}^{\infty} \rho^j \delta_{t+1+j} = 0$, and the drift $\mu_t = c$. The fundamental price resumes a random walk with drift $c$, but at a permanently elevated level reflecting the cumulative technology-driven dividend growth $\sum_{s=T_1}^{T_2} \delta_s$. \qed

\subsection*{Proof of Proposition~\ref{prop:pd_limit} (Deterministic Price-Dividend Ratio Limit)}

Under the price-dividend implementation, $y_t = C + \mathcal{T}_{t,T} + u_t$. We restrict attention to subsample windows $[\lfloor r_1 T \rfloor, \lfloor r_2 T \rfloor]$ with $(r_1, r_2) \subseteq (\lambda_1, \lambda_1 + \kappa(\lambda_2-\lambda_1))$ (i.e.\ strictly inside the ramp-up phase where $q$ is strictly increasing and strictly convex), and we impose the signal-dominance condition
\begin{equation*}
	\frac{\dmax(T)^2}{T\,\sigma_u^2} \;\longrightarrow\; \infty \qquad\text{as } T \to \infty,
\end{equation*}
equivalently $\dmax(T)/(\sigma_u\sqrt{T}) \to \infty$. This requires $\dmax(T)$ to grow sufficiently quickly relative to $\sqrt{T}$; in particular, neither the fixed-$\dmax$ regime nor the local parametrization $\dmax(T)=c_\delta/\sqrt{T}$ satisfies this condition. Those non-divergent cases are outside the scope of this proposition. Consider the augmented Dickey--Fuller regression $\Delta y_t = \hat{\alpha} + \hat{\beta}\, y_{t-1} + \hat{u}_t$ on this subsample.

\medskip\noindent\textit{Step 1: Deterministic limit of $\hat{\beta}$.}
Under the continuous-time embedding, $\mathcal{T}_{\lfloor Tr \rfloor, T} = \dmax(T)\, q(r) + o(\dmax(T))$, so $\Delta y_{\lfloor Tr \rfloor} \approx \dmax(T)\, q'(r)/T + \Delta u_t$. The OLS formula with demeaning gives
\[
\hat{\beta}_{r_1,r_2} = \frac{\sum_{t} (y_{t-1} - \bar{y}_{-1})(\Delta y_t - \overline{\Delta y})}{\sum_{t} (y_{t-1} - \bar{y}_{-1})^2}.
\]
Under the signal-dominance condition, the deterministic component dominates both sample moments. Denote $\tilde{q}(r) = q(r) - \frac{1}{r_2-r_1}\int_{r_1}^{r_2} q(s)ds$. Then:
\begin{align*}
	\sum_{t} (y_{t-1} - \bar{y}_{-1})^2 &= T\, \dmax(T)^2 \int_{r_1}^{r_2} \tilde{q}(r)^2\, dr \;+\; O_p(T \sigma_u^2) \;+\; O_p(\sqrt{T}\,\dmax(T)\sigma_u),\\
	\sum_{t} (y_{t-1} - \bar{y}_{-1})(\Delta y_t - \overline{\Delta y}) &= \dmax(T)^2 \int_{r_1}^{r_2} \tilde{q}(r)\, q'(r)\, dr \;+\; O_p(\dmax(T)\sigma_u) \;+\; O_p(\sigma_u^2).
\end{align*}
Under the imposed signal-dominance condition $\dmax(T)^2/(T\sigma_u^2)\to\infty$, these stochastic remainder terms are negligible for the probability limit of $T\,\hat{\beta}_{r_1,r_2}$, so the deterministic component governs the first-order behavior. Multiplying through by $T$ yields
\[
T\, \hat{\beta}_{r_1,r_2} \plim B(r_1, r_2) = \frac{\int_{r_1}^{r_2} \tilde{q}_{r_1,r_2}(r)\, q'(r)\, dr}{\int_{r_1}^{r_2} \tilde{q}_{r_1,r_2}(r)^2\, dr}.
\]

\medskip\noindent\textit{Step 2: Sign of $B(r_1,r_2)$.}
The denominator is always positive. For the numerator, integration by parts gives
\[
\int_{r_1}^{r_2} \tilde{q}(r)\, q'(r)\, dr = \Big[\tilde{q}(r)\, q(r)\Big]_{r_1}^{r_2} - \int_{r_1}^{r_2} q(r)\, \tilde{q}'(r)\, dr.
\]
Since $\tilde{q}' = q'$ (the mean subtraction is a constant) and using $\tilde{q}(r) = q(r) - \bar{q}$:
\[
\int_{r_1}^{r_2} \tilde{q}(r)\, q'(r)\, dr = \frac{1}{2}\Big[\tilde{q}(r_2)^2 - \tilde{q}(r_1)^2\Big].
\]
It remains to sign $\tilde{q}(r_2)^2 - \tilde{q}(r_1)^2$. Because $q$ is strictly increasing and strictly convex on $[r_1,r_2]$ (by Proposition~\ref{prop:explosive}(ii) under Assumption~\ref{ass:T_convexity}; see Lemma~\ref{lem:T_convex_verify} for primitive sufficient conditions), the following observation suffices.

\emph{Sub-Lemma (revised).} Let $q$ be strictly convex and strictly increasing on $[r_1,r_2]$, with $\bar q = (r_2-r_1)^{-1}\int_{r_1}^{r_2} q(s)\,ds$ and $\tilde q(r)=q(r)-\bar q$. Then: (i) by strict convexity, $\bar q < \tfrac12(q(r_1)+q(r_2))$, hence $\tilde q(r_1)+\tilde q(r_2)>0$; (ii) since $q$ is strictly increasing, $\tilde q(r_1)<\tilde q(r_2)$, so together with (i) we obtain $\tilde q(r_1)<0<\tilde q(r_2)$ and $\tilde q(r_2)>|\tilde q(r_1)|$; (iii) therefore
\[
\int_{r_1}^{r_2}\tilde q(r)\,q'(r)\,dr
=
\frac12\big[\tilde q(r_2)^2-\tilde q(r_1)^2\big]
>0,
\]
and hence $B(r_1,r_2)>0$.

\medskip\noindent\textit{Step 3: Divergence of the $t$-statistic.}
Under the signal-dominance condition $\dmax(T)^2/(T\sigma_u^2) \to \infty$, the standard error of $\hat{\beta}$ is $\mathrm{se}(\hat{\beta}) = O_p(T^{-1/2}\, \sigma_u / \dmax(T))$ (derived from the leading deterministic denominator and the residual variance $\sigma_u^2 + o_p(1)$). Since $\hat{\beta} = B(r_1,r_2)/T + o_p(1/T)$ with $B > 0$:
\[
t_{\hat{\beta}} = \frac{\hat{\beta}}{\mathrm{se}(\hat{\beta})} = O_p\!\left(\frac{B/T}{T^{-1/2}\, \sigma_u/\dmax(T)}\right) = O_p\!\left(\frac{B\,\dmax(T)}{\sigma_u\sqrt{T}}\right).
\]
Therefore $t_{\hat{\beta}}\to+\infty$ whenever $\dmax(T)/(\sigma_u\sqrt{T}) \to \infty$, equivalently under the imposed signal-dominance condition. In particular, fixed $\dmax>0$ does \emph{not} yield divergence as $T\to\infty$; divergence requires $\dmax(T)$ to grow sufficiently fast relative to $\sqrt{T}$. Hence the generalized supremum augmented Dickey--Fuller test rejects with probability approaching one whenever the admissible class contains such a ramp-up window. \qed

\subsection*{Proof of Proposition~\ref{prop:adj_size} (Size of the Adjusted Test)}

Under $b_t = 0$, consider first the detrended log-price implementation. By Definition~\ref{def:adj_test},
\begin{equation*}
	y_t^{\mathrm{adj}}
	=
	y_t-\sum_{s=1}^t\hat\delta_s-\hat{\mathcal T}_t
	=
	d_t + C + \mathcal T_t - \sum_{s=1}^t\hat\delta_s - \hat{\mathcal T}_t.
\end{equation*}
Using $d_t = d_0 + ct + \sum_{s=1}^t(\delta_s + \varepsilon_s)$, this becomes
\begin{equation*}
	y_t^{\mathrm{adj}}
	=
	d_0 + C + ct + \sum_{s=1}^t \varepsilon_s + \sum_{s=1}^t(\delta_s-\hat\delta_s) + (\mathcal T_t-\hat{\mathcal T}_t).
\end{equation*}
Hence, under the uniform-consistency assumptions of the proposition,
\begin{equation*}
	y_t^{\mathrm{adj}}
	=
	d_0 + C + ct + \sum_{s=1}^t \varepsilon_s + o_p(1)
\end{equation*}
uniformly in $t$. After detrending,
\begin{equation*}
	\tilde y_t^{\mathrm{adj}}=\sum_{s=1}^t \varepsilon_s + r_{t,T},
	\qquad
	\sup_{1\le t\le T}|r_{t,T}|=o_p(1).
\end{equation*}
If $Y_t^{(0)}=\sum_{s=1}^t\varepsilon_s$, then by continuous mapping for the GSADF functional applied uniformly over admissible windows, the GSADF statistic computed from $Y^{(0)}+r$ converges in distribution to the GSADF statistic computed from $Y^{(0)}$. Hence the limit distribution coincides with the standard PSY null, and the adjusted detrended log-price test has asymptotic size $\alpha$.

For the price-dividend implementation,
\begin{equation*}
	y_t^{\mathrm{adj,PD}}
	=
	(p_t-d_t)-\hat{\mathcal T}_t
	=
	C+\mathcal T_t-\hat{\mathcal T}_t
	=
	C+o_p(1)
\end{equation*}
uniformly under $b_t=0$. Therefore GSADF applied to a constant plus $o_p(1)$ has the same limit as on a constant, so no technology-driven explosive behavior remains detectable after adjustment. \qed

\subsection*{Proof of Proposition~\ref{prop:adj_power} (Power of the Adjusted Test)}

Under the alternative $p_t=f_t+b_t$, consider first the detrended log-price implementation. By Definition~\ref{def:adj_test},
\begin{equation*}
	y_t^{\mathrm{adj}}
	=
	p_t-\sum_{s=1}^t\hat\delta_s-\hat{\mathcal T}_t
	=
	d_t + C + \mathcal T_t + b_t - \sum_{s=1}^t\hat\delta_s - \hat{\mathcal T}_t.
\end{equation*}
Using $d_t = d_0 + ct + \sum_{s=1}^t(\delta_s+\varepsilon_s)$ and the uniform-consistency conditions,
\begin{equation*}
	y_t^{\mathrm{adj}}
	=
	d_0 + C + ct + \sum_{s=1}^t \varepsilon_s + b_t + o_p(1)
\end{equation*}
uniformly in $t$. Detrending removes $d_0+ct$, leaving
\begin{equation*}
	\tilde y_t^{\mathrm{adj}}
	=
	\sum_{s=1}^t \varepsilon_s + b_t + r_{t,T},
	\qquad
	\sup_{1\le t\le T}|r_{t,T}|=o_p(1).
\end{equation*}
Thus the explosive component $b_t$ remains. If the standard PSY statistic for the technology-free benchmark $\sum_{s=1}^t \varepsilon_s + b_t$ rejects at rate $\omega_T\to\infty$, then any remainder with $\sup_t|r_{t,T}|=o_p(\omega_T)$ is asymptotically negligible for the GSADF functional, so the adjusted test has the same rejection probability limit, namely one.

For the price-dividend implementation,
\begin{equation*}
	y_t^{\mathrm{adj,PD}}
	=
	(p_t-d_t)-\hat{\mathcal T}_t
	=
	C+u_t+b_t+(\mathcal T_t-\hat{\mathcal T}_t)
	=
	C+u_t+b_t+o_p(1).
\end{equation*}
Hence the adjustment removes only the technology component and leaves the explosive bubble term $b_t$ intact. The adjusted price-dividend-ratio test therefore has the same first-order power as the standard PSY test applied to the corresponding technology-free series. \qed

\subsection*{Proof of Proposition~\ref{prop:tv_returns} (Time-Varying Expected Returns)}

Under the horizon-by-horizon specification~\eqref{eq:tv_returns}, the present-value identity becomes
\begin{align}
    f_t &= d_t + \frac{\kappa}{1-\rho} + \sum_{j=0}^{\infty} \rho^{j}\, \E_t[\Delta d_{t+1+j} - r_{t+1+j}] \notag \\
    &= d_t + \frac{\kappa}{1-\rho} + \sum_{j=0}^{\infty} \rho^{j}\, \big[(c + \delta_{t+1+j}) - (\rbar + \phi\, \delta_{t+1+j})\big] \notag \\
    &= d_t + \frac{\kappa}{1-\rho} + \frac{c - \rbar}{1-\rho} + (1 - \phi)\sum_{j=0}^{\infty} \rho^{j}\,\delta_{t+1+j} \notag \\
    &= d_t + C + (1-\phi)\,\mathcal T_t.
\end{align}
Hence
\[
f_t-d_t = C + (1-\phi)\,\mathcal T_t.
\]
If $\phi<1$, then $1-\phi>0$, so the technology component of the log price-dividend ratio has exactly the same sign, monotonicity, and convexity properties as $\mathcal T_t$. Indeed,
\[
\Delta(f_t-d_t) = (1-\phi)\,\Delta\mathcal T_t,
\qquad
\Delta^2(f_t-d_t) = (1-\phi)\,\Delta^2\mathcal T_t.
\]
If $\phi<0$, then $1-\phi>1$, so the level, slope, and curvature of the price-dividend-ratio technology component are amplified by the multiplicative factor $1-\phi$.

For the drift of the log price itself,
\[
\Delta f_t = c + \delta_t + \varepsilon_t + (1-\phi)\,\Delta\mathcal T_t.
\]
Thus any comparison of the drift with $c$ depends on the sign of $\Delta\mathcal T_t$ on the window under consideration; no stronger global drift claim follows from $\phi<1$ alone. \qed

\subsection*{Proof of Proposition~\ref{prop:stochastic_tech} (Stochastic Technology Amplification)}

\medskip\noindent\textit{Part (i): Jumps compound explosive dynamics via persistent posterior levels.}
Under Assumption~\ref{ass:stochastic_tech}, agents observe
\[
\delta_t = g_\star(t)\,\bar\delta + \sigma_\xi\,\xi_t,\qquad \xi_t \sim \mathcal N(0,1)
\]
during the learning window, and form the posterior mean $\hat\delta_t = \E[\bar\delta\mid\mathcal F_t]$. The Kalman filter therefore gives
\[
\hat\delta_t \;=\; \hat\delta_{t-1} + K_t\,\nu_t,\qquad \nu_t := \delta_t - g_\star(t)\,\hat\delta_{t-1},\qquad K_t := \frac{P_{t-1}\,g_\star(t)}{g_\star(t)^2\,P_{t-1}+\sigma_\xi^2},
\]
with posterior variance recursion
\[
P_t \;=\; P_{t-1} - K_t\,g_\star(t)\,P_{t-1}
\;=\;
\frac{P_{t-1}\sigma_\xi^2}{g_\star(t)^2\,P_{t-1}+\sigma_\xi^2},
\]
where $P_{t-1} = \Var(\bar\delta\mid\mathcal F_{t-1})$. The innovation $\nu_t$ satisfies
\[
\E[\nu_t\mid\mathcal F_{t-1}] \;=\; g_\star(t)\,\E[\bar\delta\mid\mathcal F_{t-1}] - g_\star(t)\,\hat\delta_{t-1} + \sigma_\xi\,\E[\xi_t\mid\mathcal F_{t-1}] \;=\; 0,
\]
so $\{\nu_t\}$ is a martingale difference sequence with respect to $\{\mathcal F_t\}$ and is therefore serially uncorrelated.

However, the posterior mean $\hat\delta_t$ is itself a martingale (by the tower property applied to $\hat\delta_t = \E[\bar\delta\mid\mathcal F_t]$), and hence its \emph{levels} are maximally persistent: $\Cov(\hat\delta_t,\hat\delta_{t+h}) = \Var(\hat\delta_t) = P_0 - P_t$, which is nondecreasing in $t$ and bounded by $P_0$. Substituting the posterior mean into the technology present-value operator,
\[
\Tt \;=\; \sum_{j=0}^\infty \rho^j\,\E_t[\delta_{t+1+j}] \;=\; \hat\delta_t\,G_\rho(t),\qquad G_\rho(t) := \sum_{j=0}^\infty \rho^j\,g_\star(t+1+j),
\]
so the innovation in $\Tt$ relative to $\E_{t-1}[\Tt]$ is $G_\rho(t)\,K_t\,\nu_t$, while the total change satisfies
\[
\Delta\Tt \;=\; G_\rho(t)\,K_t\,\nu_t + \hat\delta_{t-1}\,\Delta G_\rho(t),
\]
where $G_\rho(t)\,K_t\,\nu_t$ is the surprise component (an MDS) and $\hat\delta_{t-1}\,\Delta G_\rho(t)$ is the predictable drift. Because $\hat\delta_s$ for $s>t$ satisfies $\E_t[\hat\delta_s] = \hat\delta_t$ (martingale property), the positive surprise component propagates into \emph{all} future values of $\Tt$ in expectation: a single $\nu_t>0$ raises $\E_t[\Tt]$ for every subsequent $t'>t$ within the adoption window, because $\E_t[\hat\delta_{t'}]G_\rho(t') = \hat\delta_t\,G_\rho(t')$ and $\hat\delta_t$ has been permanently revised upward.

During the early adoption phase $[T_1,T_1+\tau]$, $K_t$ is relatively high because $P_{t-1}$ is close to $P_0$ when little learning has yet taken place. Moreover, $G_\rho(t)>0$ throughout the adoption window and is bounded away from zero on a non-empty subinterval near the peak. Hence the early-adoption innovation term is non-trivially amplified relative to the post-maturation regime where $G_\rho(t)\to0$. The distinction between innovation-level serial uncorrelation and posterior-level persistence is the precise economic content: PSY's rejection probability is driven by the \emph{level path} of $\Tt$, not by autocorrelation of $\Delta\Tt$, so a martingale-difference innovation sequence that raises the level persistently is sufficient to amplify spurious explosiveness. The finite-sample quantitative effect is documented in Panel~A of Table~\ref{tab:mc_stochastic}.

\medskip\noindent\textit{Part (ii): Asymptotic negligibility of the learning channel at PSY scale.}
Under Assumption~\ref{ass:stochastic_tech}, the fundamental price innovation decomposes as $\Delta f_t = \mu_t^{\mathrm{det}} + \varepsilon_t + K_t\, \nu_t\, G_\rho(t)$, where $\mu_t^{\mathrm{det}}$ is the deterministic drift from Proposition~\ref{prop:dynamics} and the last term captures the Bayesian learning innovation ($\nu_t$ is the Kalman innovation, a martingale difference with conditional variance $g_{\star}(t)^2 P_{t-1} + \sigma_\xi^2$, mutually independent of $\varepsilon_t$). We establish that under the local parametrization of Assumption~\ref{ass:local_framework} with a fixed discount factor $\rho\in[0,1)$, the partial-sum martingale $S_t \equiv \sum_{s\le t} K_s\,\nu_s\,G_\rho(s)$ is asymptotically negligible at PSY scale, so the gsADF limit law coincides with the deterministic-$\delta$ limit law of Theorem~\ref{thm:contaminated_limit}. Any $\sigma_\xi$-dependence of the test's rejection probability is therefore a finite-sample, not a first-order asymptotic, phenomenon.

\emph{Step (ii.a): Scaling of $g_\star$, $K_s$, and $G_\rho$.} Assumption~\ref{ass:stochastic_tech} normalizes $g_\star$ as a discrete density on the adoption--maturation window $[T_1,T_2]$. Under Assumption~\ref{ass:local_framework} this window has length $\Theta(T)$, so the natural local specification is
\begin{equation*}
g_{\star,T}(s) \;=\; T^{-1}\,h_\star(s/T) \;+\; o(T^{-1}) \qquad \text{uniformly in } s\in[T_1,T_2],
\end{equation*}
for a continuous shape function $h_\star:[\lambda_1,\lambda_2]\to\mathbb{R}_+$ with $\int_{\lambda_1}^{\lambda_2} h_\star(u)\,du = 1$. To keep the deterministic local channel $\dmax(T)\,q(\cdot) = c_\delta T^{-1/2}\,q(\cdot)$ nondegenerate and on the same scale, prior moments are nested as $\mu_\delta = O(\sqrt T)$ and $\sigma_\delta = O(\sqrt T)$, i.e.\ $P_0 = T\tau_\delta^2$ for some $\tau_\delta\in(0,\infty)$; the innovation variance $\sigma_\xi$ is kept fixed in $(0,\infty)$.

The Kalman precision recursion $P_s^{-1} = P_0^{-1} + \sigma_\xi^{-2}\sum_{j\le s} g_{\star,T}(j)^2$ then gives, uniformly in $r\in[\lambda_1,\lambda_2]$,
\begin{equation*}
P_{\lfloor Tr\rfloor,T} \;=\; T\,p(r) + o(T),\qquad p(r) = \Big[\tau_\delta^{-2} + \sigma_\xi^{-2}\!\int_{\lambda_1}^{r\wedge\lambda_2} h_\star(u)^2\,du\Big]^{-1},
\end{equation*}
and hence $K_{\lfloor Tr\rfloor,T} = p(r)\,h_\star(r)/\sigma_\xi^2 + o(1) = O(1)$. For fixed $\rho\in[0,1)$, the geometrically weighted kernel satisfies
\begin{equation*}
G_{\rho,T}(s) \;=\; \sum_{j\ge 0} \rho^j\, g_{\star,T}(s+1+j) \;=\; T^{-1}\,\frac{h_\star(s/T)}{1-\rho} \;+\; o(T^{-1}),
\end{equation*}
uniformly on $[T\lambda_1,T\lambda_2]$. Consequently
\begin{equation*}
K_s\,G_{\rho,T}(s) \;=\; T^{-1}\,\kappa_\rho(s/T) + o(T^{-1}),\qquad \kappa_\rho(u) = \frac{p(u)\,h_\star(u)^2}{\sigma_\xi^2\,(1-\rho)}.
\end{equation*}

\emph{Step (ii.b): Negligibility of $T^{-1/2} S_{\lfloor Tr\rfloor}$.} Because $\{\nu_s\}$ is a martingale difference array with conditional variance $g_{\star,T}(s)^2 P_{s-1,T} + \sigma_\xi^2 = \sigma_\xi^2 + o(1)$ (the first term is $O(T^{-1})$), the predictable quadratic variation of $T^{-1/2} S_{\lfloor Tr\rfloor}$ satisfies
\begin{align*}
T^{-1}\sum_{s\le\lfloor Tr\rfloor}\!\big(K_s\,G_{\rho,T}(s)\big)^2 \E[\nu_s^2\mid\mathcal F_{s-1}]
&= T^{-1}\sum_{s\le\lfloor Tr\rfloor}\!\big(T^{-1}\kappa_\rho(s/T)\big)^2\!\big(\sigma_\xi^2 + o(1)\big) \\
&= O(T^{-2})\cdot T\cdot O(1) \;=\; O(T^{-2}) \;\to\; 0.
\end{align*}
By Chebyshev's inequality applied to the martingale $T^{-1/2} S_{\lfloor Tr\rfloor}$, $T^{-1/2} S_{\lfloor Tr\rfloor} \xrightarrow{L^2} 0$, and tightness of the martingale \citep{Phillips1987} yields
\begin{equation*}
T^{-1/2}\, S_{\lfloor Tr\rfloor} \;\Rightarrow\; 0 \qquad \text{in } D[0,1].
\end{equation*}

\emph{Step (ii.c): gsADF limit unchanged.} Decomposing the detrended process as $\tilde y_t = \tilde y_t^{\mathrm{det}} + S_t$ on any subsample window, Step (ii.b) and the Continuous Mapping Theorem applied to the sample-ratio functional $\mathcal F$ of~\eqref{eq:contaminated_F} imply that $\gsadf_T$ converges to the same limit as the deterministic-$\delta$ case, namely the limit of Theorem~\ref{thm:contaminated_limit}. Hence $\sigma_\xi$ does \emph{not} appear in the first-order limit law.

\emph{Remark on finite-sample behavior.} The negligibility result holds at PSY scale $T^{-1/2}$ under fixed $\rho$. At finite $T$, the learning-innovation contribution $S_t$ is of root-mean-square order $T^{-1}\,\big(\int_0^r\kappa_\rho(u)^2\,du\big)^{1/2}\,\sigma_\xi$, which is nonzero and depends on $\sigma_\xi$ through the weight $\kappa_\rho$. This pathwise contribution can shift the rejection probability in either direction depending on the drift $c_\delta\,q(\cdot)$, the window $\mathcal R$, and the specific empirical implementation. Accordingly we do not state a monotone size result as a theorem; the empirical direction is documented in Table~\ref{tab:mc_stochastic}, where Panel~A (price-dividend ratio) exhibits a monotonic increase in rejection rate with $\sigma_\delta$ and Panel~B (detrended log price) shows near-invariance. A nondegenerate Brownian-functional limit featuring $\sigma_\xi$ arises only under a local-to-unity parametrization $\rho_T = 1-c_\rho/T$; that extension is outside the scope of this paper, which takes $\rho$ as the fixed steady-state log-linearization constant of \citet{Campbell1988}.

\medskip\noindent\textit{Part (iii): Adjusted test validity.}
For the detrended log-price implementation, Definition~\ref{def:adj_test} uses the full adjustment
\[
y_t^{\mathrm{adj}}=y_t-\sum_{s=1}^t\hat\delta_s-\hat{\mathcal T}_t.
\]
Hence the conclusion follows from Proposition~\ref{prop:adj_size}: after subtracting both the cumulative technology shocks and the present-value term, the adjusted series has the same first-order asymptotics as the technology-free benchmark. For the log price-dividend-ratio implementation, Definition~\ref{def:adj_test} subtracts only $\hat{\mathcal T}_t$, and Proposition~\ref{prop:adj_size} yields $y_t^{\mathrm{adj}}=C+o_p(1)$ under $b_t=0$. This proves the implementation-specific validity claim. \qed

\subsection*{Proof of Lemma~\ref{lem:T_recursion} (Recursion of $\Tt$)}

By definition,
\begin{equation*}
    \Ttm = \sum_{j=0}^{\infty} \rho^{j}\, \delta_{t+j} = \delta_t + \sum_{j=1}^{\infty} \rho^{j}\, \delta_{t+j}.
\end{equation*}
Separating the first term and re-indexing:
\begin{equation*}
    \Ttm = \delta_t + \rho \sum_{j=0}^{\infty} \rho^{j}\, \delta_{t+1+j} = \delta_t + \rho\, \Tt.
\end{equation*}
Solving for $\Tt$:
\begin{equation*}
    \Tt = \rho^{-1}\, \Ttm - \rho^{-1}\, \delta_t. 
\end{equation*}

\qed

\subsection*{Proof of Lemma~\ref{lem:T_convex_verify} (Sufficient Conditions for Assumption~\ref{ass:T_convexity})}

\noindent The argument relies on the identity
\begin{equation}
	\Delta^{2}\mathcal{T}_t \;=\; \sum_{j=1}^{\infty} \rho^{\,j-1}\,\Delta^{2}\delta_{t+j},
	\label{eq:D2T_identity}
\end{equation}
which was derived in the proof of Proposition~\ref{prop:explosive} above. Hence $\mathcal{T}_t$ is (strictly) convex at $t$ iff the $\rho$-weighted sum of future second differences of $\delta$ is (strictly) positive.

\emph{Case (i): Uniformly convex smooth ramp-up.} Suppose $g\in C^{2}[0,\tau]$ with $m := \inf_{x\in[0,\tau]} g''(x) > 0$. For early ramp-up indices, the mean-value theorem applied twice gives $\Delta^{2}\delta_{t+j} = g''(\xi_{t,j}) \ge m$ for some $\xi_{t,j}\in(0,\tau)$ whenever $t+j$ remains on the ramp-up branch. Because $\delta$ has bounded support, the remaining terms in \eqref{eq:D2T_identity} form a geometrically discounted tail. For $t$ sufficiently close to $T_1$, the positive ramp-up contribution therefore dominates that discounted tail, so $\Delta^{2}\mathcal{T}_t>0$. Hence Assumption~\ref{ass:T_convexity} holds on some non-empty initial subinterval $[T_1,T_1+\tau^\star)$ with $\tau^\star\in(0,\tau]$.

\emph{Case (ii): Smooth profiles with a convex left tail (Gaussian, Beta with $\alpha,\beta \geq 2$, logistic-S).} Suppose $g\in C^{2}[0,\tau]$ has $g''(0) \geq m_0 > 0$. By continuity of $g''$, there exists $\bar\tau \in (0,\tau]$ such that $g''(x) \geq m_0/2 > 0$ on $[0,\bar\tau]$. Applying the argument of Case~(i) on this left-tail subinterval yields $\Delta^{2}\mathcal{T}_t>0$ on a non-empty initial subinterval $[T_1,T_1+\tau^\star)$ for some $\tau^\star\in(0,\bar\tau]$. This establishes Assumption~\ref{ass:T_convexity} for all smooth profiles whose second derivative is strictly positive in a right neighbourhood of the origin.

\emph{Remark on the triangular profile.} The piecewise-linear triangular profile $\delta_t = \dmax\,g_\triangle(t-T_1)$, while useful as a benchmark in the Monte Carlo design, does \emph{not} belong to $C^{2}$ and fails Assumption~\ref{ass:T_convexity} on the ramp-up. Direct computation with $\Delta^{2}\delta_s$ concentrated at the three kinks $\{T_1+1, T_1+\tau+1, T_2+1\}$ shows that for interior ramp-up dates $t \in \{T_1+1,\dots,T_1+\tau-1\}$,
\[
\Delta^{2}\mathcal{T}_t \;=\; \dmax\,\rho^{T_1+\tau-t}\!\left[-\tfrac{1}{\tau} + \tfrac{\rho^{T_2-T_1-\tau}-1}{T_2-T_1-\tau}\right] \;<\; 0,
\]
since $\rho \in (0,1)$ makes both bracketed terms negative. Economically, the triangular profile's sharp slope discontinuity at the peak concentrates all curvature at a single date, overwhelming the positive second-difference atom at $T_1+1$ under $\rho$-discounting of the later negative atom. For this profile, the price-dividend-ratio theory of Proposition~\ref{prop:pd_limit} does not apply through Assumption~\ref{ass:T_convexity}; the detrended-price theory of Theorem~\ref{thm:contaminated_limit}, which relies on convexity of the integrated hump $g$ (not of $\mathcal{T}_t$), remains valid for the triangular profile. \qed

\subsection*{Proof of Theorem~\ref{thm:size_distortion} (Size Distortion)}

We prove each part of the theorem.

\medskip\noindent\textit{Part (i): $\psi_T > \alpha$ for all $T$ sufficiently large.}
This is established as a corollary of Theorem~\ref{thm:contaminated_limit}(e) (for detrended log prices) and Proposition~\ref{prop:pd_limit} (for the log price-dividend ratio), whose detailed proofs are given above. In both cases, the technology shock introduces a deterministic component that ultimately drives the GSADF statistic above its null benchmark, so the rejection probability exceeds $\alpha$.

\medskip\noindent\textit{Part (ii): $\psi_T \to 1$ along sequences with $\sqrt{T}\,\dmax(T)\to\infty$.}
For detrended log prices, let $c_\delta(T):=\sqrt{T}\,\dmax(T)\to\infty$. Fix a window $(r_1^\star,r_2^\star)\subset(\lambda_1,\lambda_1+\kappa(\lambda_2-\lambda_1))$ with $r_w^\star\ge r_0$. By the Sub-Lemma in the proof of Proposition~\ref{prop:pd_limit}, applied to $q=g$,
\[
\mathcal B^\star
:=
\frac{\int_{r_1^\star}^{r_2^\star}\tilde g(r)\,g'(r)\,dr}
{\sigma\big(\int_{r_1^\star}^{r_2^\star}\tilde g(r)^2\,dr\big)^{1/2}}
>0.
\]
On this window, the corresponding Dickey--Fuller functional satisfies
\[
\mathcal F(r_1^\star,r_2^\star;c_\delta(T))
=
c_\delta(T)\,\mathcal B^\star + O_p(1),
\]
so $\mathcal F(r_1^\star,r_2^\star;c_\delta(T))\to\infty$ in probability. Since the GSADF supremum dominates this fixed-window statistic, Theorem~\ref{thm:contaminated_limit}(e) yields $\psi_T\to 1$ for the detrended-price implementation.

For the log price-dividend ratio, Proposition~\ref{prop:pd_limit} gives the analogous one-window divergence under the same signal-growth condition, up to the fixed scale factor $\sigma_u>0$. Hence the associated BSADF/GSADF statistic also diverges, and $\psi_T\to 1$.

\medskip\noindent\textit{Remark on comparative statics in $\rho$.}
For the log price-dividend ratio, $f_t - d_t = C + \Tt$ where $\Tt = \sum_{j=0}^{\infty} \rho^j \delta_{t+1+j}$. Since $\delta_{t+1+j} \geq 0$ for $j$ in the adoption window,
\[
\frac{\partial \Tt}{\partial \rho}
=
\sum_{j=1}^{\infty} j\, \rho^{j-1}\, \delta_{t+1+j}
>
0
\]
during the adoption phase. A higher $\rho$ therefore amplifies the level of $\Tt$ and strengthens the local signal in the price-dividend-ratio implementation. For detrended log prices, the effect is more nuanced because the detrended hump $g(r)$ depends both on amplitude and on temporal concentration; a lower $\rho$ can sharpen the hump enough to offset the lower level.

\begin{assumption}[Self-Similar Duration Family]\label{ass:duration_family}
For the duration comparative static, the normalized local shape belongs to a dilation family
\[
h_L(r)=h_0\!\left(\frac{r-\lambda_1}{L}\right)\mathbf{1}\{\lambda_1\le r\le \lambda_1+L\},
\qquad L\in(0,1-\lambda_1],
\]
where $h_0:[0,1]\to\mathbb{R}_+$ is continuous, $\max_{u\in[0,1]}h_0(u)=1$, strictly increasing on $[0,\kappa]$, and strictly decreasing on $[\kappa,1]$.
\end{assumption}

\medskip\noindent\textit{Remark on comparative statics in adoption duration.}
Under Assumption~\ref{ass:duration_family}, increasing the duration parameter $L$ stretches the same normalized shape $h_0$ over a longer interval via $h_L(r)=h_0((r-\lambda_1)/L)\mathbf 1\{\lambda_1\le r\le \lambda_1+L\}$. On ramp-up windows this enlarges the span over which the deterministic component departs from linearity, so once the signal is sufficiently strong the associated ADF functional becomes more likely to dominate the stochastic term. This is the intended comparative-static sense in which a longer adoption window raises the induced rejection probability.

\medskip\noindent\textit{Part (v): $\psi_T \to 1$ as $\dmax \to \infty$ (price-dividend ratio).}
This follows directly from Proposition~\ref{prop:pd_limit}. By Proposition~\ref{prop:pd_limit} (as patched in this revision), under signal-dominance $\dmax(T)^2/(T\sigma_u^2)\to\infty$, the $t$-statistic on the price-dividend ratio diverges, and the GSADF rejects with probability $\to 1$. \qed

\begin{lemma}[Augmentation lags negligible under i.i.d.\ innovations]\label{lem:adf_augmentation}
Under Assumptions~\ref{ass:dividend}, \ref{ass:tech_shock}, \ref{ass:const_returns}, and \ref{ass:local_framework}, with i.i.d.\ $\varepsilon_t$ and any fixed $K \geq 0$, the $K$-augmented Dickey--Fuller $t$-statistic $\mathrm{ADF}_K(r_1,r_2)$ and the $K=0$ statistic $\mathrm{DF}(r_1,r_2)$ satisfy
\[
\mathrm{ADF}_K(r_1,r_2)-\mathrm{DF}(r_1,r_2)=o_p(1)
\]
uniformly over $(r_1,r_2)\in \mathcal R=\{(r_1,r_2): r_2-r_1\ge r_0,\ r_2\in[r_0,1],\ r_1\in[0,r_2-r_0]\}$.
\end{lemma}

\begin{proof}
Because $r_2-r_1\ge r_0>0$, the admissible-window design matrices are uniformly well behaved over $\mathcal R$. Under i.i.d.\ innovations with finite variance, the lagged differences $\Delta y_{t-j}$ are stationary, mean-zero, and have finite second moments, so for fixed $K$ the OLS coefficients on the augmentation lags are $O_p(T^{-1/2})$ uniformly on $\mathcal R$. By the standard Frisch--Waugh decomposition, their contribution to the unit-root coefficient is $o_p(T^{-1})$, while their effect on residual-variance estimation is $o_p(1)$. Uniform tightness of the lag-coefficient sequence therefore yields $\mathrm{ADF}_K(r_1,r_2)-\mathrm{DF}(r_1,r_2)=o_p(1)$ uniformly on $\mathcal R$.
\end{proof}

\subsection*{Proof of Theorem~\ref{thm:contaminated_limit} (Contaminated Brownian Motion Limit)}

\medskip\noindent\textit{Part (a): Weak convergence.}
Under Assumption~\ref{ass:local_framework} with $\dmax(T) = c_\delta/\sqrt{T}$, the technology present-value term satisfies $\mathcal{T}_{\lfloor Tr \rfloor, T} = \sum_{j=0}^\infty \rho^j \delta_{\lfloor Tr \rfloor+1+j,T} \leq \frac{1}{1-\rho}\dmax(T) = O(T^{-1/2})$, which is asymptotically negligible. In particular, $T^{-1/2}\,\mathcal{T}_{\lfloor Tr \rfloor, T} = O(T^{-1}) \to 0$ uniformly in $r \in [0,1]$, so this term contributes nothing to the Functional Central Limit Theorem limit and is absorbed into the $o_p(1)$ remainder of the scaled process. The dominant deterministic component is the cumulated shock $M_t = \sum_{s=1}^t \delta_{s,T}$:
\begin{equation*}
M_{\lfloor Tr \rfloor} = \sum_{s=1}^{\lfloor Tr \rfloor} \frac{c_\delta}{\sqrt{T}}\, h(s/T)\, \mathbf{1}\{\lambda_1 \leq s/T \leq \lambda_2\} = \sqrt{T}\, c_\delta \int_0^r h(s)\, \mathbf{1}\{\lambda_1 \leq s \leq \lambda_2\}\, ds + o(\sqrt{T})
\end{equation*}
by Riemann sum approximation. Writing $G(r) = \int_0^r h(s)\, \mathbf{1}\{\lambda_1 \leq s \leq \lambda_2\}\, ds$, the detrended fundamental price is
\[
\tilde{f}_t = \sum_{s=1}^t \varepsilon_s + m_{t,T},
\]
where $m_{t,T}$ denotes the linearly detrended version of $M_t + \mathcal{T}_{t,T}$. Let $S_t:=\sum_{s=1}^t \varepsilon_s$ and $G_T(r):=c_\delta^{-1}T^{-1/2}m_{\lfloor Tr\rfloor,T}$. By Slutsky's theorem applied to $(T^{-1/2}S_{\lfloor Tr\rfloor},\, c_\delta G_T(r))$, the first component converges weakly in $C[0,1]$ to $\sigma W(r)$ by Donsker's FCLT. The second converges uniformly to $c_\delta g(r)$ by continuity of $g$ and Riemann-sum approximation, since the sample linear-detrending coefficients $(\hat\alpha_{0,T},\hat\alpha_{1,T})\to(\alpha_0,\alpha_1)$ by continuity of the $L^2$ projection defining $(\alpha_0,\alpha_1)$ in Theorem~\ref{thm:contaminated_limit}. Joint convergence in $C[0,1]$ therefore follows from Slutsky, and
\[
T^{-1/2} \tilde{f}_{\lfloor Tr \rfloor} \Rightarrow \sigma W(r) + c_\delta\, g(r) = Z_{c_\delta}(r). 
\]

\medskip\noindent\textit{Part (b): Augmented Dickey--Fuller functional limit.}
The augmented Dickey--Fuller regression on the subsample $[\lfloor r_1 T \rfloor, \lfloor r_2 T \rfloor]$ takes the form \eqref{eq:ADF_reg} with $y_t = \tilde{f}_t$. By Lemma~\ref{lem:adf_augmentation}, it suffices to derive the limit for the $K=0$ Dickey--Fuller $t$-statistic; the fixed-$K$ augmented version differs by $o_p(1)$ uniformly on $\mathcal R$. Under standard unit root asymptotics \citep{Phillips1987}, the sample moments of the detrended series converge to functionals of the limit process:
\begin{align*}
T^{-2} \sum_{t=\lfloor r_1 T \rfloor}^{\lfloor r_2 T \rfloor} \tilde{f}_{t-1}^2 &\Rightarrow \int_{r_1}^{r_2} \tilde{Z}_{c_\delta}^2(r)\, dr, \\
T^{-1} \sum_{t=\lfloor r_1 T \rfloor}^{\lfloor r_2 T \rfloor} \tilde{f}_{t-1}\, \Delta \tilde{f}_t &\Rightarrow \int_{r_1}^{r_2} \tilde{Z}_{c_\delta}(r)\, dZ_{c_\delta}(r),
\end{align*}
where $\tilde{Z}_{c_\delta}(r)$ denotes the demeaned process on $[r_1, r_2]$. To evaluate the stochastic integral, write $Z := Z_{c_\delta}$ and $\bar{Z} := (r_2-r_1)^{-1}\int_{r_1}^{r_2} Z(s)\,ds$, so that $\tilde{Z}(r) = Z(r) - \bar{Z}$. By It\^{o}'s lemma applied to $Z(r)^2$ (the quadratic-variation term comes only from the Brownian component $\sigma W$),
\begin{equation*}
\tfrac12\, d\!\left[Z(r)^2\right] \;=\; Z(r)\,dZ(r) \;+\; \tfrac12\,\sigma^2\,dr,
\end{equation*}
so
\begin{equation*}
\int_{r_1}^{r_2} Z(r)\, dZ(r) \;=\; \tfrac12\!\left[Z(r_2)^2 - Z(r_1)^2\right] \;-\; \tfrac12\,\sigma^2\,(r_2-r_1).
\end{equation*}
Subtracting the demeaning term and using $\int_{r_1}^{r_2} dZ(r) = Z(r_2) - Z(r_1)$ gives
\begin{equation*}
\int_{r_1}^{r_2} \tilde{Z}(r)\, dZ(r) \;=\; \tfrac12\!\left[Z(r_2)^2 - Z(r_1)^2\right] \;-\; \tfrac12\,\sigma^2\,(r_2-r_1) \;-\; \bar{Z}\,[Z(r_2) - Z(r_1)].
\end{equation*}
Residual variance $\hat{\sigma}_{r_1,r_2}^2 \to_p \sigma^2$ by standard ergodicity arguments under i.i.d.\ $\varepsilon_t$. The denominator in \eqref{eq:contaminated_F} is strictly positive almost surely, since $Z_{c_\delta}(r)=\sigma W(r)+c_\delta g(r)$ is almost surely non-affine on every admissible window and Brownian motion is almost surely not linear on any non-degenerate interval. Writing $r_w=r_2-r_1$ and substituting into the Dickey--Fuller ratio gives the representation in \eqref{eq:contaminated_F}, establishing part~(b).

\medskip\noindent\textit{Part (c): Generalized supremum augmented Dickey--Fuller limit.}
Define $H:C[0,1]\to\mathbb R$ by
\[
H(z)=\sup_{(r_1,r_2)\in\mathcal R} F_z(r_1,r_2),
\]
where $F_z(r_1,r_2)$ is the functional in \eqref{eq:contaminated_F} with $z$ in place of $Z_{c_\delta}$.

\smallskip\noindent\emph{Step 1 (compactness).} The set
\[
\mathcal R=\{(r_1,r_2): r_2-r_1\ge r_0,\ r_2\in[r_0,1],\ r_1\in[0,r_2-r_0]\}
\]
is closed and bounded in $\mathbb R^2$, hence compact.

\smallskip\noindent\emph{Step 2 (joint continuity of $F_z$).} Fix $z\in C[0,1]$ that is non-affine on every admissible window. The numerator and denominator in \eqref{eq:contaminated_F} are continuous in $(r_1,r_2)$ because they are formed from endpoint evaluations and integrals of continuous functions over compact intervals whose endpoints vary continuously. The denominator is strictly positive on $\mathcal R$ for such $z$, and by compactness it is bounded away from zero. Hence $F_z$ is jointly continuous on $\mathcal R$.

\smallskip\noindent\emph{Step 3 (continuity of $H$ at non-affine paths).} If $z_n\to z$ uniformly and $z$ is non-affine on every admissible window, then the endpoint and integral terms defining $F_{z_n}$ converge uniformly to those defining $F_z$ on $\mathcal R$. Because the denominator for $z$ is uniformly bounded away from zero, $F_{z_n}\to F_z$ uniformly on $\mathcal R$, and therefore $H(z_n)\to H(z)$.

\smallskip\noindent\emph{Step 4 (a.s. non-affine sample paths).} The limit process is $Z_{c_\delta}(r)=\sigma W(r)+c_\delta g(r)$. Since $h$ is continuous, $g$ is $C^1$. Brownian paths are almost surely nowhere differentiable and hence almost surely non-affine on every non-degenerate interval; adding the $C^1$ path $c_\delta g$ preserves this property. Thus $Z_{c_\delta}$ is almost surely non-affine on every admissible window.

\smallskip\noindent\emph{Step 5 (continuous mapping).} By Part~(a), $T^{-1/2}\tilde f_{\lfloor T\cdot\rfloor}\Rightarrow Z_{c_\delta}(\cdot)$ in $C[0,1]$. By Steps 1--4, $H$ is continuous at $Z_{c_\delta}$ almost surely. The continuous mapping theorem therefore yields
\[
\gsadf_T
=
H\!\left(T^{-1/2}\tilde f_{\lfloor T\cdot\rfloor}\right)
\Rightarrow
H(Z_{c_\delta})
=
\sup_{(r_1,r_2)\in\mathcal R}\mathcal F(r_1,r_2;c_\delta).
\]

\medskip\noindent\textit{Part (d): Size distortion.}
\emph{Step 1 (strict convexity of $g$ on a ramp-up window).} On the ramp-up sub-interval $(\lambda_1,\lambda_1+\kappa(\lambda_2-\lambda_1))$, we have $G'(r)=h(r)$, and Assumption~\ref{ass:local_framework} makes $h$ strictly increasing there. Hence $G$ is strictly convex on that interval. Subtracting the affine function $\alpha_0+\alpha_1 r$ preserves strict convexity, so $g$ is strictly convex there.

\smallskip\noindent\emph{Step 2 (one-window pathwise divergence).} Fix any window $(r_1^\star,r_2^\star)\subset(\lambda_1,\lambda_1+\kappa(\lambda_2-\lambda_1))$ with $r_w^\star\ge r_0$. By the Sub-Lemma in the proof of Proposition~\ref{prop:pd_limit}, applied to $q=g$, we have
\[
\int_{r_1^\star}^{r_2^\star}\tilde g(r)\,g'(r)\,dr>0.
\]
Hence the deterministic term in the numerator of $\mathcal F(r_1^\star,r_2^\star;c)$ is positive and of order $c^2$, while the denominator is of order $c$ because it is the square root of a quadratic form in $\sigma W + cg$. Therefore $\mathcal F(r_1^\star,r_2^\star;c)\to +\infty$ in probability as $c\to\infty$.

\smallskip\noindent\emph{Step 3 (supremum divergence).} Since
\[
\sup_{(r_1,r_2)\in\mathcal R}\mathcal F(r_1,r_2;c)\ge \mathcal F(r_1^\star,r_2^\star;c),
\]
we also have $\sup_{(r_1,r_2)\in\mathcal R}\mathcal F(r_1,r_2;c)\to\infty$ in probability as $c\to\infty$.

\smallskip\noindent\emph{Step 4 (existence of $c_\delta^\star$).} Define
\[
\Psi(c):=P\!\left(\sup_{(r_1,r_2)\in\mathcal R}\mathcal F(r_1,r_2;c)>cv_\alpha\right).
\]
Step 3 implies $\Psi(c)\to 1$ as $c\to\infty$. Hence there exists at least one $c_\delta^\star>0$ such that $\Psi(c_\delta^\star)>\alpha$. By the continuity assumption at this $c_\delta^\star$, the threshold is well defined; and since the deterministic component shifts the statistic to the right as $c$ increases, $\Psi$ is nondecreasing, so $\Psi(c)\ge \Psi(c_\delta^\star)>\alpha$ for all $c>c_\delta^\star$.

\smallskip\noindent\emph{Step 5 (conclusion).} By Part~(c),
\[
\lim_{T\to\infty}P(\gsadf_T>cv_\alpha)=\Psi(c_\delta),
\]
so for every $c_\delta>c_\delta^\star$,
\[
\lim_{T\to\infty}P(\gsadf_T>cv_\alpha)>\alpha.
\]

\medskip\noindent\textit{Part (e): Divergence.}
Let $a_T:=\sqrt{T}\,\dmax(T)\to\infty$, and keep the same ramp-up window $[r_1^\star,r_2^\star]$ as in Part~(d). From Part~(a),
\[
\tilde f_t
=
\sum_{s=1}^t\varepsilon_s + m_{t,T},
\qquad
m_{t,T}=T\,\dmax(T)\,g(t/T)+o\!\big(T\,\dmax(T)\big)
\]
uniformly on that window. The stochastic partial sums satisfy $\sup_t\left|\sum_{s=1}^t\varepsilon_s\right|=O_p(\sqrt{T})$, so the deterministic component dominates because $a_T\to\infty$.

Consider first the $K=0$ ADF regression on $[r_1^\star,r_2^\star]$, and write
\[
\widetilde g(r)
=
g(r)-\frac{1}{r_w}\int_{r_1^\star}^{r_2^\star}g(s)\,ds.
\]
The same Riemann-sum calculations used in Proposition~\ref{prop:pd_limit} give
\begin{align*}
\sum ( \tilde f_{t-1}-\bar{\tilde f}_{-1})^2
&=
T^3\dmax(T)^2 \int_{r_1^\star}^{r_2^\star}\widetilde g(r)^2\,dr
+O_p(T^2 a_T)+O_p(T^2),\\
\sum (\tilde f_{t-1}-\bar{\tilde f}_{-1})(\Delta\tilde f_t-\overline{\Delta\tilde f})
&=
T^2\dmax(T)^2 \int_{r_1^\star}^{r_2^\star}\widetilde g(r)g'(r)\,dr
+O_p(T a_T)+O_p(T).
\end{align*}
Since $a_T\to\infty$, the deterministic terms dominate, and therefore
\[
T\,\hat\beta_{r_1^\star,r_2^\star}\plim
B^\star
:=
\frac{\int_{r_1^\star}^{r_2^\star}\widetilde g(r)g'(r)\,dr}
{\int_{r_1^\star}^{r_2^\star}\widetilde g(r)^2\,dr}
>0.
\]
Moreover,
\[
\mathrm{se}(\hat\beta_{r_1^\star,r_2^\star})
=
O_p\!\left(\frac{1}{T^{3/2}\dmax(T)}\right),
\]
so
\[
t_{\hat\beta,r_1^\star,r_2^\star}
=
\frac{\hat\beta_{r_1^\star,r_2^\star}}{\mathrm{se}(\hat\beta_{r_1^\star,r_2^\star})}
=
a_T\,\frac{B^\star}{\sigma}+o_p(a_T)
\;\overset{p}{\longrightarrow}\;
\infty.
\]
By Lemma~\ref{lem:adf_augmentation}, adding any fixed number $K$ of augmentation lags changes the $t$-statistic only by $o_p(1)$. Since $\gsadf_T$ is the supremum over admissible windows and includes $[r_1^\star,r_2^\star]$, it follows that
\[
\gsadf_T\overset{p}{\longrightarrow}\infty,
\qquad
P(\gsadf_T>cv_\alpha)\to 1.
\] \qed

%

\clearpage
\clearpage
\section{Monte Carlo Simulations}
\label{sec:simulation}

\subsection{Design}
\label{subsec:mc_design}

We assess the finite-sample size properties of the PSY test under the technology-augmented data-generating process. The baseline calibration (Table~\ref{tab:calibration}) uses parameter values consistent with postwar U.S.\ equity market data: $\rho = 0.95$, $c = 0.02$, $\rbar = 0.06$, $\sigma_\varepsilon = 0.10$, $T = 300$, with a triangular technology shock (Remark~\ref{rmk:triangular}) on a 120-period adoption window. The PSY minimum window fraction follows $r_0 = 0.01 + 1.8/\sqrt{T}$, with $K=1$ lag and 2,000 Monte Carlo critical-value replications.


We conduct three experiments ($M = 200$ replications each): (A)~rejection rate vs.\ $\dmax \in \{0, 0.02, \ldots, 0.20\}$ for detrended log prices ($\rho = 0.95$); (B)~the same for the log price-dividend ratio with first-order autoregressive (AR(1)) noise ($\phi = 0.95$, $\sigma_u = 0.15\sqrt{1-\phi^2}$); (C)~rejection rate vs.\ $\rho \in \{0.90, \ldots, 0.99\}$ for detrended log prices ($\dmax = 0.15$). Each replication computes both the unadjusted rejection indicator and the technology-adjusted indicator (PSY applied after removing $\Tt$).

\subsection{Illustrative Example}
\label{subsec:illustration}

Figure~\ref{fig:components} illustrates the mechanism with a single simulated path ($\dmax = 0.12$). The technology shock $\delta_t$ is hump-shaped; the present-value term $\Tt$ rises sharply during adoption; the implied drift $\mu_t^*$ inherits this shape, generating a permanent price-level increase. Applying PSY to this path (Figure~\ref{fig:psy_with_tech}), the BSADF statistic exceeds the 5\% critical value during the adoption phase (GSADF statistic $= 3.890 > 2.148$), spuriously detecting a ``bubble'' covering 79\% of the technology window. After removing the technology component (Figure~\ref{fig:psy_no_tech}), the date-stamping statistic stays well below the critical value throughout (test statistic $= 1.120$).


\subsection{Results}
\label{subsec:mc_results}

Table~\ref{tab:mc_results} reports the results. Panel~A (detrended log prices): the rejection rate rises from 4.5\% at $\dmax = 0$ to 34.0\% at $\dmax = 0.08$ and 100\% at $\dmax = 0.20$, while the adjusted counterfactual remains at 4.5\% throughout. Panel~B (price-dividend ratio): distortion is more severe, reaching 64.0\% at $\dmax = 0.08$ and 100\% from $\dmax = 0.14$, because the price-dividend ratio is dominated by $\Tt$. Panel~C (varying $\rho$): rejection rates decline from 95.0\% at $\rho = 0.90$ to 42.0\% at $\rho = 0.99$, consistent with the comparative-static remark following Theorem~\ref{thm:size_distortion} (Comparative Statics in Discount Factor). Figure~\ref{fig:mc_all} presents these results graphically.



\paragraph{Calibration plausibility.}
Is $\dmax = 0.15$ empirically reasonable? During the late-1990s IT boom, excess real earnings growth was roughly 7\% annually ($\dmax \approx 0.07$). For AI-era Magnificent Seven stocks, earnings growth of 30--50\% in 2023--2024 implies $\dmax \in [0.10, 0.20]$. The 120-period adoption window matches observed GPT cycle durations. Our calibration is therefore conservative.

\subsection{Robustness to Alternative Technology Shock Shapes}
\label{subsec:robustness_shapes}

The baseline Monte Carlo experiments use the triangular specification of Remark~\ref{rmk:triangular}. A natural concern is whether the size distortion documented in Table~\ref{tab:mc_results} is an artifact of this functional form or, as our theory predicts, a general consequence of hump-shaped technology shocks. To address this, we repeat Experiment~(A) under four alternative specifications for $\delta_t$, each satisfying the general conditions of Assumption~\ref{ass:tech_shock} but differing in the sharpness of the peak, the symmetry of the adoption and maturation phases, and the smoothness of the profile.

The four specifications are:
\begin{enumerate}[label=(\roman*)]
    \item Triangular (baseline): The piecewise-linear specification of equation~\eqref{eq:delta_triangular}, with a sharp peak at $T_1 + \tau$.
    \item Gaussian: A smooth, symmetric bell-shaped profile $\delta_t = \dmax \cdot \exp\!\big(-\tfrac{1}{2}\big(\tfrac{t - (T_1+\tau)}{\sigma_g}\big)^2\big)$ for $t \in [T_1, T_2]$ and zero otherwise, with $\sigma_g = (T_2 - T_1)/4$ chosen so that the tails are negligible at the boundaries.
    \item Beta: A rescaled Beta density on $[T_1, T_2]$ with shape parameters $(2, 5)$, producing an asymmetric profile with a relatively fast rise and slower decay. After rescaling so that $\max_t \delta_t = \dmax$, this captures the empirically common pattern of rapid initial adoption followed by gradual maturation \citep{Jovanovic2005}.
    \item Gamma-like: An asymmetric profile with fast rise and slow decay, given by $\delta_t = \dmax \cdot s \cdot \exp(-\tfrac{1}{2}(s/\tau)^2) / \max(\cdot)$ where $s = t - T_1$. This specification concentrates the peak effect earlier in the adoption window and features a heavier right tail than the Beta specification.
\end{enumerate}

Figure~\ref{fig:delta_shapes} (Online Appendix) displays the four $\delta_t$ profiles under the baseline calibration ($\dmax = 0.15$, $T_1 = 80$, $T_2 = 200$, $\tau = 30$). The profiles differ substantially in their peak timing, sharpness, and tail behavior, providing a meaningful test of the robustness of our results. Table~\ref{tab:mc_shapes} (Online Appendix) reports the Monte Carlo rejection rates for each shape specification as a function of $\dmax$, using detrended log prices and the same simulation design as Experiment~(A) ($M = 200$ replications, $T = 300$, $\rho = 0.95$).

The results are consistent across specifications. All four specifications generate substantial and monotonically increasing size distortion, confirming the theoretical prediction of Theorem~\ref{thm:size_distortion}. At a moderate calibration of $\dmax = 0.08$, the rejection rates range from 34.0\% (Triangular) to 40.0\% (Gaussian), all far exceeding the 5\% nominal level. At $\dmax = 0.12$, rejection rates range from 74.5\% (Triangular and Beta) to 80.5\% (Gaussian). At the baseline $\dmax = 0.15$, all shapes produce rejection rates above 87\%, and by $\dmax = 0.20$ all four shapes reach or exceed 99.5\%.

The ordering of rejection rates across shapes is informative. The Gaussian specification, which features a smooth and concentrated peak, produces the highest rejection rate at nearly every $\dmax$ value. This is consistent with the mechanism identified in Theorem~\ref{thm:size_distortion}: a concentrated but smooth rise in $\delta_t$ generates a pronounced locally explosive phase in $\Tt$, which is most likely to trigger spurious detection. The Gamma-like specification, which concentrates the technology shock earlier in the adoption window, produces the second-highest rates. The Triangular and Beta(2,5) specifications, which distribute the technology effect more gradually, produce the lowest rejection rates. The differences across shapes are quantitatively modest, however, confirming that the size distortion phenomenon is robust to the functional form of the technology shock.

Figure~\ref{fig:mc_shapes} (Online Appendix) displays these results graphically, with the technology-adjusted counterfactual (flat at 4.5\%) shown for reference. Crucially, the technology-adjusted counterfactual rejection rate is exactly 4.5\% for all four specifications and all $\dmax$ values, confirming that the adjustment procedure of Section~\ref{sec:correction} is robust to the functional form of the technology shock. This is expected theoretically: the adjustment removes $\Tt$ regardless of the shape of $\delta_t$, and the residual is always a random walk with constant drift.

\subsection{Stochastic Technology Uncertainty and Size Amplification}
\label{subsec:mc_stochastic}

The baseline experiments above treat the technology shock as deterministic. We now assess the quantitative implications of the stochastic extension developed in Section~\ref{subsec:stochastic_tech}, where the cumulative technology impact $\bar{\delta}$ is uncertain ex ante and agents learn about it over time through noisy signals.

\subsubsection{Design}

We simulate the stochastic data-generating process of Assumption~\ref{ass:stochastic_tech}. For each replication $m = 1, \ldots, 500$ and each value of the coefficient of variation $\mathrm{CV} = \sigma_{\bar{\delta}} / \mu_{\bar{\delta}}$, we draw $\bar{\delta}^{(m)} = \mu_{\bar{\delta}} + z^{(m)} \cdot \sigma_{\bar{\delta}}$, where $z^{(m)} \sim N(0,1)$ is a common standard normal draw held fixed across the CV grid (common random numbers). Period-by-period signals $\delta_t^{\mathrm{obs}} = \bar{\delta} \cdot g(t) + \sigma_\xi \cdot \xi_t$ generate Bayesian updating of $\hat{\delta}_t$, and the present-value technology component $\hat{\Tt} = \hat{\delta}_t \cdot G_\rho(t)$ enters the price process. The CV grid ranges from 0 (deterministic baseline) to 1.00 (prior standard deviation equal to the mean), with $\sigma_\xi = 0.005$, matching a moderate news-noise environment. Common shocks (dividend innovations $\varepsilon_t$, technology news $\xi_t$, and price-dividend noise $u_t$) are pre-generated and held constant across the entire design, ensuring that differences in rejection rates across the CV grid are driven solely by prior uncertainty about $\bar{\delta}$.

We report results for two implementations:
\begin{itemize}
  \item Panel A: Log price-dividend ratio. The price-dividend ratio is $y_t = C + \hat{\Tt} + u_t$, where $u_t$ follows a first-order autoregressive process with persistence $\phi = 0.95$ and innovation standard deviation $\sigma_u = 0.15\sqrt{1 - \phi^2}$, and the technology present-value enters as a level effect ($\dmax = 0.04$, corresponding to a deterministic baseline rejection rate of approximately 10\%).
  \item Panel B: Detrended log price. The log price is $f_t = \sum_{s=1}^t (c + \delta_s^{\mathrm{obs}} + \varepsilon_s) + C + \hat{\Tt}$, detrended by OLS ($\dmax = 0.06$, baseline rejection approximately 23\%).
\end{itemize}

\subsubsection{Results}

Table~\ref{tab:mc_stochastic} reports the results. The two panels reveal a clear contrast that quantifies the distributional dependence on $\sigma_\xi$ established in Proposition~\ref{prop:stochastic_tech}(ii) and illuminates the empirical direction of the size change.


In Panel~A, size distortion increases monotonically with prior uncertainty. The deterministic baseline rejection rate is 10.4\%; at moderate uncertainty (CV $= 0.30$), the rate rises to 16.2\%, and at high uncertainty (CV $= 1.00$), it nearly triples to 29.0\%. The amplification is substantial: moving from a deterministic to a stochastic technology specification with CV $= 0.50$ more than doubles the size distortion from 10.4\% to 23.2\%. This monotonic pattern reflects two reinforcing channels. First, by Jensen's inequality, the convexity of the rejection function $\psi(\bar{\delta})$ over the relevant range ensures that $\E[\psi(\bar{\delta})] > \psi(\E[\bar{\delta}])$: averaging over draws of $\bar{\delta}$ that produce both stronger and weaker technology shocks yields a higher average rejection rate than the deterministic case where $\bar{\delta} = \mu_{\bar{\delta}}$ with certainty. Second, the Bayesian learning channel generates martingale dynamics in $\hat{\delta}_t$ that propagate directly into the price-dividend ratio as level shifts, compounding the locally explosive pattern detected by PSY.

Panel~B shows a qualitatively different pattern. Under the detrended log price implementation, rejection rates are essentially flat across the CV grid, with a modest amplification emerging only at very high uncertainty (CV $\geq 0.75$). The contrast with Panel~A is explained by the cumulation mechanism: in the detrended log price, the technology shock enters through $\sum_{s=1}^t \delta_s^{\mathrm{obs}}$, and the cumulation smooths out the martingale fluctuations in $\hat{\delta}_t$ that drive the price-dividend ratio result. The linear detrending further attenuates any residual level shifts. The technology-adjusted rejection rates (0.8\% for the price-dividend ratio and 6.8\% for the detrended log price) confirm that the adjustment procedure eliminates size distortion under both the deterministic and stochastic specifications.

Figure~\ref{fig:mc_stochastic} displays the rejection rate as a function of CV for both implementations, with the deterministic baseline and technology-adjusted rates shown for reference. The monotonic increase in Panel~A and the flat profile in Panel~B together provide a quantitative, implementation-dependent refinement of Proposition~\ref{prop:stochastic_tech}(ii): the theorem guarantees only that the limit law depends nontrivially on $\sigma_\xi$, and the simulations show that this dependence translates into substantial size amplification through the price-dividend ratio channel---where the learning-driven martingale component enters as a direct level effect---while being largely absorbed by cumulation and linear detrending in the detrended log-price channel. The deterministic results of Section~\ref{subsec:mc_results} therefore represent a conservative lower bound on the size distortion for the price-dividend implementation when agents face genuine uncertainty about the technology's ultimate impact, precisely the empirically relevant case during general-purpose technology adoption episodes.


\subsection{Distinguishing Technology from Speculation: An Overlap Experiment}
\label{subsec:overlap_experiment}

The preceding experiments establish that PSY produces spurious rejections when fundamentals incorporate technology-driven growth, even in the complete absence of a speculative bubble. A natural follow-up question is whether the technology-adjusted PSY test identifies a genuine speculative bubble while avoiding the spurious detection generated by the technology component when both are present. We address this question through a controlled overlap experiment.

\subsubsection{Design}

We construct a data-generating process in which a technology shock and a genuine speculative bubble overlap in time:
\begin{itemize}
    \item Technology window: $[T_1, T_2] = [50, 150]$ with a Gamma-like $\delta_t$ profile, $\dmax = 0.25$ (a strong technology shock chosen to ensure robust spurious detection by the unadjusted PSY test).
    \item Bubble window: $[t_2, t_3] = [100, 200]$. The speculative bubble component $b_t$ is initialized at $b_{t_2} = 0.3$ and follows an explosive AR(1) process $b_t = \rho_{\mathrm{bub}}\, b_{t-1} + \eta_t$ with $\rho_{\mathrm{bub}} = 1.035$ and $\eta_t \sim \mathcal{N}(0, 0.10^2)$ during $[t_2, t_3]$. After $t_3$, the bubble collapses via $b_t = 0.5\, b_{t-1} + \eta_t$.
\end{itemize}

The observed price is $p_t = f_t + b_t + e_t$, where $f_t$ is the technology-augmented fundamental and $e_t \sim \mathcal{N}(0, 0.30^2)$ is observation noise. The technology-adjusted price removes both the present-value term $\Tt$ and the cumulative technology contribution to dividends: $p_t^{\mathrm{adj}} = p_t - \Tt - \sum_{s=1}^t \delta_s$. We apply PSY to the detrended log price, which is the transformation used in our empirical analysis.

The key prediction is that the unadjusted PSY test should first trigger during the technology-only phase (between $t = 50$, the start of technology window, and $t = 100$, the start of bubble window), constituting a spurious early detection, while the technology-adjusted PSY test should first trigger during or after the bubble window (a genuine detection).

\subsubsection{Results}

Figure~\ref{fig:overlap} presents the results for a representative simulation (seed = 10000, selected to illustrate the mechanism clearly). 

Panel~(A) displays the raw and technology-adjusted log prices, with the technology window (blue shading, $[50, 150]$) and bubble window (red shading, $[100, 200]$) indicated. The overlap region $[100, 150]$ is where both technology and speculative components are simultaneously active. The raw log price shows a pronounced rise during both windows, while the adjusted price is flat during the technology-only phase but rises sharply once the genuine bubble emerges.

Panel~(B) shows the BSADF statistic from the unadjusted PSY test. The statistic exceeds the 5\% critical value sequence beginning at $t = 59$, well within the technology-only phase $[50, 100]$, producing a spurious early detection during a period with no speculative activity. 

Panel~(C) shows the BSADF statistic from the technology-adjusted PSY test. The statistic remains below the critical value throughout the technology-only phase and first exceeds it at $t = 184$, well within the genuine bubble window ($[100, 200]$).

Panel~(D) provides a timeline comparison. The unadjusted test identifies ``bubble'' episodes beginning 41 periods before any speculative activity starts (at $t = 59$ versus the bubble onset at $t = 100$), conflating fundamental and speculative dynamics. The technology-adjusted test correctly restricts its first detection to the bubble window, with a first detection date 125 periods later than the unadjusted test.


This experiment illustrates the practical consequences of ignoring technology-driven fundamentals in bubble detection. The unadjusted PSY test produces a spurious early warning 41 periods before any speculative activity begins, with the first detection occurring during the technology-only phase ($t = 59 \in [50, 100]$). In contrast, the technology-adjusted test avoids the false alarm and identifies the genuine bubble with a detection lag of approximately 84 periods from its onset, within the range of detection delays typically observed in PSY applications to genuine explosive processes. The overlap design is particularly informative: because the technology and bubble windows share a 50-period overlap ($[100, 150]$), the unadjusted test cannot distinguish the source of the explosive dynamics, while the adjusted test identifies only the speculative component.

\clearpage
\clearpage
\begin{figure}
\centerline{\includegraphics[width=0.95\textwidth]{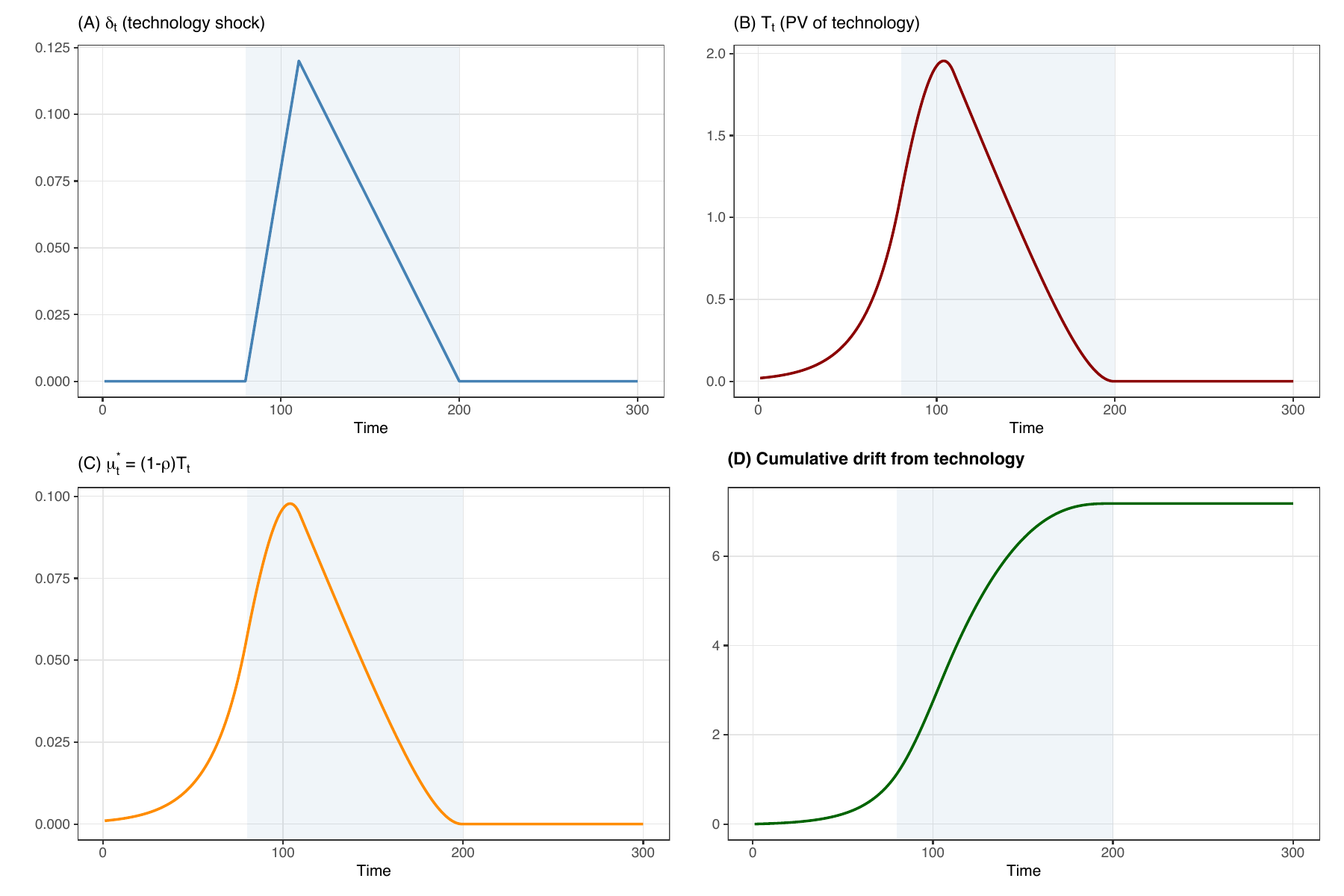}}
\noindent\caption{{\bf A technology shock with peak $\dmax = 0.12$ raises the fundamental price through the adoption window.} The top panel plots the technology shock $\delta_t$ against time, the middle panel plots the present-value technology component $\mathcal{T}_t$ against time, and the bottom panel plots the implied drift $\mu_t^*$ against time. The takeaway is that a temporary hump in technology growth produces a sustained increase in the level of the fundamental price.\label{fig:components}}
\end{figure}

\clearpage
\begin{figure}
\centerline{\includegraphics[width=0.95\textwidth]{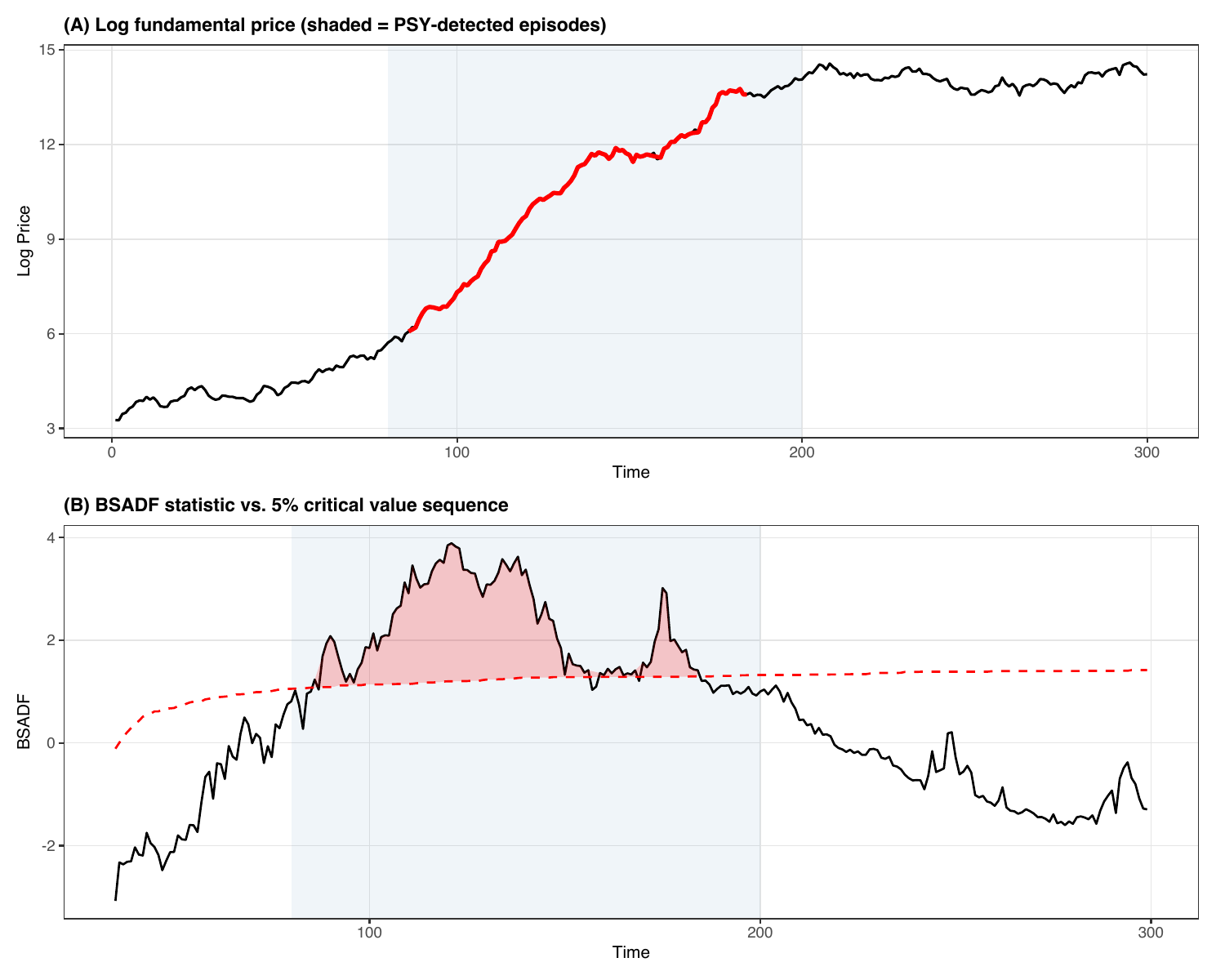}}
\noindent\caption{{\bf The standard Phillips--Shi--Yu (PSY) test falsely classifies the technology-driven simulation as explosive.} The horizontal axis is time, the solid line is the BSADF date-stamping statistic, and the dashed line is the 5 percent critical value sequence. The statistic crosses the critical value during the adoption phase and produces a GSADF statistic of 3.890, so the unadjusted procedure labels most of the technology window as a bubble.\label{fig:psy_with_tech}}
\end{figure}

\clearpage
\begin{figure}
\centerline{\includegraphics[width=0.95\textwidth]{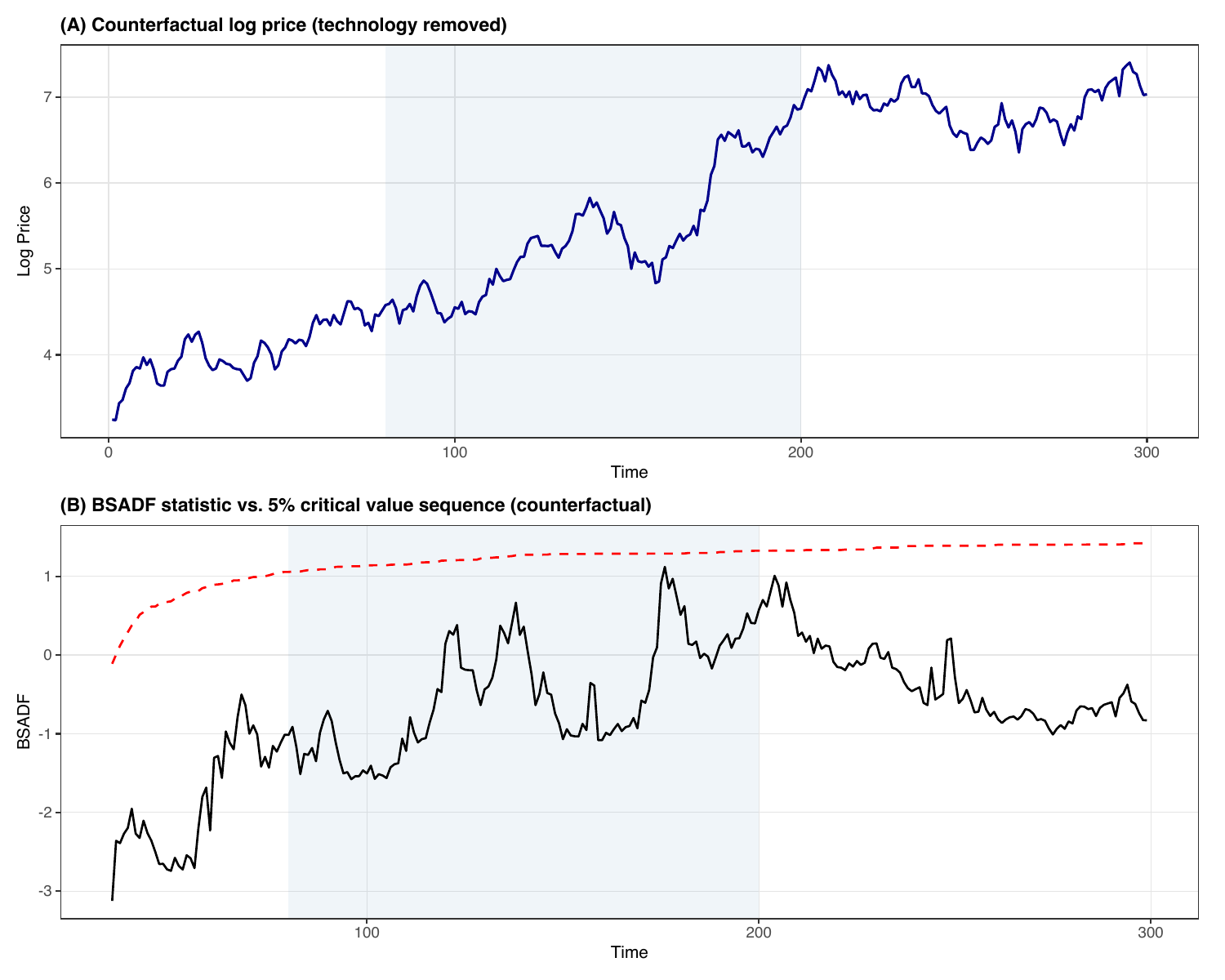}}
\noindent\caption{{\bf Technology adjustment removes the false bubble signal in the simulated path.} The horizontal axis is time, the solid line is the BSADF date-stamping statistic after adjustment, and the dashed line is the 5 percent critical value sequence. The adjusted statistic remains below the threshold throughout, with a GSADF statistic of 1.120.\label{fig:psy_no_tech}}
\end{figure}

\clearpage
\begin{figure}
\centerline{\includegraphics[width=0.95\textwidth]{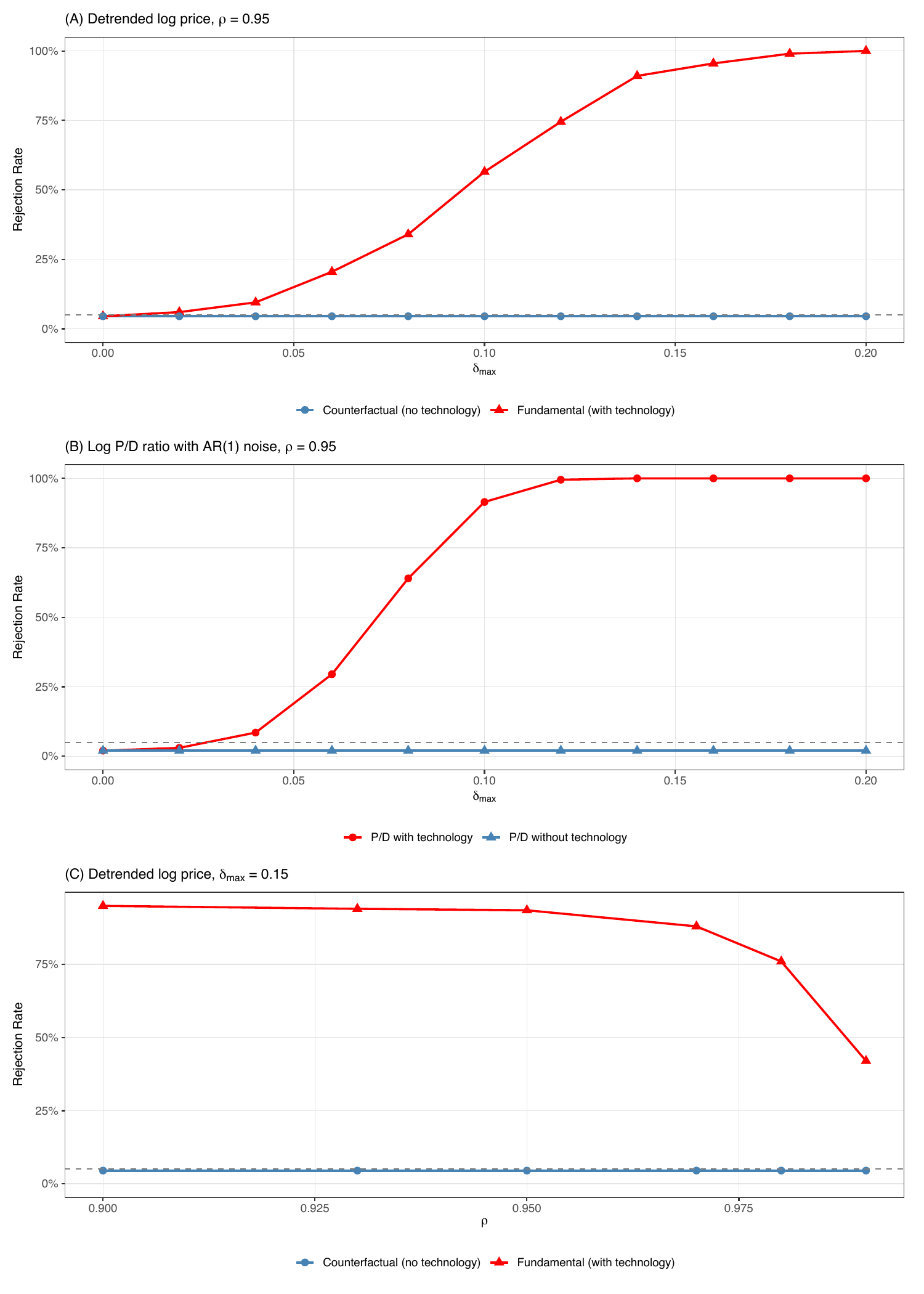}}
\noindent\caption{{\bf False rejection rates rise monotonically with the technology shock peak $\dmax$ in the unadjusted test.} The horizontal axis is the technology shock peak $\dmax$, and the vertical axis is the rejection rate at the 5 percent nominal level. Solid lines plot the unadjusted PSY rejection rates, while the dashed line plots the technology-adjusted rejection rate, which stays near 4.5 percent.\label{fig:mc_all}}
\end{figure}

\clearpage
\begin{figure}
\centerline{\includegraphics[width=0.95\textwidth]{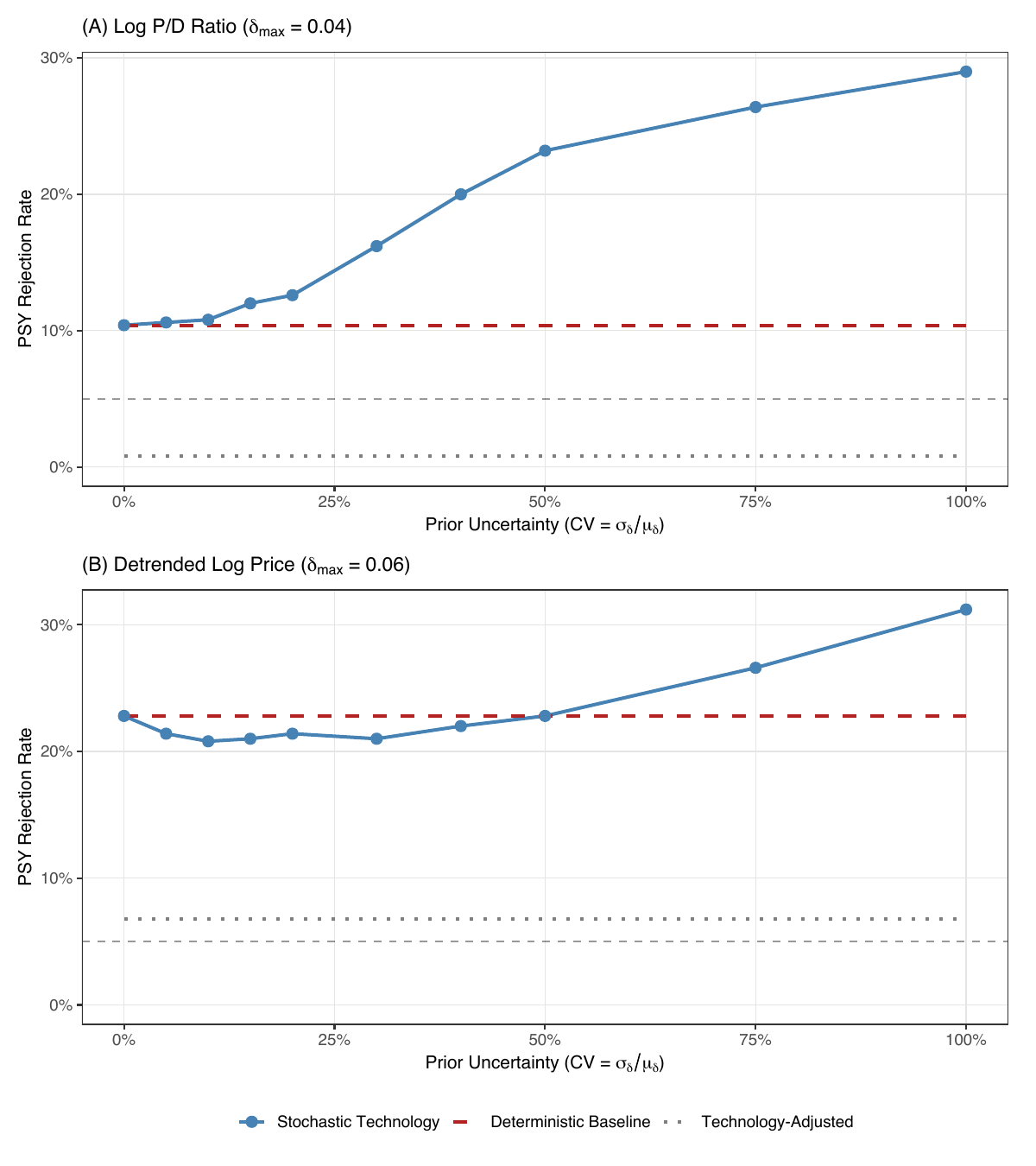}}
\noindent\caption{{\bf Technology uncertainty raises rejection rates mainly through the price-dividend ratio channel.} The horizontal axis is the coefficient of variation of prior uncertainty about cumulative technology impact, and the vertical axis is the rejection rate at the 5 percent nominal level. Panel A plots the log price-dividend ratio with technology shock peak $\dmax = 0.04$, Panel B plots detrended log prices with technology shock peak $\dmax = 0.06$, and the circle and triangle markers identify the deterministic baseline and the adjusted rejection rate.\label{fig:mc_stochastic}}
\end{figure}

\clearpage
\begin{figure}
\centerline{\includegraphics[width=0.95\textwidth]{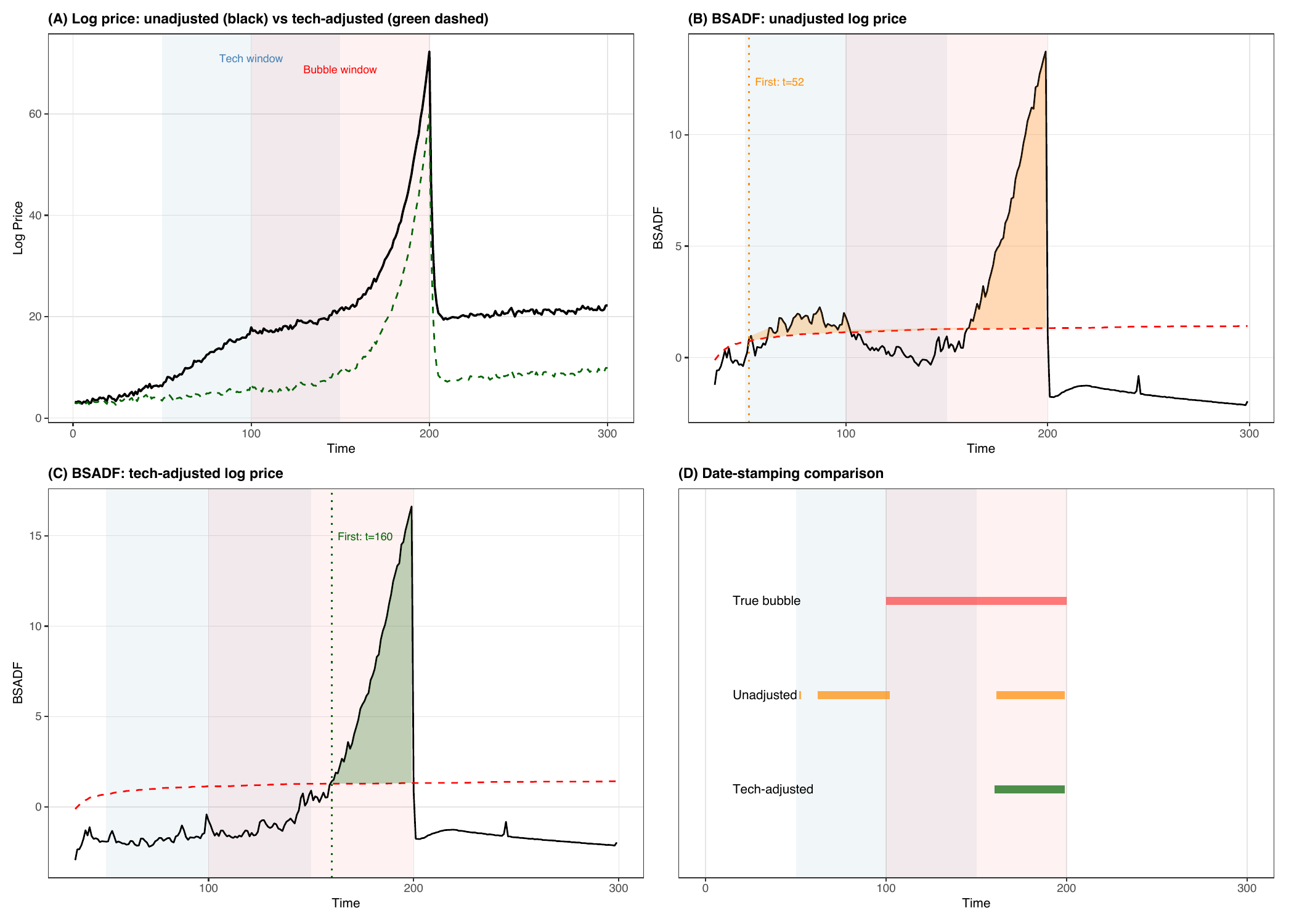}}
\noindent\caption{{\bf The adjusted procedure delays detection until the genuine bubble, while the unadjusted procedure fires during the technology-only phase.} Panel A plots raw and technology-adjusted log prices against time, with the technology window shaded in blue and the bubble window shaded in red. Panel B plots the unadjusted BSADF date-stamping statistic and its 5 percent critical value sequence, Panel C plots the adjusted date-stamping statistic and its 5 percent critical value sequence, and Panel D summarizes the detected episodes on a timeline.\label{fig:overlap}}
\end{figure}

\clearpage
\begin{table}
\noindent\caption{{\bf Baseline calibration for Monte Carlo experiments.} Parameter values used in the simulation design. The discount factor $\rho$, dividend growth rate $c$, and expected return $\rbar$ are calibrated to postwar U.S.\ equity market data. The technology adoption window $[T_1, T_2]$ and peak lag $\tau$ are consistent with observed GPT cycle durations. The triangular technology shock is defined over the adoption window and reaches its peak within that window. The reported critical-value replications are used to construct the 5 percent Monte Carlo critical values for PSY.\label{tab:calibration}}
\centering
\begin{tabular}{llll}
\toprule
\textbf{Parameter} & \textbf{Symbol} & \textbf{Baseline Value} & \textbf{Description} \\
\midrule
Discount factor & $\rho$ & 0.95 & Campbell--Shiller log-linearization \\
Dividend growth & $c$ & 0.02 & Steady-state log dividend growth rate \\
Expected return & $\rbar$ & 0.06 & Mean log return \\
Innovation s.d. & $\sigma_\varepsilon$ & 0.10 & Std.\ dev.\ of dividend growth shocks \\
Sample length & $T$ & 300 & Total number of periods \\
Technology shock shape & --- & Triangular & See Remark~\ref{rmk:triangular} \\
Adoption window & $[T_1, T_2]$ & $[80, 200]$ & 120 periods \\
Minimum window fraction & $r_0$ & $0.01 + 1.8/\sqrt{T}$ & PSY recursion \\
Number of lags & $K$ & 1 & Augmented Dickey--Fuller lag order \\
Critical-value replications & --- & 2{,}000 & Monte Carlo critical values \\
\bottomrule
\end{tabular}
\end{table}

\clearpage
\begin{table}
\noindent\caption{{\bf Technology shocks sharply raise false rejection rates unless the adjustment is applied.} Panel A reports rejection rates for detrended log prices as the technology shock peak $\dmax$ varies with discount factor $\rho = 0.95$, Panel B reports rejection rates for the log price-dividend ratio as the technology shock peak $\dmax$ varies, and Panel C reports rejection rates for detrended log prices as the discount factor $\rho$ varies with technology shock peak $\dmax = 0.15$. Across all panels, the adjusted procedure keeps rejection rates near the nominal level while the unadjusted procedure over-rejects. The unadjusted column reports rejection rates from standard PSY applied to the raw simulated series, the technology-adjusted column reports rejection rates after removing the present-value technology component $\mathcal{T}_t$, and critical values are computed from 2,000 Monte Carlo replications.\label{tab:mc_results}}
\centering
\begin{tabular}{lcc}
\toprule
\textbf{Parameter Value} & \textbf{Unadjusted (\%)} & \textbf{Technology-Adjusted (\%)} \\
\midrule
\multicolumn{3}{l}{\textit{Panel A: Detrended Log Prices, $\rho = 0.95$}} \\
\midrule
$\dmax = 0.00$ & 4.5 & 4.5 \\
$\dmax = 0.02$ & 6.0 & 4.5 \\
$\dmax = 0.04$ & 9.5 & 4.5 \\
$\dmax = 0.06$ & 20.5 & 4.5 \\
$\dmax = 0.08$ & 34.0 & 4.5 \\
$\dmax = 0.10$ & 56.5 & 4.5 \\
$\dmax = 0.12$ & 74.5 & 4.5 \\
$\dmax = 0.14$ & 91.0 & 4.5 \\
$\dmax = 0.16$ & 95.5 & 4.5 \\
$\dmax = 0.18$ & 99.0 & 4.5 \\
$\dmax = 0.20$ & 100.0 & 4.5 \\
\midrule
\multicolumn{3}{l}{\textit{Panel B: Log Price-Dividend Ratio}} \\
\midrule
$\dmax = 0.00$ & 2.0 & 2.0 \\
$\dmax = 0.02$ & 3.0 & 2.0 \\
$\dmax = 0.04$ & 8.5 & 2.0 \\
$\dmax = 0.06$ & 29.5 & 2.0 \\
$\dmax = 0.08$ & 64.0 & 2.0 \\
$\dmax = 0.10$ & 91.5 & 2.0 \\
$\dmax = 0.12$ & 99.5 & 2.0 \\
$\dmax = 0.14$ & 100.0 & 2.0 \\
$\dmax = 0.16$ & 100.0 & 2.0 \\
$\dmax = 0.18$ & 100.0 & 2.0 \\
$\dmax = 0.20$ & 100.0 & 2.0 \\
\midrule
\multicolumn{3}{l}{\textit{Panel C: Varying $\rho$, $\dmax = 0.15$}} \\
\midrule
$\rho = 0.90$ & 95.0 & 4.5 \\
$\rho = 0.93$ & 94.0 & 4.5 \\
$\rho = 0.95$ & 93.5 & 4.5 \\
$\rho = 0.97$ & 88.0 & 4.5 \\
$\rho = 0.98$ & 76.0 & 4.5 \\
$\rho = 0.99$ & 42.0 & 4.5 \\
\bottomrule
\end{tabular}
\end{table}

\clearpage
\begin{table}
\noindent\caption{{\bf Prior uncertainty amplifies false rejections most strongly in the price-dividend ratio specification.} Panel A reports rejection rates and percentage-point changes for the log price-dividend ratio as the coefficient of variation of cumulative technology impact, $\sigma_{\bar{\delta}}/\mu_{\bar{\delta}}$, increases with technology shock peak $\dmax = 0.04$; Panel B reports the same quantities for detrended log prices with technology shock peak $\dmax = 0.06$. The table shows that uncertainty raises distortion markedly in Panel A, while the adjusted procedure keeps rejection rates low in both panels. The final row reports the technology-adjusted rejection rates for the two specifications.\label{tab:mc_stochastic}}
\centering
\begin{tabular}{lcccc}
\toprule
\textbf{CV} & \textbf{A: Rej.\ (\%)} & \textbf{A: $\Delta$ (pp)} & \textbf{B: Rej.\ (\%)} & \textbf{B: $\Delta$ (pp)} \\
\midrule
0.00 (deterministic) & 10.4 & --- & 22.8 & --- \\
0.05 & 10.6 & +0.2 & 21.4 & $-$1.4 \\
0.10 & 10.8 & +0.4 & 20.8 & $-$2.0 \\
0.15 & 12.0 & +1.6 & 21.0 & $-$1.8 \\
0.20 & 12.6 & +2.2 & 21.4 & $-$1.4 \\
0.30 & 16.2 & +5.8 & 21.0 & $-$1.8 \\
0.40 & 20.0 & +9.6 & 22.0 & $-$0.8 \\
0.50 & 23.2 & +12.8 & 22.8 & +0.0 \\
0.75 & 26.4 & +16.0 & 26.6 & +3.8 \\
1.00 & 29.0 & +18.6 & 31.2 & +8.4 \\
\midrule
Technology-adjusted & 0.8 &  & 6.8 &  \\
\bottomrule
\end{tabular}
\end{table}

\clearpage
\section{Robustness Tests: Identification and Specification}
\label{app:robustness}

This appendix reports five complementary diagnostics that assess the separability of technology proxies from the speculative bubble component, which is the paper's main identification assumption. Proposition~\ref{prop:microfoundation} provides a structural benchmark for separability, while Section~\ref{subsec:feedback_bounds} shows how the residual changes when separability fails. The tests below evaluate whether the empirical conclusions are fragile to observable signs of feedback and specification choice.

\subsection{Granger Causality Tests}
\label{subsec:granger}

Under the separability assumption, stock returns should not Granger-cause the technology proxies during speculative periods: if the bubble component $b_t$ is orthogonal to fundamentals, explosive price dynamics should not feed back into real technology variables. During non-speculative periods, however, bidirectional causality between technology and returns is expected, since technology fundamentals drive prices and stock market valuations affect corporate investment. These tests are diagnostic rather than decisive; failure to reject return-to-technology predictability supports separability but cannot rule out all slower feedback channels.

We test this prediction by estimating bivariate Granger causality regressions between NASDAQ returns and each of the three technology proxies (IT-investment growth, TFP growth, and patent growth), separately for speculative and non-speculative regimes. The PSY-detected bubble periods are too narrow (5--6~months for the dot-com episode) to support meaningful inference, so we define broader speculative periods based on the run-up phase: January~1997 to March~2000 for the dot-com era ($n = 39$) and January~2023 to December~2025 for the AI era ($n = 33$--$35$). Lag lengths are selected by the Bayesian Information Criterion (BIC). HC3 heteroskedasticity-robust $p$-values could not be computed because of near-singular covariance matrices in the small speculative-period samples; we therefore report both standard $F$-test and wild bootstrap $p$-values, which yield consistent conclusions throughout.

Table~\ref{tab:granger_causality} reports the results. During both candidate speculative periods, no technology proxy is predicted by lagged returns in the Granger sense: all $p$-values exceed~0.70 for the dot-com episode and~0.54 for the AI era. The reverse direction is also insignificant: technology variables do not predict returns in the Granger sense during these periods. This pattern holds uniformly across all three technology proxies and both eras, providing evidence consistent with the separability assumption.

The non-speculative periods tell a complementary story. In the dot-com non-speculative sample ($n = 136$--$141$), only one marginally significant relationship emerges: returns Granger-cause IT investment at the~10\% level ($F = 2.476$, $p = 0.062$), consistent with an accelerator channel whereby equity market valuations affect IT spending with a lag. No other direction is significant. In the AI non-speculative sample ($n = 36$), the results are richer: IT investment, TFP, and patent grants all Granger-cause returns at the~5\% level or better ($p = 0.019$, $0.015$, and $0.005$, respectively), while returns also Granger-cause TFP ($p = 0.016$). This bidirectional causality during normal periods is expected and economically sensible: technology fundamentals drive asset prices, while stock market wealth effects and Tobin's~$q$ channels feed back into productivity. The absence of this bidirectional link during candidate speculative regimes is consistent with the separability benchmark, but the small samples mean the result should be read as supportive rather than conclusive.

\subsection{Placebo Tests: Non-Technology Covariates}
\label{subsec:placebo}

If the technology-adjusted test eliminates explosive signals specifically because the technology proxies capture fundamental content, then replacing these proxies with less relevant covariates should not systematically reproduce the same pattern. We construct ``placebo'' counterfactuals by regressing log NASDAQ prices on three non-technology variables (consumer sentiment (University of Michigan Index), the VIX implied volatility index, and industrial production (INDPRO)) via DOLS over the same training windows, then apply PSY to the resulting price gaps. Table~\ref{tab:placebo_tests} reports the results alongside the technology baseline.

For the AI era, all three placebos yield GSADF statistics well below their respective 10 percent critical values ($-0.223$, $0.712$, and $0.869$ against $\text{CV}_{0.10} = 1.334$), consistent with the baseline technology adjustment (test statistic $= -0.275$). This is expected: when no explosive dynamic exists in the data, neither the correct nor incorrect covariates will produce a rejection.

The dot-com results are more nuanced. Industrial production does not reject ($\text{test statistic} = 1.460 < \text{CV}_{0.10} = 1.768$; $R^2 = 0.743$), showing that a generic macroeconomic variable does not replicate the technology adjustment's residual pattern. The VIX could not be tested for the dot-com era because the CBOE VIX series begins in January~1990, providing only 12 training observations, which is insufficient for reliable DOLS estimation. However, consumer sentiment rejects at both the 10 and 5 percent levels ($\text{test statistic} = 2.283$, $\text{CV}_{0.10} = 1.768$, $\text{CV}_{0.05} = 2.025$). The result requires careful interpretation. The sentiment-based DOLS specification achieves a low $R^2$ of only~0.308, indicating that sentiment is a poor predictor of log price levels; the resulting price gap retains most of the original price variation, including both the fundamental and speculative components. Moreover, during the dot-com boom, consumer sentiment was itself heavily influenced by the technology-driven equity rally through wealth effects \citep{Shiller2015}, making it endogenous to the very dynamics it is intended to ``placebo'' against. The sentiment rejection therefore does not undermine the identification strategy, but it reinforces the broader point that covariate choice and exogeneity matter.

\subsection{Cross-Sectional Identification: PCA of Magnificent Seven Price Gaps}
\label{subsec:pca}

The Magnificent Seven panel provides a cross-sectional test of the technology-adjustment methodology. We apply PCA to the technology-adjusted price gaps $\hat{g}_{i,t}$ of the seven firms, computed from OLS regressions of each firm's log price on firm-specific Compustat fundamentals (real revenue per share and real R\&D expenditure per share) using the 2006--2019 training window (excluding each firm's own bubble months). The analysis uses 72~months of common post-training observations.

Table~\ref{tab:pca_mag7} reports the variance decomposition and PSY results for each principal component. The first principal component (PC1) explains 60.0\% of the total variation, revealing a dominant common factor in the adjusted price gaps. The loadings on PC1 are broadly similar across six of the seven firms (GOOGL~($-0.472$), TSLA~($-0.437$), NVDA~($-0.423$), AMZN~($-0.374$), META~($-0.371$), and AAPL~($-0.358$)), but conspicuously low for MSFT~($-0.054$). Microsoft's near-zero loading indicates that its price-gap dynamics are driven by idiosyncratic factors rather than the common AI-driven technology repricing that affects the other six firms. PC2 accounts for 20.6\% of variance and is loaded most heavily on MSFT~(0.754) and AMZN~($-0.435$), capturing the idiosyncratic divergence between these two firms. The first three components jointly explain 92.3\% of variation.

The key result is that none of the seven principal components exhibits explosive behavior. All GSADF statistics fall below the 10 percent critical value of~1.334: PC1 yields $\text{GSADF} = 0.844$, the closest to rejection is PC3 at $\text{GSADF} = 0.683$, and the remaining components range from $-1.218$ to $0.043$. The absence of explosiveness in any linear combination of the adjusted price gaps provides cross-sectional evidence against a common residual AI-era bubble in this sample. Even the dominant common factor, which loads on all firms except Microsoft, is stationary, consistent with the interpretation that the firm-level technology adjustment removes the common fundamental repricing component, leaving no common residual explosive factor. Figure~\ref{fig:pca_mag7} displays the scree plot and the time series of the first principal component.

\subsection{Training-Window Stability Analysis}
\label{subsec:training_stability}

A potential concern with the technology-adjusted test is that results depend on the training period used to estimate the DOLS coefficients. We systematically vary the training window end-date: from January~1985 to December~1994 for the dot-com application (with a fixed start of January~1975) and from January~2010 to December~2019 for the AI application (with a fixed start of January~2006). We reestimate the DOLS coefficients on each subsample and apply PSY to the resulting price gaps. Table~\ref{tab:training_stability} reports representative windows; Figure~\ref{fig:training_stability} displays the complete trajectories.

The two eras present a clear contrast. For the AI era, the GSADF statistic ranges from $-0.30$ to $0.84$ across 121~training windows, never approaching the 10 percent critical value of~1.176. The ``no AI bubble'' conclusion is entirely insensitive to training window choice. This robustness holds despite substantial variation in the DOLS coefficients: the IT-investment coefficient $\hat{\beta}_{\text{IT}}$ ranges from~0.35 to~2.65, and the patent coefficient $\hat{\beta}_{\text{Pat}}$ spans $-0.16$ to~1.19, yet the adjusted price gap never produces a rejection. The $R^2$ remains high throughout (0.82--0.91), indicating that the technology proxies explain a large and stable share of price variation regardless of the specific training window.

For the dot-com era, the picture is more complex. The test statistic varies from approximately~1.30 (training ending September~1987) to~2.76 (training ending June~1985), a range of~1.46. With the 10 percent critical value at~1.745, early training windows (ending before mid-1987) yield rejections, a gap from approximately mid-1987 through late~1988 yields non-rejections, then windows ending from late~1988 through early~1992 return to rejection, before settling below the critical value for windows ending after mid-1992. The source of this instability is primarily the patent coefficient, which varies from $-0.39$ (January~1985) to~$+1.04$ (December~1994), while the IT coefficient is comparatively stable at~$0.42$--$0.49$. This sensitivity reflects the changing role of patenting in the technology-growth relationship: in the early~1980s, patent grants were a poor proxy for innovation \citep{Kogan2017}, while their informational content improved substantially through the 1990s.

\subsection{Hansen (1992) Parameter Instability Test}
\label{subsec:hansen}

The training-window analysis reveals coefficient variation, but does not formally test whether the cointegrating relationship is stable within each training period. We apply the \citet{Hansen1992} $L_c$ test for parameter instability in cointegrating regressions, which tests the null hypothesis that the DOLS coefficients are constant against the alternative of a martingale parameter process. Table~\ref{tab:hansen_test} reports the results.

The $L_c$ test rejects parameter stability in both eras and both specifications. For the dot-com baseline (technology-only), $L_c = 1.845$ ($p = 0.009$); for the AI-era baseline, $L_c = 3.994$ ($p < 0.001$). The extended specifications including the term spread yield similar conclusions ($L_c = 1.780$, $p = 0.018$ for dot-com; $L_c = 3.941$, $p < 0.001$ for the AI era). All statistics far exceed the 1\% critical values (1.160 for baseline, 1.278 for extended), indicating strong evidence of parameter instability.

This finding has important implications. The cointegrating relationship between log prices and technology proxies is not stable over the training periods. For the dot-com era (1975--1990), this is unsurprising: the technology-growth nexus evolved dramatically as IT transitioned from mainframes to personal computers to networked systems. For the AI era (2006--2019), instability reflects the transition from mobile and cloud computing to the nascent AI paradigm, with the patent coefficient (capturing innovation composition) absorbing most of the structural change, as the training-window analysis of Section \ref{subsec:training_stability} confirms.

The Hansen rejection does not invalidate the technology-adjusted test, but it qualifies the interpretation. The DOLS estimates represent an average technology-price relationship over the training period, and the adjustment is valid insofar as this average captures the typical equilibrium. The training-window analysis provides reassurance for the AI-era application: the non-rejection finding holds across all 121~training windows, even as individual coefficients vary substantially. The dot-com results are more fragile, with the GSADF statistic crossing the critical value as the training window shifts, though residual explosiveness around the 1999--2000 peak appears across many specifications. Future work could address parameter instability more formally through time-varying coefficient approaches, such as the \citet{BierensMartins2010} time-varying cointegration framework or state-space models with Kalman-filtered parameters.

\clearpage
\begin{figure}
\centerline{\includegraphics[width=0.95\textwidth]{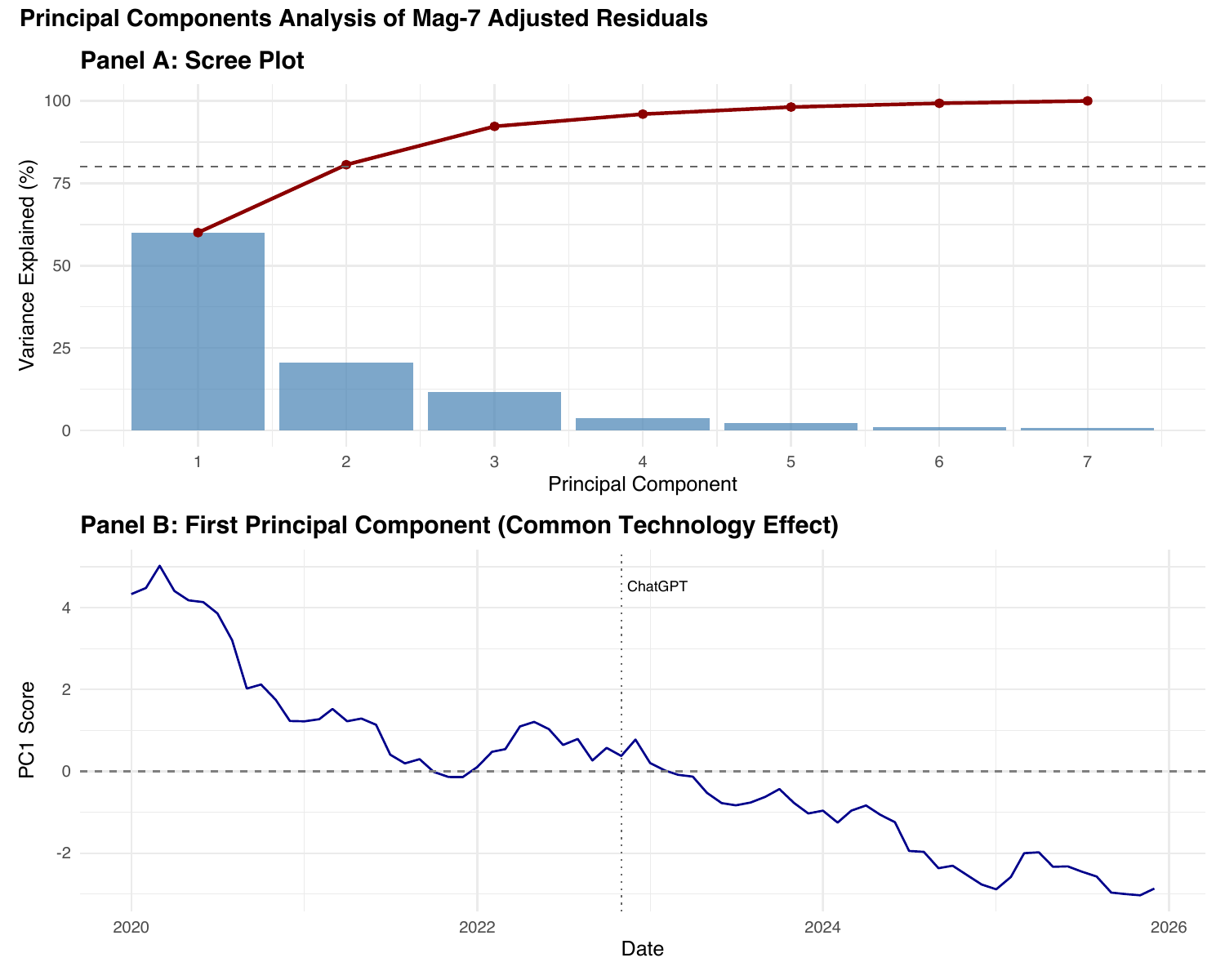}}
\noindent\caption{{\bf The dominant common factor in adjusted Magnificent Seven price gaps is large but not explosive.} The left panel plots the scree diagram with var.\ expl.\ by each principal component, and the right panel plots the time series of the first principal-component scores against calendar time. The first component captures most common variation, but it does not generate an explosive signal.\label{fig:pca_mag7}}
\end{figure}

\clearpage
\begin{figure}
\centerline{\includegraphics[width=0.95\textwidth]{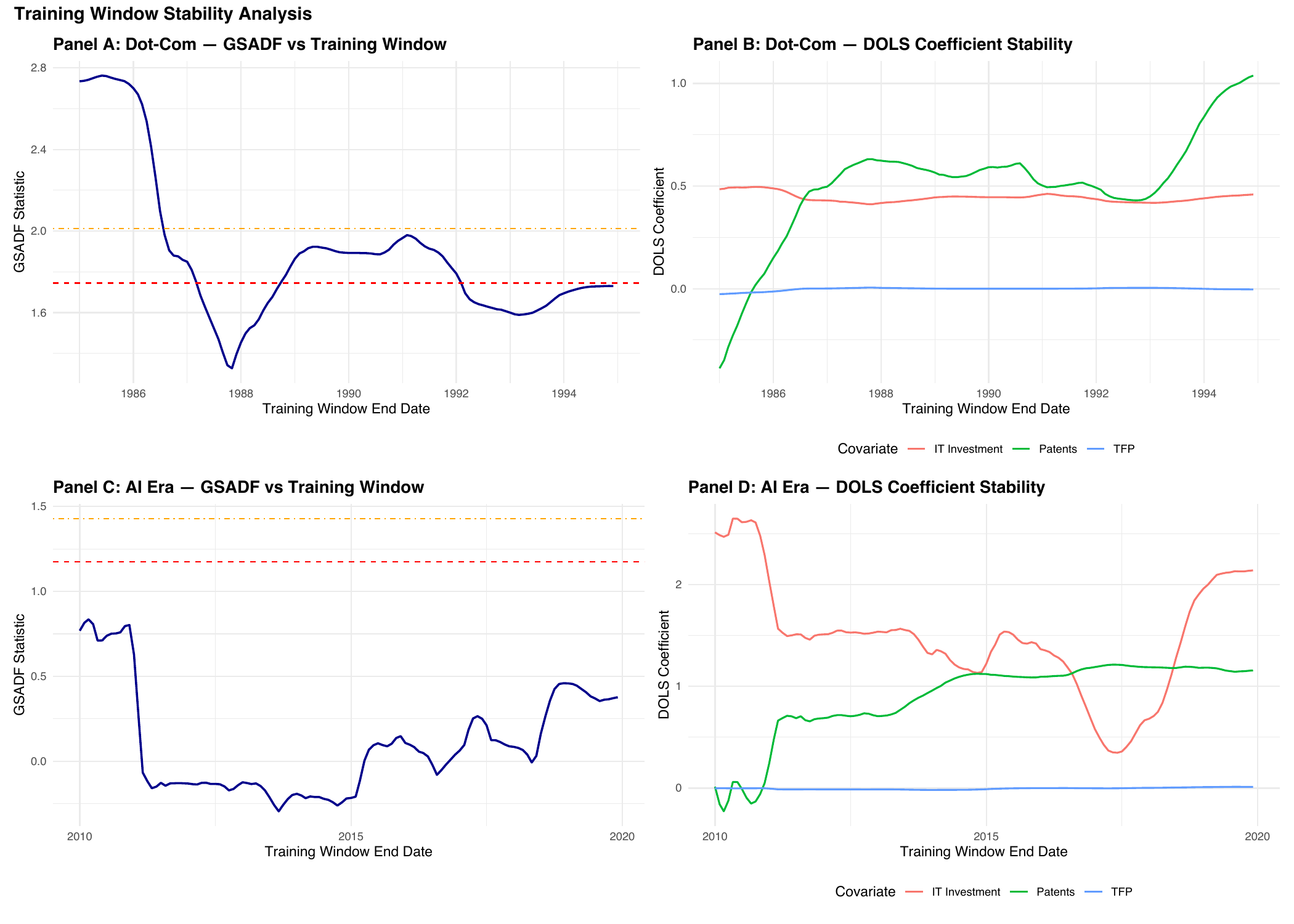}}
\noindent\caption{{\bf Training-window variation leaves the AI-era result intact but moves the dot-com result across the threshold.} The horizontal axis is the training-window end date, and the vertical axes plot GSADF statistics in the top panels and DOLS coefficients in the bottom panels. The left column corresponds to the dot-com era, the right column corresponds to the AI era, and the horizontal lines mark the 10 and 5 percent Monte Carlo critical values.\label{fig:training_stability}}
\end{figure}

\clearpage
\begin{figure}
\centerline{\includegraphics[width=0.95\textwidth]{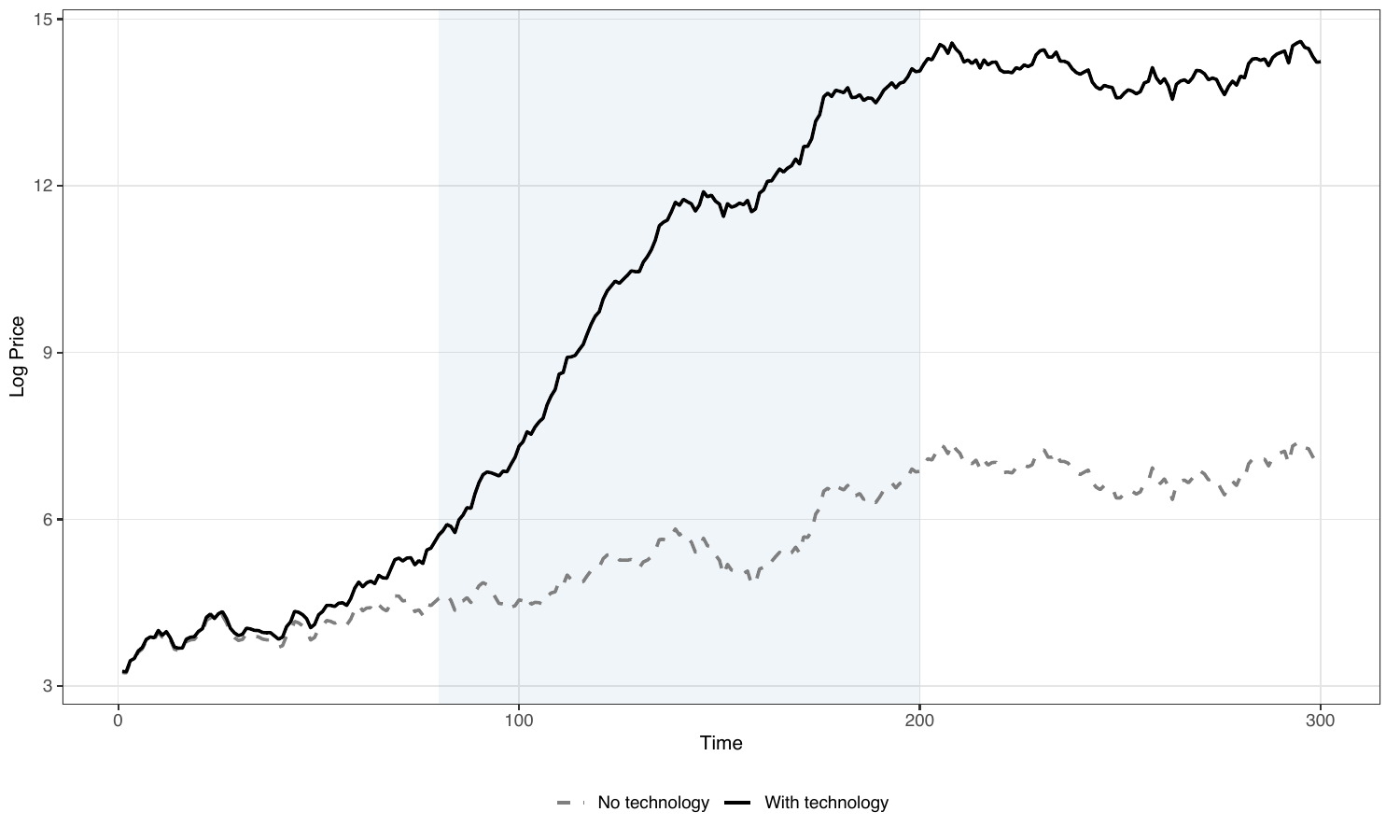}}
\noindent\caption{{\bf A temporary technology shock leaves the simulated fundamental price on a permanently higher path.} The horizontal axis is time, the solid line is the simulated fundamental price with the technology component, and the dashed line is the counterfactual price without that component. Under the illustrative calibration, the technology-augmented path accelerates during adoption and then settles at a higher level.\label{fig:price_paths}}
\end{figure}

\clearpage
\begin{figure}
\centerline{\includegraphics[width=0.95\textwidth]{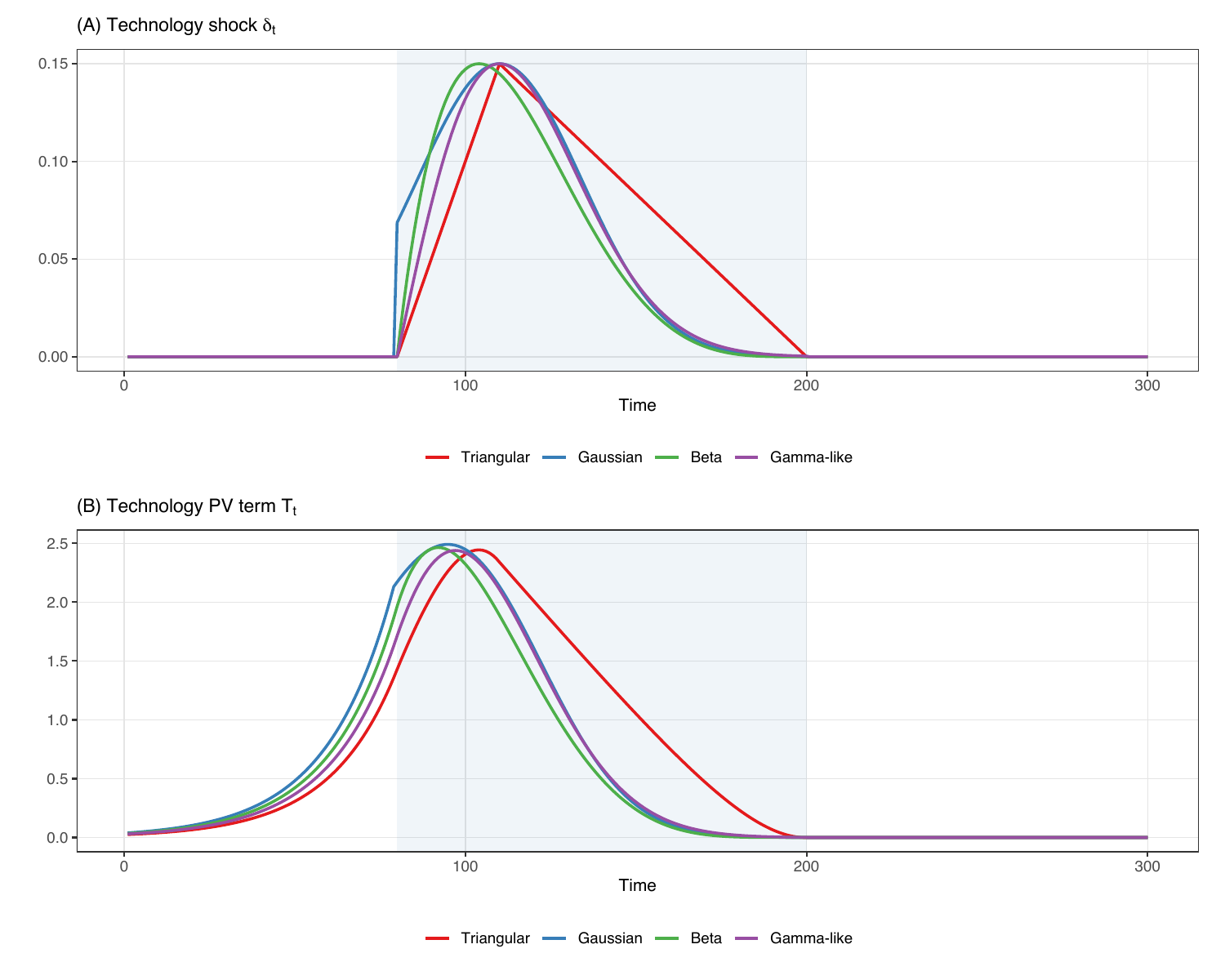}}
\noindent\caption{{\bf Alternative hump-shaped technology profiles differ in timing and smoothness, not in qualitative shape.} The horizontal axis is time over the adoption window from $T_1 = 80$ to $T_2 = 200$, and the vertical axis is the technology shock $\delta_t$. The four plotted profiles share the same technology shock peak $\dmax = 0.15$ and show the alternative shapes used in the robustness simulations.\label{fig:delta_shapes}}
\end{figure}

\clearpage
\begin{figure}
\centerline{\includegraphics[width=0.95\textwidth]{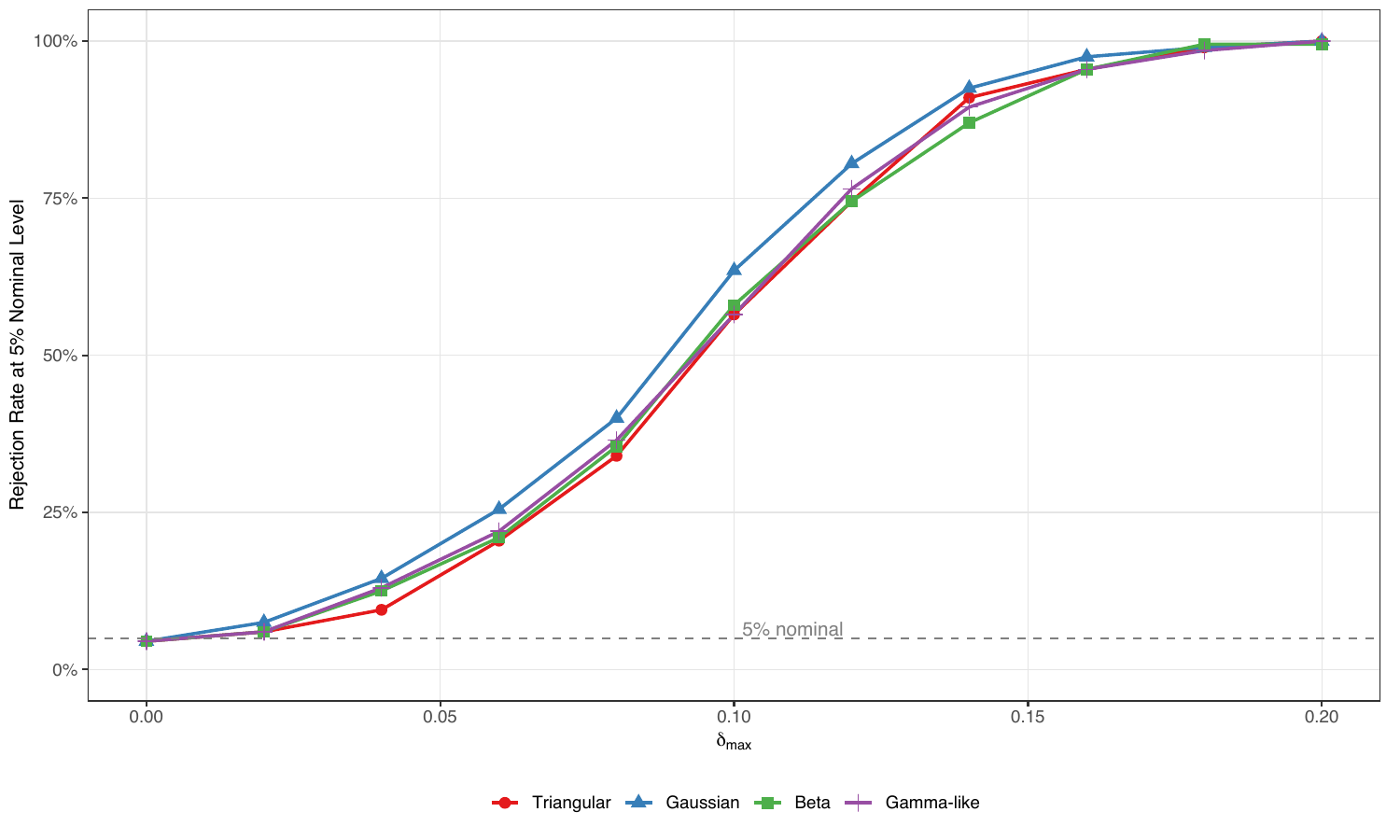}}
\noindent\caption{{\bf Rejection rates rise with the technology shock peak for every shock shape.} The horizontal axis is the technology shock peak $\dmax$, and the vertical axis is the rejection rate at the 5 percent nominal level. The Gaussian profile produces the largest distortion at most values, while the dashed technology-adjusted line stays near the nominal rate regardless of shock shape.\label{fig:mc_shapes}}
\end{figure}

\clearpage
\begin{figure}
\centering
\begin{subfigure}[b]{0.95\textwidth}
\centering
\includegraphics[width=\textwidth]{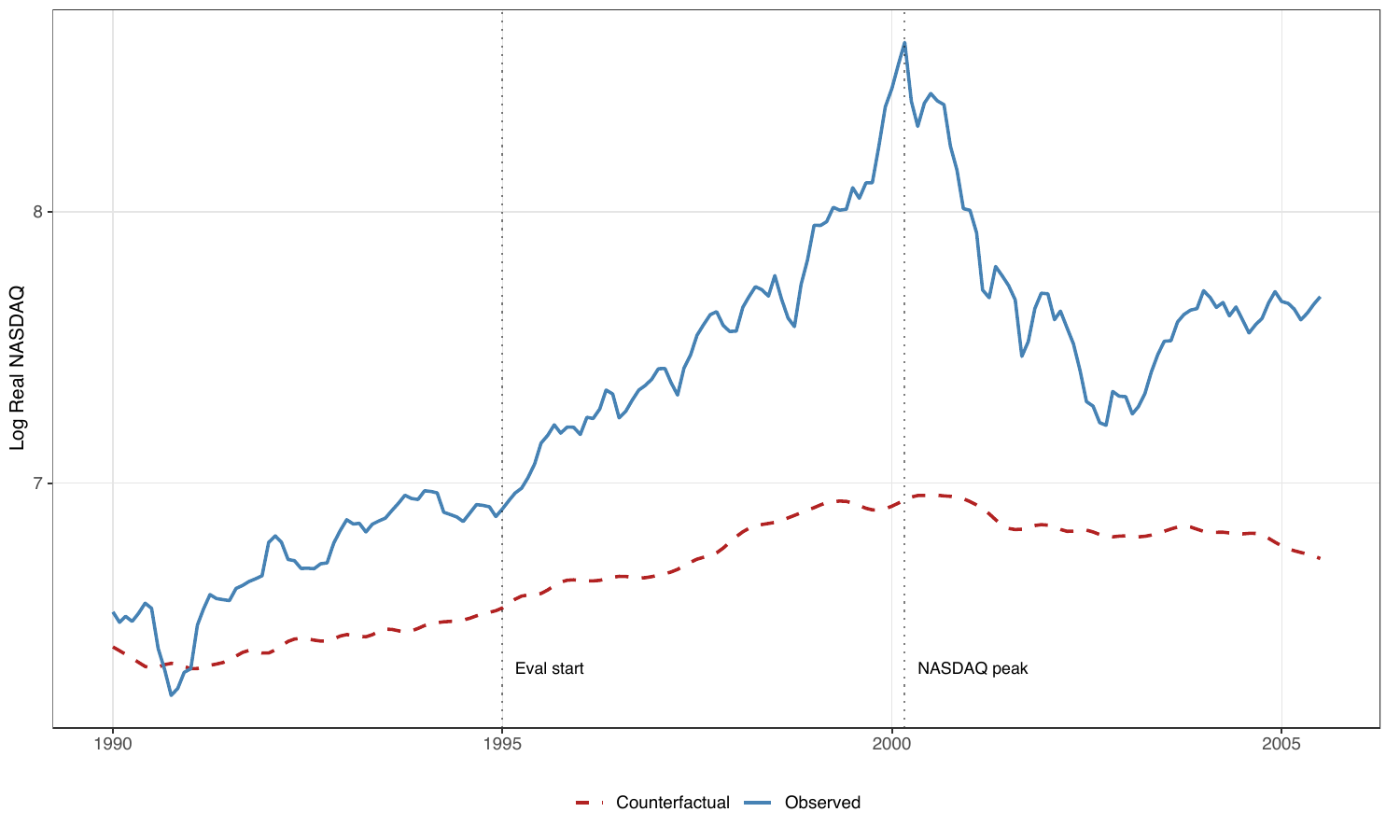}
\caption{Observed vs.\ counterfactual fundamental log NASDAQ price, 1995--2005. Solid line: observed log real NASDAQ Composite. Dashed line: DOLS counterfactual $\hat{f}_t$ estimated from the 1975--1990 training period using log IT investment, TFP, and log patent grants. The counterfactual tracks the observed price through 1997, but the residual gap widens dramatically during 1998--2000.}
\label{fig:counterfactual_empirical}
\end{subfigure}

\vspace{0.2cm}

\begin{subfigure}[b]{0.95\textwidth}
\centering
\includegraphics[width=\textwidth]{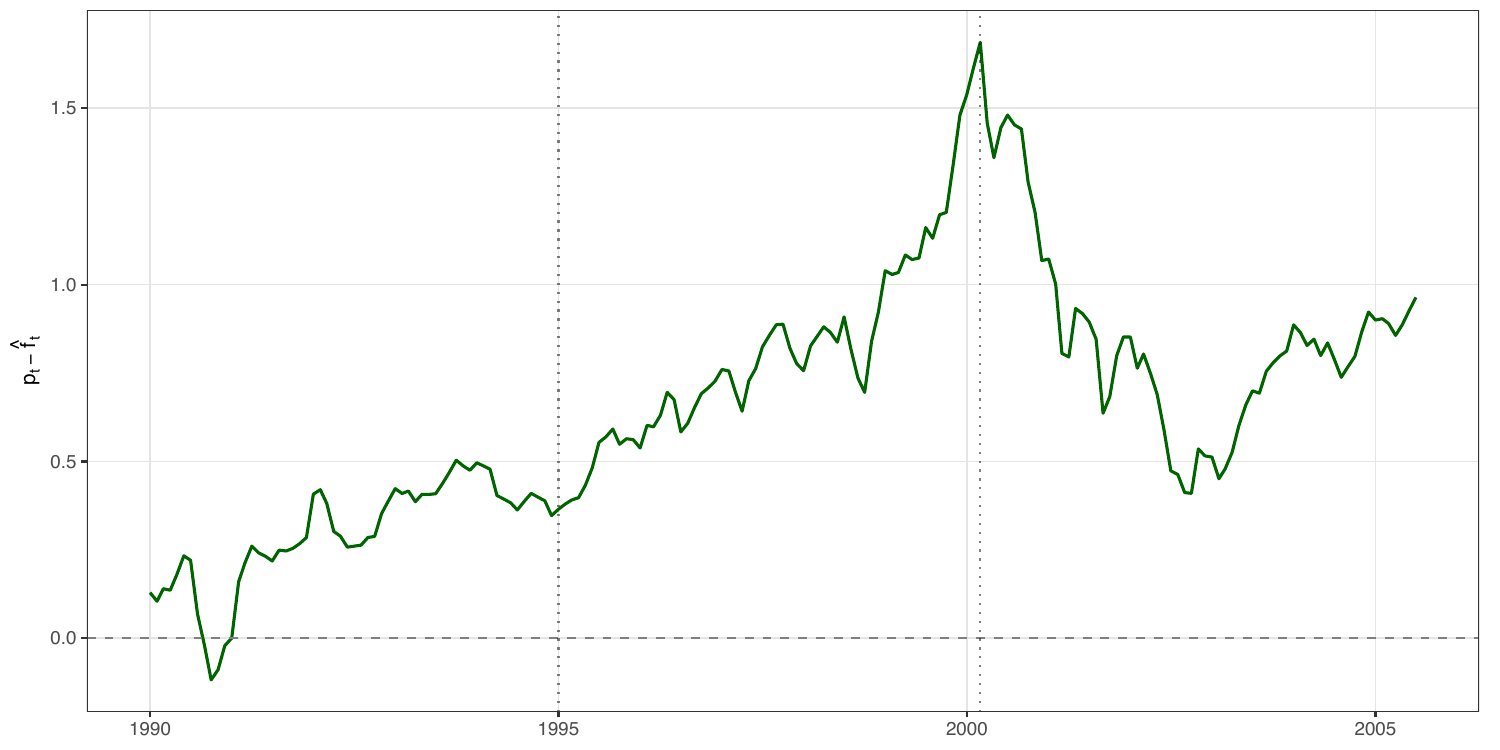}
\caption{Price gap $p_t - \hat{f}_t$ (observed minus counterfactual), 1995--2005. The technology-adjusted series to which PSY is applied. The gap rises sharply during 1998--2000 and narrows after the crash, consistent with a speculative bubble inflating and then collapsing.}
\label{fig:price_gap}
\end{subfigure}
\noindent\caption{{\bf The counterfactual-fundamental view of the dot-com episode isolates a widening residual gap near the peak.} Panel~(a) compares observed and counterfactual prices; Panel~(b) plots the corresponding price gap. Together, the two panels show that residual overvaluation becomes large only late in the run-up.\label{fig:counterfactual_and_gap_dotcom}}
\end{figure}

\clearpage
\begin{landscape}
\begin{table}
\noindent\caption{{\bf Returns do not predict technology proxies during speculative periods.} The table reports bivariate Granger-causality $F$ tests between NASDAQ returns and each technology proxy, estimated separately for speculative and non-speculative regimes. The key result is the absence of return-to-technology predictability during speculative periods, which supports separability. Lag lengths are selected by BIC, speculative periods are January~1997--March~2000 for the dot-com era and January~2023--December~2025 for the AI era, bootstrap $p$-values are based on 2{,}000 wild bootstrap replications, and $^{*}$, $^{**}$, and $^{***}$ denote significance at the 10, 5, and 1 percent levels, respectively.\label{tab:granger_causality}}
\centering
\begin{tabular}{llccccc}
\toprule
\textbf{Direction} & \textbf{Tech.\ var.} & \textbf{Spec.\ $F$} & \textbf{Spec.\ $p$ (std)} & \textbf{Spec.\ $p$ (boot)} & \textbf{Non-Spec.\ $F$} & \textbf{Non-Spec.\ $p$ (std)} \\
\midrule
\multicolumn{7}{l}{\textit{Panel A: Dot-Com Era}} \\
Ret.\ $\to$ Tech.\ & IT Inv. & 0.254 & 0.776 & 0.777 & 2.476 & 0.062$^{*}$ \\
Ret.\ $\to$ Tech.\ & TFP & 0.242 & 0.785 & 0.783 & 1.409 & 0.246 \\
Ret.\ $\to$ Tech.\ & Patents & 0.446 & 0.721 & 0.733 & 0.684 & 0.505 \\
Tech.\ $\to$ Ret.\ & IT Inv. & 0.726 & 0.488 & 0.482 & 1.956 & 0.121 \\
Tech.\ $\to$ Ret.\ & TFP & 0.752 & 0.476 & 0.485 & 0.618 & 0.540 \\
Tech.\ $\to$ Ret.\ & Patents & 1.010 & 0.395 & 0.392 & 0.206 & 0.814 \\
\midrule
\multicolumn{7}{l}{\textit{Panel B: Artificial-Intelligence Era}} \\
Ret.\ $\to$ Tech.\ & IT Inv. & 0.239 & 0.868 & 0.877 & 1.890 & 0.160 \\
Ret.\ $\to$ Tech.\ & TFP & 0.716 & 0.548 & 0.569 & 3.792 & 0.016$^{**}$ \\
Ret.\ $\to$ Tech.\ & Patents & 0.547 & 0.586 & 0.599 & 0.577 & 0.450 \\
Tech.\ $\to$ Ret.\ & IT Inv. & 1.413 & 0.251 & 0.256 & 4.258 & 0.019$^{**}$ \\
Tech.\ $\to$ Ret.\ & TFP & 1.046 & 0.381 & 0.387 & 3.809 & 0.015$^{**}$ \\
Tech.\ $\to$ Ret.\ & Patents & 0.390 & 0.681 & 0.670 & 8.586 & 0.005$^{***}$ \\
\bottomrule
\end{tabular}
\end{table}
\end{landscape}

\clearpage
\begin{table}
\noindent\caption{{\bf Non-technology placebos do not reproduce the technology-adjustment result.} The table reports GSADF statistics for price gaps constructed from placebo counterfactuals that use non-technology covariates alongside the technology baseline. The main takeaway is that generic macro and sentiment variables do not systematically remove explosive signals. Critical values come from 2{,}000 Monte Carlo replications, the technology baseline uses log IT investment, TFP, and log patent grants, $R^2$ is from the DOLS training regression, and the Chicago Board Options Exchange Volatility Index is omitted for the dot-com era because only 12 pre-1990 training observations are available.\label{tab:placebo_tests}}
\centering
\begin{tabular}{llcccccc}
\toprule
\textbf{Era} & \textbf{Covariate} & \textbf{GSADF} & \textbf{CV$_{0.10}$} & \textbf{CV$_{0.05}$} & \textbf{Rej.\ 10\%} & \textbf{$R^2$} & \textbf{Train.\ obs.} \\
\midrule
\multirow{4}{*}{Dot-Com} & Technology (baseline) & 1.952 & 1.918 & 2.175 & Yes & 0.835 & --- \\
& Consumer Sentiment & 2.283 & 1.768 & 2.025 & Yes & 0.308 & 191 \\
& VIX & --- & --- & --- & --- & --- & 12$^\dagger$ \\
& Industrial Production & 1.460 & 1.768 & 2.025 & No & 0.743 & 192 \\
\midrule
\multirow{4}{*}{AI Era} & Technology (baseline) & $-0.275$ & 1.810 & 2.046 & No & 0.907 & --- \\
& Consumer Sentiment & $-0.223$ & 1.334 & 1.623 & No & 0.727 & 168 \\
& VIX & 0.712 & 1.334 & 1.623 & No & 0.410 & 168 \\
& Industrial Production & 0.869 & 1.334 & 1.623 & No & 0.521 & 168 \\
\bottomrule
\end{tabular}
\end{table}

\clearpage
\begin{landscape}
\begin{table}
\noindent\caption{{\bf No common factor in the adjusted Magnificent Seven price gaps is explosive.} The table reports the variance share of each principal component, cumulative variance explained, GSADF statistics, critical values, and rejection decisions. The leading common component explains 60.0 percent of variation but still remains below the rejection threshold. PCA is applied to technology-adjusted price gaps of the Magnificent Seven over 72 common months in 2020--2025, each firm's gap is computed from OLS on firm-specific Compustat fundamentals (real revenue per share and real R\&D expenditure per share) using 2006--2019 training data that exclude bubble months, and PSY uses 2{,}000 Monte Carlo critical values.\label{tab:pca_mag7}}
\centering
\begin{tabular}{lcccccc}
\toprule
\textbf{PC} & \textbf{Var.\ Expl.} & \textbf{Cum.\ Var.} & \textbf{GSADF} & \textbf{CV$_{0.10}$} & \textbf{CV$_{0.05}$} & \textbf{Rej.\ 10\%} \\
\midrule
PC1 & 60.0\% & 60.0\% & 0.844 & 1.334 & 1.623 & No \\
PC2 & 20.6\% & 80.7\% & 0.043 & 1.334 & 1.623 & No \\
PC3 & 11.6\% & 92.3\% & 0.683 & 1.334 & 1.623 & No \\
PC4 & 3.7\% & 96.0\% & $-0.425$ & 1.334 & 1.623 & No \\
PC5 & 2.2\% & 98.1\% & $-0.499$ & 1.334 & 1.623 & No \\
PC6 & 1.1\% & 99.3\% & $-1.218$ & 1.334 & 1.623 & No \\
PC7 & 0.7\% & 100.0\% & $-1.053$ & 1.334 & 1.623 & No \\
\midrule
\multicolumn{7}{l}{\textit{PC1 loadings:} GOOGL $-0.472$; TSLA $-0.437$; NVDA $-0.423$; AMZN $-0.374$; META $-0.371$; AAPL $-0.358$; MSFT $-0.054$} \\
\bottomrule
\end{tabular}
\end{table}
\end{landscape}

\clearpage
\begin{table}
\noindent\caption{{\bf AI-era conclusions are stable across training windows, while dot-com conclusions vary more with sample choice.} The table reports representative GSADF statistics, critical values, DOLS coefficients, and coefficients of determination across alternative training-window end dates. The main takeaway is stable non-rejection in the AI era and greater sensitivity in the dot-com era. Dot-com windows run from January~1985 to December~1994 with evaluation over 1991--2005, AI windows run from January~2010 to December~2019 with evaluation over 2020--2025, and critical values come from 2{,}000 Monte Carlo replications.\label{tab:training_stability}}
\centering
\begin{tabular}{lcccccc}
\toprule
\textbf{Training End} & \textbf{GSADF} & \textbf{CV$_{0.10}$} & \textbf{$\hat{\beta}_{\text{IT}}$} & \textbf{$\hat{\beta}_{\text{TFP}}$} & \textbf{$\hat{\beta}_{\text{Pat}}$} & \textbf{$R^2$} \\
\midrule
\multicolumn{7}{l}{\textit{Panel A: Dot-Com Era}} \\
Jan 1985 & 2.734 & 1.745 & 0.485 & $-0.026$ & $-0.388$ & 0.842 \\
Jan 1987 & 1.849 & 1.745 & 0.430 & 0.001 & 0.497 & 0.789 \\
Jan 1988 & 1.453 & 1.745 & 0.417 & 0.004 & 0.624 & 0.805 \\
Jan 1989 & 1.864 & 1.745 & 0.445 & 0.000 & 0.566 & 0.814 \\
Dec 1990 & 1.952 & 1.745 & 0.456 & $-0.000$ & 0.513 & 0.835 \\
Jan 1992 & 1.792 & 1.745 & 0.436 & 0.002 & 0.492 & 0.842 \\
Jan 1993 & 1.600 & 1.745 & 0.418 & 0.004 & 0.450 & 0.859 \\
Dec 1994 & 1.730 & 1.745 & 0.459 & $-0.004$ & 1.038 & 0.882 \\
\midrule
\multicolumn{7}{l}{\textit{Panel B: AI Era}} \\
Jan 2010 & 0.769 & 1.176 & 2.515 & $-0.002$ & 0.013 & 0.909 \\
Jan 2012 & $-0.131$ & 1.176 & 1.512 & $-0.013$ & 0.688 & 0.872 \\
Jan 2014 & $-0.192$ & 1.176 & 1.315 & $-0.018$ & 0.959 & 0.864 \\
Jan 2016 & 0.108 & 1.176 & 1.366 & 0.000 & 1.095 & 0.849 \\
Jan 2018 & 0.085 & 1.176 & 0.682 & 0.002 & 1.189 & 0.871 \\
Jan 2019 & 0.459 & 1.176 & 1.961 & 0.010 & 1.184 & 0.888 \\
Dec 2019 & 0.377 & 1.176 & 2.141 & 0.012 & 1.157 & 0.907 \\
\bottomrule
\end{tabular}
\end{table}

\clearpage
\begin{table}
\noindent\caption{{\bf Parameter-instability tests reject coefficient constancy in both eras.} The table reports Hansen's $L_c$ statistics, $p$-values, critical values, and sample sizes for the baseline and extended cointegrating regressions. The key takeaway is that DOLS coefficients vary over time in both the dot-com and AI samples. The baseline specification uses log IT investment, TFP, and log patent grants, the extended specification adds the term spread, critical values come from \citet{Hansen1992}, and $^{**}$ and $^{***}$ denote significance at the 5 and 1 percent levels, respectively.\label{tab:hansen_test}}
\centering
\begin{tabular}{llccccc}
\toprule
\textbf{Era} & \textbf{Specification} & \textbf{$L_c$} & \textbf{$p$-Val.} & \textbf{CV$_{0.05}$} & \textbf{CV$_{0.01}$} & \textbf{$n$} \\
\midrule
Dot-Com & Baseline (technology-only) & 1.845 & 0.009$^{***}$ & 0.788 & 1.160 & 174 \\
Dot-Com & Extended (technology and spread) & 1.780 & 0.018$^{**}$ & 0.884 & 1.278 & 174 \\
AI Era & Baseline (technology-only) & 3.994 & ${<}0.001^{***}$ & 0.788 & 1.160 & 168 \\
AI Era & Extended (technology and spread) & 3.941 & ${<}0.001^{***}$ & 0.884 & 1.278 & 168 \\
\bottomrule
\end{tabular}
\end{table}

\clearpage
\begin{table}
\noindent\caption{{\bf All hump-shaped technology shocks produce severe over-rejection in the unadjusted simulations.} The table reports rejection rates at the 5 percent nominal level for detrended log prices across four technology-shock shapes and multiple values of the technology shock peak $\dmax$. The main takeaway is that Gaussian, triangular, Beta, and Gamma-like profiles all generate the same qualitative distortion pattern. Rejection rates are based on the GSADF statistic with $M = 200$ replications, $T = 300$, and $\rho = 0.95$, and the technology-adjusted counterfactual rejection rate is 0.045 for all shapes and all values of $\dmax$ and is therefore omitted from the table body.\label{tab:mc_shapes}}
\centering
\begin{tabular}{ccccc}
\toprule
\textbf{$\dmax$} & \textbf{Triangular} & \textbf{Gaussian} & \textbf{Beta(2,5)} & \textbf{Gamma-Like} \\
\midrule
0.00 & 0.045 & 0.045 & 0.045 & 0.045 \\
0.02 & 0.060 & 0.075 & 0.060 & 0.060 \\
0.04 & 0.095 & 0.145 & 0.125 & 0.130 \\
0.06 & 0.205 & 0.255 & 0.210 & 0.220 \\
0.08 & 0.340 & 0.400 & 0.355 & 0.365 \\
0.10 & 0.565 & 0.635 & 0.580 & 0.565 \\
0.12 & 0.745 & 0.805 & 0.745 & 0.765 \\
0.14 & 0.910 & 0.925 & 0.870 & 0.895 \\
0.16 & 0.955 & 0.975 & 0.955 & 0.955 \\
0.18 & 0.990 & 0.990 & 0.995 & 0.985 \\
0.20 & 1.000 & 1.000 & 0.995 & 1.000 \\
\bottomrule
\end{tabular}
\end{table}

\clearpage
\begin{table}
\noindent\caption{{\bf The empirical design combines technology proxies, financial fundamentals, and discount-rate controls.} The table lists each variable, its economic interpretation, and its data source for the January 1975 to December 2005 sample. It shows how the empirical analysis maps observable series into the counterfactual-fundamental framework. All nominal variables are deflated by CPI, and quarterly variables are interpolated to monthly frequency by cubic spline.\label{tab:data_variables}}
\centering
\begin{tabular}{lp{9.5cm}l}
\toprule
\textbf{Variable} & \textbf{Description} & \textbf{Source} \\
\midrule
\multicolumn{3}{l}{\textit{Target Variable}} \\
$p_t$ & Log NASDAQ Composite index (monthly close) & FRED \\
$d_t$ & Log aggregate dividends, NASDAQ-listed firms & CRSP/WRDS \\
\midrule
\multicolumn{3}{l}{\textit{Group 1: Technology Proxies}} \\
$\mathrm{RD}_t$ & Log real aggregate R\&D expenditure, NASDAQ firms & Compustat/WRDS \\
$\mathrm{PAT}_t$ & Log U.S. Patent and Trademark Office patent grants (IT sector) & USPTO \\
$\mathrm{ITINV}_t$ & Log real private fixed investment in IT equipment & BEA/FRED \\
$\mathrm{TFP}_t$ & Utilization-adjusted TFP (log level) & SF Fed \\
\midrule
\multicolumn{3}{l}{\textit{Group 2: Financial Fundamentals}} \\
$\mathrm{REV}_t$ & Log real aggregate revenue, NASDAQ firms & Compustat/WRDS \\
$\mathrm{BV}_t$ & Log real aggregate book equity, NASDAQ firms & Compustat/WRDS \\
$\mathrm{IP}_t$ & Log industrial production index & FRED \\
\midrule
\multicolumn{3}{l}{\textit{Group 3: Discount-Rate Controls}} \\
$\mathrm{GS10}_t$ & 10-year Treasury yield & FRED \\
$\mathrm{SPREAD}_t$ & BAA minus AAA corporate credit spread & FRED \\
$\mathrm{CPI}_t$ & Log CPI (for deflation) & FRED \\
\bottomrule
\end{tabular}
\end{table}

\clearpage
\begin{table}
\noindent\caption{{\bf The 2006--2019 training regression links NASDAQ valuations strongly to technology fundamentals.} The table reports DOLS coefficient estimates with one lead and one lag of differenced covariates, Newey--West standard errors, and the residual cointegration test. IT investment, patent grants, and TFP all load positively in this later sample. All nominal variables are deflated by CPI, quarterly variables are interpolated to monthly frequency by cubic spline, and the baseline specification uses technology-only covariates.\label{tab:coint_regression_ai}}
\centering
\begin{tabular}{lrrrr}
\toprule
\textbf{Variable} & \textbf{Coef.} & \textbf{Newey--West S.E.} & \textbf{$t$-Stat.} & \textbf{$p$-Val.} \\
\midrule
Intercept & $-20.7103$ & 1.9408 & $-10.671$ & $<0.001$ \\
Log IT investment & $2.1412$ & 0.4089 & $5.236$ & $<0.001$ \\
TFP (utilization-adjusted) & $0.0121$ & 0.0043 & $2.848$ & $0.005$ \\
Log patent grants & $1.1574$ & 0.1202 & $9.628$ & $<0.001$ \\
\midrule
$R^2$ & \multicolumn{4}{l}{0.9065 \quad (Adjusted $R^2 = 0.8992$)} \\
Engle--Granger ADF test statistic & \multicolumn{4}{l}{$-2.6699$ \quad (Rejects at 1\%)} \\
Train.\ obs. & \multicolumn{4}{l}{165 months (Jan 2006--Dec 2019)} \\
\bottomrule
\end{tabular}
\end{table}

\end{document}